\documentclass[%
 reprint,
 superscriptaddress,
 amsmath,amssymb,
 aps,
 pra,
 showkeys,
]{revtex4-2}

\usepackage{graphicx}%
\usepackage{xcolor}
\usepackage{tikz}%
\usepackage{array}[=2016-10-06]%
\usepackage{dcolumn}%
\usepackage{bm}%
\usepackage{hyperref}%
\usepackage{physics} %
\usepackage{amsthm}
\usepackage{xspace}

\newtheorem{theorem}{Theorem}[section]
\newtheorem{corollary}{Corollary}[theorem]
\newtheorem{lemma}[theorem]{Lemma}
\newtheorem{proposition}[theorem]{Proposition}

\newtheorem{definition}{Definition}[section]

\DeclareMathOperator{\EX}{\mathbb{E}}%

\newcommand{\PM}{P\&M\xspace}

\newcommand{\dbit}[1]{\mathrm{#1}}
\newcommand{\dstr}[1]{\mathbf{#1}}
\newcommand{\numd}{n_\text{d}}
\newcommand{\nume}{n_\text{e}}
\newcommand{\nums}{n_\text{s}}
\newcommand{\hmarg}{\tau}
\newcommand{\cthr}{c_\text{thr}}
\newcommand{\aset}{\Omega}
\newcommand{\pknown}{p_\text{k}}
\newcommand{\ECprob}{\epsilon_\text{EC}}
\newcommand{\keydb}{\dstr{k}_\text{db}}
\newcommand{\keysh}{\dstr{k}_\text{sh}}
\newcommand{\num}{\ensuremath{n}}
\newcommand{\pidx}{\ensuremath{m}}
\newcommand{\ebit}{e_b}
\newcommand{\ephase}{e_p}
\newcommand{\keyrate}{\ensuremath{r}}
\newcommand{\Psubstr}[2][]{\dstr{s}_{\text{#2}#1}}

\newcommand{\CSSprob}{\epsilon_\text{thr}}
\newcommand{\Secret}{\epsilon_\text{sec}}
\newcommand{\Correct}{\epsilon_\text{corr}}
\newcommand{\SecretA}{\epsilon_\text{A}}
\newcommand{\SecretB}{\epsilon_\text{B}}
\newcommand{\SecretS}{\epsilon_\text{s}}
\newcommand{\CorrectS}{\epsilon_\text{s}}
\newcommand{\Abortprob}{p_\text{abort}}
\newcommand{\USDprob}{p_\text{USD}}

\newcommand{\pUSDN}[1]{{\USDprob^{(#1)}}}
\newcommand{\pUSDeff}{{\USDprob^{\text{eff}}}}
\newcommand{\pUSDsift}{{\USDprob^{\text{sift}}}}
\newcommand{\pUSDplus}{{\USDprob^{+}}}

\usepackage{letltxmacro}
\LetLtxMacro{\amstextorig}{\text}
\DeclareRobustCommand{\text}[1]{\amstextorig{\normalfont #1}}

\begin{document}

\title{One-Way Quantum Symmetric Private Information Retrieval Protocol From A Single Database Server Using NISQ Devices}

\author{Xiang Zou}
\email{xiang.zou@mail.utoronto.ca}
\thanks{Both authors contributed equally to this work.}
\affiliation{%
 Department of Physics, University of Toronto, 60 St. George Street, Toronto, ON, M5S 1A7, Canada
}

\author{H. F. Chau}
\email{hfchau@hku.hk}
\thanks{Both authors contributed equally to this work.}
\affiliation{%
 Department of Physics, The University of Hong Kong, Pokfulam Road, Hong Kong Island, Hong Kong
}%

\date{\today}

\begin{abstract}
Symmetric private information retrieval (SPIR) is the primitive for an inquirer to obtain the information of one specific bit only without leaking any information on the bit index to the database owner.  It is a useful cryptographic primitive but is very difficult to construct.  With a single server this is impossible for parties with unrestricted quantum power.  We show that it becomes possible, with information-theoretical security, when the server and the user hold only noisy intermediate-scale quantum (NISQ) devices without long-term quantum memory and measure each qubit on arrival, while the eavesdropper remains all-powerful; this is a stricter form of the memory restriction of the bounded- and noisy-quantum-storage models.  The server sends BB84 states one way; the user discriminates them unambiguously and so knows a random subset of the bits.  After a joint random permutation into blocks, the user announces the conclusive locations of one privately chosen block, and the server hashes and takes a random parity of these locations in every block; the permutation makes it exponentially unlikely that one set of locations fits two blocks.  We prove secrecy against the eavesdropper, database privacy through an explicit size rule for the announced set, and user privacy against a server who prepares the prescribed BB84 states and announces truthfully, exactly for single photons and up to a small multi-photon term over a tolerated noise channel.  The only addition to BB84 equipment is a passive linear-optical measurement with one ancillary mode.  For a $10^{4}$-bit database and database privacy $10^{-6}$, blocks of about a thousand qubits suffice over a tolerated noise channel for a user who applies the prescribed measurement; certifying an arbitrary memoryless user and committing to the permutation need larger auxiliary transmissions, which we quantify with finite keys.  Database privacy extends to decoy-state weak coherent pulses, and simulations show that the bound is nearly tight.
\end{abstract}

\maketitle

\section{Introduction}

Secure access to remote databases is a central primitive in both classical and quantum cryptography~\cite{QPQ1, QPQ2, QPQ3, QSPIR_multiple_database1, QSPIR_multiple_database2, Insecurity_Quantum_Oblivious_transfer1}, with applications ranging from cloud storage and biometric identification to distributed ledgers and secure multiparty computation. In this context, private information retrieval (PIR) allows a client to retrieve partial data from a database held by a server without revealing which item was requested to the database holder, thereby protecting the privacy of the client's query. Symmetric private information retrieval (SPIR) strengthens this notion by additionally requiring database privacy.  More precisely, the client must not learn anything about the database beyond the requested item, even if the client is computationally unbounded~\cite{PIR2, QSPIR_theory1, QSPIR_multiple_database1, SPIR1, SPIR2, SPIR3}.

Classical PIR and SPIR have been extensively studied~\cite{PIR1, PIR2, PIR3, SPIR1, SPIR2, SPIR3}. In the information-theoretic setting, it is well known that both single-server PIR and SPIR suffer from strong limitations in classical communication. If only one database server is available and no computational assumptions are made, then any PIR scheme with perfect user privacy must communicate at least as many bits as the database holds, so that downloading the whole database is optimal~\cite{PIR1}, and information-theoretic SPIR with a single server is impossible~\cite{PIR2}.

To obtain nontrivial communication complexity and database privacy, classical SPIR protocols assume either multiple non-communicating database servers that share randomness~\cite{PIR2, SPIR1, SPIR2, SPIR3}, or they rely on computational hardness assumptions~\cite{KO97_PIR, NaorPinkas99}.
These multi-server SPIR protocols can achieve sublinear communication complexity~\cite{PIR2, SPIR3}, optimal download rates~\cite{SPIR1, SPIR2}, and strong security guarantees, but they require architectural assumptions, such as multiple independent data centers and pre-distributed shared randomness, that may be difficult or costly to realize in practice, if not impossible.  Besides, the independent data center assumption may not be justified in the real world.

Quantum communication offers new tools to revisit PIR and SPIR. Genuine quantum protocols can outperform classical ones in certain regimes. For example, quantum symmetrically-private information retrieval (QSPIR) schemes achieve sublinear communication with multiple servers and no shared randomness
while keeping the user's index private and the database private against honest users~\cite{QSPIR_theory1}.
A parallel line of work on quantum private queries (QPQ) has proposed cheat-sensitive protocols~\cite{QPQ1, QPQ2, QPQ3} in which a single quantum server can answer private queries to a classical database with strong user privacy and bounded database leakage per query, at the price of relaxing full information-theoretic guarantees: server misbehavior is detected rather than prevented.
More recently, quantum-key-distribution (QKD) assisted SPIR protocols~\cite{QSPIR_multiple_database1} and an experimental demonstration~\cite{Wang_SPIR_experiment} (see also the commentary in Ref.~\cite{QSPIR_multiple_database2}) have shown that quantum networks can be leveraged to distribute the secret keys and shared randomness needed for multi-database SPIR. Nevertheless, these schemes still fundamentally rely on multiple non-colluding databases or trust in specific network assumptions.

Despite this progress, the status of fully symmetric, information-theoretic PIR from a single classical database remains unclear in the quantum setting. Existing negative results show that, under natural correctness and privacy requirements, perfect single-database SPIR is impossible even when quantum communication is allowed~\cite{Insecurity_Quantum_Oblivious_transfer1}, mirroring the classical impossibility of information-theoretic single-server SPIR.

Is it possible to design a provably secure QSPIR scheme for single database server under a different set of reasonable security assumptions?  We give an affirmative answer to this question in this work.
Specifically, we report a one-way quantum-communication protocol that realizes SPIR between a single client and a single classical database server in an information-theoretic setting with statistical security.  Here our security model is deliberately asymmetric. The eavesdropper Eve is all-powerful, unbounded in quantum computation and quantum storage, and security against her should be made information-theoretic. Alice and Bob, in contrast, are assumed to have only present-day noisy intermediate-scale quantum (NISQ) technology so that long-term quantum storage is not allowed.  We also assume that Bob measures the photons one at a time; the storage restriction alone does not imply this [assumption (A2) of Appendix~\ref{app:usd_certification}].  For user privacy we assume in addition that Alice prepares the prescribed BB84 states and confirms truthfully which of Bob's set guesses were correct, while she may send any number of photons per position [assumption (P) of Section~\ref{sec:Alice_cheating}]; with several photons per position our bound needs a noiseless channel.
Using this security model, our protocol adopts the prepare-and-measure (\PM) QKD-type architecture in which Alice, who holds the database, sends single-photon polarization states drawn from the BB84 alphabet over a one-way quantum channel to Bob, who is interested in part of the database. Classical messages use an authenticated public channel.  The joint-permutation step additionally uses auxiliary BB84 transmissions in the same Alice-to-Bob direction to implement string commitment under the quantum-storage restriction.
Note that our protocol actively exploits the no quantum storage limitation.  As Bob is unable to store photon states, he is forced to measure each one on arrival, before receiving Alice's announcements that come with delay time much longer than the decoherence time of the quantum storage device in Bob's lab.

Restrictions of this kind on the quantum memory of the parties underlie the bounded-~\cite{DFSS_BQSM} and noisy-quantum-storage models~\cite{WST_noisy_storage, Koenig_Wehner_Wullschleger_storage}, in which oblivious transfer and bit commitment are possible with BB84 states and one-way quantum communication.  Our restriction is stricter, and our protocol differs in three respects.  It is a direct $1$-out-of-$\num$ construction rather than a composition of $1$-out-of-$2$ transfers; it keeps an external eavesdropper with unbounded quantum memory away from both the database and the index; and its database privacy rests on a single announced set applied to every block, with parameters and costs given explicitly for finite keys.

Operationally, an attractive feature of our protocol is that it only requires current QKD technology including low intensity laser sources, threshold photon detectors~\cite{Xu_QKD_review} as well as standard post-processing methods such as error correction and privacy amplification.  The one-way nature of the quantum communication, together with the fact that only a single database server is required, makes the architecture conceptually simple and directly compatible with existing quantum-communication links between client and data center.
The only element beyond a QKD receiver is Bob's unambiguous measurement, a three-outcome POVM on a polarization qubit.  It needs neither quantum memory nor any active element.  With one ancillary optical mode, it becomes a projective measurement of a single photon over three modes, which passive linear optics and three detectors implement.  Measurements of this kind have been demonstrated~\cite{Huttner_USD_experiment, Clarke_USD_experiment}.%

This work thus provides a concrete candidate for achieving SPIR with one client, one classical database, and one-way quantum communication, together with explicit resource trade-offs between the database size, the block length, the unambiguous discrimination success probability, and the tolerated quantum bit error rate.

The rest of this paper is organized as follows.  Section~\ref{sec:protocol} states our goals and the protocol.  Sections~\ref{sec:security_eve}, \ref{sec:Bob_Cheating} and~\ref{sec:Alice_cheating} analyze its security against Eve, Bob and Alice, respectively.  The central result is Theorem~\ref{thm:success_protocol}, which bounds what Bob learns beyond the one bit he is entitled to.  Section~\ref{sec:numerics} chooses the parameters, computes the key overhead with finite keys, and validates the underlying block statistics by Monte-Carlo simulation.  Section~\ref{sec:decoy_state} extends database privacy to weak coherent sources with decoy states, and Section~\ref{sec:conclusion} discusses the limitations of our analysis.

\section{Our Goals And Protocol}
\label{sec:protocol}

\subsection{Objectives}
\label{subsec:objectives}
We first state our objectives in simple terms.
Let $\Sigma=\{0, 1\}$. Suppose Alice holds a private bit-string database $\dstr{d} = (\dbit{d}_0, \dbit{d}_1, \cdots, \dbit{d}_{\num -1})\in\Sigma^\num$.  Bob wants to obtain the value of the $\pidx$-th bit in this database and to keep the index $\pidx$ private.  At the same time, Alice does not want Bob to know more than one bit of information on her database stored in a single server.  These two mutually untrusted parties can communicate through an authenticated classical public channel plus an insecure one-way quantum channel.  An all-powerful eavesdropper, Eve, who has unbounded quantum computational and storage power, reads every classical message and completely controls the quantum channel.  Alice and Bob, in contrast, are restricted to NISQ devices and, in particular, hold no long-term quantum storage.  Last but not least, we assume that Alice, Bob and Eve do not trust each other.  In fact, similar to the case of quantum bit commitment~\cite{Mayers_no_BC, Lo_Chau_no_BC}, all QSPIR protocols are insecure if any two of them collude.  Hence, we only consider the situation that at most one of them is dishonest.  In other words, whenever we analyze the cheating strategy of one party, the remaining two parties are assumed to follow the protocol honestly.  We adopt this convention throughout the security analysis below, as it greatly simplifies the discussion.

In simple terms, upon completion of this protocol, if all three of the following criteria are met, the protocol is considered successfully implemented.
\begin{enumerate}
 \item Eve obtains negligible information on Alice's database $\dstr{d}$ and Bob's index $\pidx$.
    \item Alice's information on Bob's index $\pidx$ is negligibly small.
    \item Bob correctly gets the value of the bit $\dbit{d}_\pidx$ he wants with high probability.  However, his knowledge on all the other bits in Alice's database $\dstr{d}$ is negligibly small.
\end{enumerate}
More importantly, the above conditions are listed in the order of decreasing priority.  Precisely, the dishonest Alice will not perform any action that tries to obtain information on Bob's index $\pidx$ if such action leaks information to Eve.  Likewise, the dishonest Bob will not perform any action to obtain more than one bit of information of the database if such action leaks information to Eve or Alice.   The beauty of this approach is that the potential rather than the actual presence of the common enemy Eve is sufficient to force Alice and Bob to cooperate in such a way to guarantee information-theoretic security.  Note, however, that our security proofs do not rely on this ordering: Sections~\ref{sec:Bob_Cheating} and~\ref{sec:Alice_cheating} bound a dishonest Bob directly, and a dishonest Alice under assumption (P) of Section~\ref{sec:Alice_cheating}.

For reference, we collect here the assumptions under which our bounds hold.  Eve is unbounded, controls the quantum channel and reads the authenticated classical channel, and at most one of the three parties is dishonest.  A dishonest Bob holds only a classical record once the waiting time of Step~\ref{QSPIR:initial} has elapsed [assumption (A1) of Appendix~\ref{app:usd_certification}] and measures the qubits one at a time, possibly adaptively but never jointly [assumption (A2), which (A1) does not imply]; the joint permutation of Step~\ref{QSPIR:shuffle} is sampled ideally, the error $\epsilon_{\mathrm{rnd}}$ of the real sampler being added to the bounds [assumption (A3)].  A dishonest Alice, for the index-privacy bound, emits at every position one and the same BB84 state in every photon she sends there, with any number of photons, and announces truthfully in Step~\ref{QSPIR:set_disclosure} which set guesses were correct and which positions are retained [assumption (P) of Section~\ref{sec:Alice_cheating}]; with weak coherent pulses she is further assumed to use the prescribed laser source [assumption (L) of Section~\ref{sec:decoy_state}].

We now discuss the underlying idea behind our protocol in Sec.~\ref{subsec:idea}, state our protocol in Sec.~\ref{subsec:protocol}, and write down the formal security definitions of our protocol in Sec.~\ref{subsec:definitions}.

\subsection{Underlying Ideas Of Our Protocol}
\label{subsec:idea}
 Below is our intuition in designing the protocol.  Suppose Alice and Bob are honest.  Further suppose Alice sends Bob a string of qubits each drawn from either the set
\begin{equation}
 S_0 = \{ \ket{0}, \ket{+} \} \enspace \text{or} \enspace S_1 = \{ \ket{1}, \ket{-} \}
 \label{E:S0_S1_def}
\end{equation}
 with equal probability.  Then Bob can correctly determine a portion of the qubits using unambiguous state discrimination (USD)~\cite{Ivanovic_USD, Dieks_USD, Peres_USD} if he could correctly guess which set a state is prepared.  More importantly, he knows which of them are correctly determined.  Note that the qubit preparation procedure by Alice is equivalent to that of the standard BB84 protocol with unbiased choice in preparation bases before.  Besides, after her announcement from which set each state is prepared, the scheme becomes B92-like~\cite{B92}.  Interestingly, Bob's USD procedure can be viewed as an imperfect QKD scheme generating an imperfect shared key between Alice and Bob.  Hence, using this key to perform one-time key pad to Alice's database $\dstr{d}$, Bob will obtain only a part of $\dstr{d}$.  Most importantly, he knows exactly which part he has correctly got.

Surely, we have to fine-tune the above method to ensure that Bob knows only one bit $\dbit{d}_\pidx$ he wants with high probability.  Besides, we have to strengthen it in order to tackle dishonest Alice and Bob.  We do so by grouping the qubits into blocks.  Bob announces the locations of the bits he determined in one block, and Alice applies a short error-correcting hash and a public random parity to those locations of every block, which yields one bit per block.  Besides, we randomize the ordering of the physical qubits used.  It turns out that a carefully designed scheme of this type can prevent Eve from obtaining information on $\dstr{d}$ and $\pidx$.  In addition, the scheme works using realistic apparatus with imperfections and loss.

\subsection{Our Protocol}
\label{subsec:protocol}

To simplify our analysis, we assume that Alice uses an ideal single-photon source up to Section~\ref{sec:decoy_state}, which turns to a phase-randomized Poissonian source with the decoy-state method.  The block statistics and the measurement test below are analyzed for finitely many qubits.  Only the key rate of Eq.~\eqref{eq:SP_rate_split} is first written in its asymptotic form, and Section~\ref{sec:Finite_key_analysis_qkd} replaces it by its finite-key counterpart.
For ease of reference, Table~\ref{tab:notation} summarizes the main symbols used throughout the paper.

\begin{enumerate}

 \item Alice randomly prepares qubits each in one of the four states, $\ket{0}, \ket{1}, \ket{+}, \ket{-}$ with equal probability and sends them to Bob. Bob acknowledges reception of each of them.  For every qubit received, Bob randomly picks either the QKD mode or the data acquisition mode.  For the QKD mode, he randomly chooses either $\sigma_Z \equiv \{ \ket{0},\ket{1}\}$ or $\sigma_X \equiv \{ \ket{\pm} \}$ bases to measure it.   Alice and Bob adopt the convention that the prepared, measured or inferred qubit states $\ket{0}$ or $\ket{-}$ ($\ket{1}$ or $\ket{+}$) correspond to the bit value $0$ ($1$) throughout the protocol.  For the data acquisition mode, he first randomly guesses with equal probability if the state is prepared in $S_0$ or $S_1$ defined in Eq.~\eqref{E:S0_S1_def}.  Then, he uses USD to unambiguously determine the state~\cite{Peres_USD, Barnett_Croke_review}.  More precisely, if he guesses that the received state is prepared from the set $S_0$, then he performs the POVM measurement with elements $(2-\sqrt{2}) \ket{1}\bra{1}$, $(2-\sqrt{2}) \ket{-}\bra{-}$ and $2(\sqrt{2}-1) (\cos \frac{\pi}{8} \ket{0} + \sin\frac{\pi}{8} \ket{1})(\cos\frac{\pi}{8} \bra{0} + \sin\frac{\pi}{8} \bra{1})$.  And for the case of $S_1$, the POVM elements used are $(2-\sqrt{2}) \ket{0}\bra{0}$, $(2-\sqrt{2}) \ket{+}\bra{+}$ and $2(\sqrt{2}-1)(\sin\frac{\pi}{8} \ket{0} - \cos\frac{\pi}{8} \ket{1})(\sin\frac{\pi}{8} \bra{0} - \cos\frac{\pi}{8} \bra{1})$.  He implements it with three detectors (Appendix~\ref{app:USD_POVM}).  If a conclusive detector clicks, he records the conclusive result, whether or not the inconclusive detector clicks too; if both conclusive detectors click, he discards the position as if the qubit were lost.  He counts the positions at which two or more detectors click; if their number $D$ exceeds a threshold $\tau_D$, fixed by the multi-photon probability of the source ($\tau_D=0$ for an ideal single-photon source), he aborts at the end of this step (Section~\ref{sec:decoy_state} adapts the count to weak coherent pulses).  Bob marks down the chosen bases used for every QKD mode qubit and the measurement results of all qubits. This process is repeated until $(4\nume + \nums)$ QKD mode qubits and at least $2(\numd \num {}+ n_\text{t})$ data acquisition mode qubits are obtained.  Of these unsifted data-acquisition positions, $n_\text{t}$ are sacrificed to the measurement test in Step~\ref{QSPIR:usd_test}; Step~\ref{QSPIR:set_disclosure} then retains $\numd\num$ non-test positions with correct set guesses, aborting if too few remain.  Alice waits for a time much longer than the coherence time of any NISQ quantum memory device used by Bob before proceeding to the next step.%
  \label{QSPIR:initial}

 \item Bob tells Alice which qubits are measured in QKD and data acquisition mode, respectively.  Alice and Bob each pick $2\nume$ of the QKD mode qubits at random, Bob first, from all of them, and Alice from the rest, so that about $\nume$ of each sample survive sifting.  For Alice's sample, Bob publicly announces the corresponding measurement bases used and outcomes obtained; for Bob's sample, Alice publicly announces the bases and bits she prepared.  Each party compares the announced data with its own record on its own sample, which it never reveals, estimates the bit flip $\ebit$ and phase flip $\ephase$ error rates of the quantum channel used, and announces the estimates.  They abort if the two estimates of either rate differ by more than a pre-agreed tolerance, and otherwise take the larger of the two.  Neither party can lower the rates so obtained: on each sample one party announces its data without seeing the other's, and the larger estimate is used.  In particular, Alice can neither keep a tampered position out of Bob's sample, which he picks first, from all QKD mode qubits, after the transmission, nor fit her bits to his outcomes, so a preparation that departs from the BB84 states raises Bob's estimate of the error rate of its basis (Section~\ref{sec:Alice_cheating}).  From the larger rates, they calculate the provable lower bound of the QKD key rate~\cite{Shor_Preskill_Proof1, Scarani_QKD_review}
\begin{equation}
    \keyrate = 1-H_2(\ebit)-H_2(\ephase),
    \label{eq:SP_rate_split}
\end{equation}
  where $H_2(\cdot)$ is the binary entropy function~\cite{Wilde_QIT_book}.  Alice and Bob proceed to the next step only if $\keyrate$ is higher than a pre-agreed (positive) threshold.
  \label{QSPIR:rate_estimation}

 \item For the remaining $\nums$ QKD mode qubits, Alice and Bob perform the standard BB84 QKD protocol.  That is, they publicly announce the bases, discard the mismatched ones to obtain a raw key, and perform classical error correction and privacy amplification using either forward or backward reconciliation to obtain a shared private key~\cite{BB84, Rates_for_realistic_code2, Privacy_Amplification, Scarani_QKD_review}.
  They confirm that their keys agree, as in any QKD session.  The key is split into an $\num$-bit pad $\keydb$ for the database and a $\lceil\log_2\num\rceil$-bit pad $\keysh$ for the shift, used in Steps~\ref{QSPIR:final_step} and~\ref{QSPIR:shift_idx}.  The key length is therefore
  \begin{equation}
   \ell=\num+\lceil\log_2\num\rceil .
   \label{eq:bnums_nums}
  \end{equation}
  Alice prepares the four states uniformly.  So whatever basis Bob measures in, a QKD-mode qubit survives sifting with probability $\tfrac12$.  The raw key thus has about $\nums/2$ bits, and asymptotically
  \begin{equation}
   \nums=2\ell/\keyrate .
   \label{eq:nums_sifting}
  \end{equation}
  \label{QSPIR:standard_secure_bits}

 \item Bob publicly announces the set $S_0$ or $S_1$ he guessed for every data-acquisition position.  These announcements are fixed before Alice selects the test sample.  Before revealing which set guesses are correct, Alice selects uniformly at random a subset $\mathcal{T}$ of $n_\text{t}$ positions from all data-acquisition positions, without sifting by guessed-set correctness.  For each $i\in\mathcal{T}$, Bob announces whether his recorded outcome was conclusive and, if so, the state he concluded, using his previously announced set.  Bob makes these declarations before Alice reveals any prepared states or correct-set information for the data-acquisition positions.  %
  For each prepared state $s$ and declaration $d$, Alice divides the number of tested positions with preparation $s$ and declaration $d$, correct and incorrect set guesses alike, by $n_\text{t}/4$, the expected number of tested positions with preparation $s$.  The result $\hat f(s,d)$ estimates $\Pr[d\mid s]$.  She aborts if any of the $24$ values $\hat f(s,d)$ deviates from the honest entry $T_{s,d}$ of Table~\ref{tab:usd_test} by more than $\delta_\text{test}$.  Here $24=4\times6$: there are four prepared states and six declarations, as each of the two guessed sets comes with one of its two states or with the inconclusive outcome.  This checks whether Bob's declarations are consistent with the prescribed measurement statistics.  The positions in $\mathcal{T}$ are then discarded; they never enter a block, so the conclusive pattern disclosed in this test concerns only discarded positions.
  \label{QSPIR:usd_test}

 \item If the measurement test is passed, Alice informs Bob which of his previously announced set guesses were correct among the remaining, non-test data-acquisition positions.  They abort if fewer than $\numd\num$ such positions remain.  If more remain, Alice selects uniformly at random $\numd\num$ of them and announces which ones to retain.  They keep only the records of these \PM qubits.
  Denote the probability that Bob can unambiguously determine the state by $\USDprob$.  (In the ideal situation, $\USDprob = 1-|\bra{0}\ket{+}| = 1-|\bra{1}\ket{-}|  = 1 - 1/\sqrt{2} \approx 0.293$~\cite{Ivanovic_USD, Dieks_USD, Peres_USD}.) If Bob can unambiguously determine the state, he marks it as ``unambiguous''.  He also jots down the bit value according to the convention in Step~\ref{QSPIR:initial}.  Otherwise, he marks it as ``ambiguous'' and records $0$ or $1$ at random.
  Likewise, Alice constructs her $n_\text{d} \num$ bit string using the same convention.  In this way, Alice and Bob share an imperfect bit-string of length $n_\text{d} \num$. Denote the bit string held by Alice as $\Psubstr{A}$, and that by Bob as $\Psubstr{B}$.
  \label{QSPIR:set_disclosure}

 \item Alice and Bob generate the blocking permutation using the storage-model string commitment specified in Appendix~\ref{app:storage_commitment}~\cite{Koenig_Wehner_Wullschleger_storage}.  The retained-position list of Step~\ref{QSPIR:set_disclosure} is fixed before this procedure begins.  Set $N=\numd\num$, $M_P=N!$, and $\ell_P=\lceil\log_2 M_P\rceil$.

 Alice privately samples $t_a$ uniformly from $\{0,\ldots,M_P-1\}$ and commits to its $\ell_P$-bit encoding.  The commitment uses a fresh auxiliary BB84 transmission from Alice to Bob, a storage-enforcing delay, and classical processing; these auxiliary records are separate from the data and QKD-key records.  Only after the commitment is complete does Bob freshly sample and announce a uniform $t_b\in\{0,\ldots,M_P-1\}$.  Alice then opens $t_a$, and Bob verifies the opening.  An invalid contribution, failed verification, or missing message aborts the session.

 Both parties compute
 \begin{equation}
     P=\operatorname{unrank}_{N}\!\left((t_a+t_b)\bmod M_P\right),
     \label{eq:joint_permutation_rank}
 \end{equation}
 where $\operatorname{unrank}_{N}$ is a fixed bijection from $\{0,\ldots,N!-1\}$ to the symmetric group $S(N)$ of permutations of $N$ positions.  Each honest rank is sampled exactly, for example by local rejection sampling.  Thus no pseudorandom expansion is used.  The ideal procedure samples a uniform permutation with the possibility of abort; the real procedure has statistical error $\epsilon_{\mathrm{rnd}}$ as specified in Appendix~\ref{app:storage_commitment}.  Uniformity is not asserted after conditioning on adversarially chosen completion.

 Alice uses $P$ to permute $\Psubstr{A}$ into $\Psubstr[0]{A}\cdots\Psubstr[(\num-1)]{A}$, each substring having length $\numd$.  Bob applies the same permutation to $\Psubstr{B}$ to obtain $\Psubstr[0]{B}\cdots\Psubstr[(\num-1)]{B}$.
  \label{QSPIR:shuffle}

 \item Alice and Bob fix the size of the announced set and the number of check symbols.  The set has $\cthr=\lfloor(\USDprob-\CSSprob)\numd\rfloor$ locations, where the slack $\CSSprob>0$ keeps $\cthr$ below the number of conclusive positions that a block typically carries.  Channel noise creates conclusive-but-erroneous bits at locations unknown to Bob.  Write $\tilde{\ebit}$ for their rate, fixed by the channel error rates estimated in Step~\ref{QSPIR:rate_estimation} (see Appendix~\ref{sec:Depolarizing_and_Dephasing_Noise}), $\gamma\ge0$ for a decoding margin, $t=\lfloor(\tilde{\ebit}+\gamma)\cthr\rfloor$ for the decoding radius, and $\ECprob$ for the tolerated decoding-failure probability.  They set the number of check symbols to
  \begin{equation}
      \hmarg=\Big\lceil \cthr H_2(\tilde{\ebit}+\gamma)+\log_2(1/\ECprob)\Big\rceil ,
      \label{eq:syndrome_length}
  \end{equation}
  and Alice proceeds only if $\cthr$ obeys the size rule of Eq.~\eqref{eq:size_rule}.  Lemma~\ref{lem:decodability} equally allows the smaller $\hmarg=\lceil\log_2\mathrm{Vol}(\cthr,t)+\log_2(1/\ECprob)\rceil$, with $\mathrm{Vol}(c,t)=\sum_{i\le t}\binom{c}{i}$, which the exact evaluations of Section~\ref{sec:numerics} use.  Note that $\hmarg$ pays for errors only; over a noiseless channel $\gamma=0$ and $\hmarg=\lceil\log_2(1/\ECprob)\rceil$.
  \label{QSPIR:code_sizing}

  \item Bob draws a block index $i$ uniformly at random, with fresh private randomness and independently of his measurement record.  If $\Psubstr[i]{B}$ carries fewer than $\cthr$ conclusive positions, he aborts without naming $i$.
  \label{QSPIR:block_draw}

  \item Otherwise, Bob draws a uniformly random set $\aset\subset\{1,\ldots,\numd\}$ of exactly $\cthr$ of the conclusive positions of $\Psubstr[i]{B}$.  He announces $\aset$ in the clear and the shift $\pidx-i$ one-time-padded with $\keysh$, and both parties cyclically shift all substrings by $\pidx-i$.  That is, they perform the operations $\Psubstr[j]{A} \to \Psubstr[(\pidx-i+j)]{A}$ and $\Psubstr[j]{B} \to \Psubstr[(\pidx-i+j)]{B}$ for all $j$.  Here the arithmetic is performed modulo $\num$.
   \label{QSPIR:shift_idx}

  \item Alice aborts unless $\aset$ consists of exactly $\cthr$ distinct locations.  Only after receiving $\aset$, she draws a uniformly random seed $\dstr{z}\in\{0,1\}^{\cthr}$ and a uniformly random check matrix $\dstr{H}\in\{0,1\}^{\hmarg\times\cthr}$.  Write $\Psubstr[j]{A}|_{\aset}$ for the $\cthr$ bits of block $j$ at the locations $\aset$; the same set is applied to every block.  Alice announces $\dstr{H}$, $\dstr{z}$ and, for every block $j$, the check symbols $\sigma_j=\dstr{H}\,\Psubstr[j]{A}|_{\aset}$, all in the clear.  Being common to all blocks, they reveal nothing about $\pidx$.  Alice takes the parity $\dstr{z}\cdot\Psubstr[j]{A}|_{\aset}$ of each block to obtain an $\num$-bit parity key $\Psubstr{s}$, and sends $\dstr{d}\oplus\Psubstr{s}\oplus\keydb$ to Bob.  Bob holds all of $\Psubstr[\pidx]{B}|_{\aset}$.  He corrects it with $\sigma_\pidx$ (Lemma~\ref{lem:decodability}), outputs $\perp$ if no unique string lies within the radius $t$, and otherwise computes the parity along $\dstr{z}$ and removes it together with the $\pidx$th bit of $\keydb$ from the $\pidx$th received bit to obtain $\dbit{d}_\pidx$.  Bob never signals $\perp$ to Alice, for she could otherwise corrupt the check symbols of chosen blocks and learn from his failure where $\pidx$ lies.
   \label{QSPIR:final_step}
\end{enumerate}

\begin{table}[!h]
\caption{\label{tab:usd_test}Honest declaration frequencies $\Pr[d\mid s]$ of the measurement of Step~\ref{QSPIR:initial} over a noiseless channel, before set sifting, checked on the unsifted test positions of Step~\ref{QSPIR:usd_test}.  Rows: prepared state; columns: the set Bob guessed and the state he claims, or ``?''\ if his outcome was inconclusive.  The factor $\tfrac12$ in each entry is Bob's uniform set guess.  The upper table gives the exact values and the lower one their decimal form; each entry is compared with the announced frequency to within $\delta_\text{test}$.  The derivation of the entries and the noisy case are given in Appendix~\ref{app:usd_certification}.}
\begin{ruledtabular}
\begin{tabular}{lcccccc}
 & \multicolumn{3}{c}{guessed $S_0$} & \multicolumn{3}{c}{guessed $S_1$}\\
prepared & $\ket0$ & $\ket+$ & ? & $\ket1$ & $\ket-$ & ?\\
\colrule
$\ket0$ & $\tfrac{2-\sqrt2}{4}$ & $0$ & $\tfrac{\sqrt2}{4}$ & $\tfrac{2-\sqrt2}{4}$ & $\tfrac{2-\sqrt2}{2}$ & $\tfrac{3\sqrt2-4}{4}$\\
$\ket+$ & $0$ & $\tfrac{2-\sqrt2}{4}$ & $\tfrac{\sqrt2}{4}$ & $\tfrac{2-\sqrt2}{2}$ & $\tfrac{2-\sqrt2}{4}$ & $\tfrac{3\sqrt2-4}{4}$\\
$\ket1$ & $\tfrac{2-\sqrt2}{4}$ & $\tfrac{2-\sqrt2}{2}$ & $\tfrac{3\sqrt2-4}{4}$ & $\tfrac{2-\sqrt2}{4}$ & $0$ & $\tfrac{\sqrt2}{4}$\\
$\ket-$ & $\tfrac{2-\sqrt2}{2}$ & $\tfrac{2-\sqrt2}{4}$ & $\tfrac{3\sqrt2-4}{4}$ & $0$ & $\tfrac{2-\sqrt2}{4}$ & $\tfrac{\sqrt2}{4}$\\
\end{tabular}
\end{ruledtabular}

\vspace{6pt}

\begin{ruledtabular}
\begin{tabular}{lcccccc}
 & \multicolumn{3}{c}{guessed $S_0$} & \multicolumn{3}{c}{guessed $S_1$}\\
prepared & $\ket0$ & $\ket+$ & ? & $\ket1$ & $\ket-$ & ?\\
\colrule
$\ket0$ & $0.1464$ & $0$ & $0.3536$ & $0.1464$ & $0.2929$ & $0.0607$\\
$\ket+$ & $0$ & $0.1464$ & $0.3536$ & $0.2929$ & $0.1464$ & $0.0607$\\
$\ket1$ & $0.1464$ & $0.2929$ & $0.0607$ & $0.1464$ & $0$ & $0.3536$\\
$\ket-$ & $0.2929$ & $0.1464$ & $0.0607$ & $0$ & $0.1464$ & $0.3536$\\
\end{tabular}
\end{ruledtabular}
\end{table}

The string commitment is an additional cryptographic resource.  It preserves one-way quantum communication and requires no honest quantum memory.  Its auxiliary qubits and waiting times are not included in the QKD-key overhead $\nums$; Appendix~\ref{app:storage_commitment} accounts for them separately.

For the string commitment of Step~\ref{QSPIR:shuffle}, we make the storage restriction explicit: after each prescribed delay $\Delta t_{\mathrm{com}}$, a dishonest participant retains at most $q_{\mathrm{com}}$ qubits across all private registers, with unrestricted classical storage.  The no-quantum-memory case is $q_{\mathrm{com}}=0$.  The auxiliary block length and security parameters must satisfy the storage-model bounds in Appendix~\ref{app:storage_commitment}; a noisy-memory implementation instead requires a specified storage channel.  This restriction applies to the participants, not to Eve.

We remark that the above protocol is very similar to the classic BB84 QKD scheme, except for the following crucial modifications.
\begin{itemize}
 \item In Step~\ref{QSPIR:initial}, after Alice sends out the states, she waits sufficiently long to ensure that Bob's quantum memory expires. This forecloses the cheating strategy in which Bob delays his measurement.

    \item By standard BB84 security analysis~\cite{Shor_Preskill_Proof1, Lo_Chau_proof, Scarani_QKD_review, Xu_QKD_review}, Steps~\ref{QSPIR:rate_estimation} and~\ref{QSPIR:standard_secure_bits} give a key about which Eve has negligible information, asymptotically and with finite keys (Section~\ref{sec:Finite_key_analysis_qkd}).  Note that Bob picks his sample first, from all QKD mode qubits, and Alice hers from the rest; on each sample only the other party announces its data, and the larger of the two estimates is used.  This step tests the channel, and it also catches a large departure of Alice's preparation from the BB84 states (Section~\ref{sec:Alice_cheating}); Bob's measurement in the data acquisition mode is tested separately, in Step~\ref{QSPIR:usd_test}.

    \item For each of the $\numd n$ states kept in the data acquisition mode, it becomes a variant of the B92 scheme~\cite{B92}. Indeed, once his set guess is confirmed, Bob faces two non-orthogonal states, keeps the conclusive outcomes, and tells Alice where they are, as a B92 receiver does.  The difference is that he does so for one block only, and that Alice applies his announcement to all $\num$ blocks at once and cannot tell to which block it belongs.
\end{itemize}

\subsection{Formal Security Definitions Of Our Protocol}
\label{subsec:definitions}
As our QSPIR scheme is based on a QKD setup, its performance and security can be characterized by the corresponding parameters in QKD supplemented by two additional security parameters, $\SecretB$ and $\SecretA$, to be introduced below.  Recall that in QKD, a key is called $\Delta$-secret if it is $\Delta$-close to a uniformly distributed key that is uncorrelated with the eavesdropper Eve.  That is,
\begin{equation}
 \frac{1}{2}\norm{\rho_{\text{KE}}-\omega_\text{K}\otimes\rho_\text{E}}_1\le\Delta,
\end{equation}
where $\omega_\text{K}$ denotes the maximally mixed state on the key register $\text{K}$, and $\rho_\text{E}$ the marginal state of Eve's system.  Moreover, a protocol is said to be $\Secret$-secret if it outputs $\Delta$-secret keys such that
\begin{equation}
    (1-\Abortprob)\Delta \le\Secret,
\end{equation}
where $\Abortprob$ is the probability of aborting the protocol~\cite{Renner_thesis, Finite_key_analysis1}.
Further recall that a protocol is said to be $\Correct$-correct if the keys held by Alice and Bob differ with probability at most $\Correct$~\cite{Finite_key_analysis1}.  These two concepts carry over directly to any QSPIR protocol.

\begin{definition}
 A QSPIR protocol is said to be $\Secret$-secret against Eve and $\Correct$-correct if the shared key of Step~\ref{QSPIR:standard_secure_bits}, which pads the shift in Step~\ref{QSPIR:shift_idx} and the database in Step~\ref{QSPIR:final_step}, is $\Secret$-secret to Eve and
 an honest Bob outputs a bit $\dbit{y}_\pidx\neq\dbit{d}_\pidx$ with probability at most $\Correct$ (an abort or $\perp$ is not an output).
 Additionally, our protocol is said to be $\SecretB$-secret against Bob if, whatever Bob does, his final state is within trace distance $\SecretB$ of a state that can be generated from a single database bit $\dbit{d}_j$, with $j$ chosen by the generator, for every database $\dstr{d}$.  For an honest Bob $j=\pidx$, and this says that the remaining database $(\dbit{d}_0,\cdots,\dbit{d}_{\pidx-1},\dbit{d}_{\pidx+1},\cdots,\dbit{d}_{\num-1})$ is $\SecretB$-secret to Bob.  Finally, it is said to be $\SecretA$-secret against Alice if whatever Alice does, her final states for any two indices $\pidx$ and $\pidx'$ are within trace distance $\SecretA$.  We bound $\SecretA$ for an Alice who obeys assumption (P) of Section~\ref{sec:Alice_cheating}.
\end{definition}

These security parameters depend on the free parameters of our protocol, $\numd$, $\CSSprob$, $\gamma$, $\ECprob$, $\nume$, $n_\text{t}$, $\delta_\text{test}$ and $\tau_D$, while $\nums$ is then fixed by Eq.~\eqref{eq:nums_sifting}.  For a database of size $\num$ over a channel with bit-flip and phase-flip error rates $\ebit$ and $\ephase$, our protocol attains required levels of $\Correct$, $\Secret$, $\SecretA$ and $\SecretB$ whenever these free parameters can be tuned accordingly.

\begin{table*}[t]

\caption{\label{tab:notation}Main symbols used throughout the paper.}
\begin{ruledtabular}
\begin{tabular}{lp{0.82\textwidth}}
Symbol & Meaning \\
\colrule
$\num$ & Size of Alice's database $\dstr{d}$; also the number of blocks. \\
$\pidx$ & Index of the database bit Bob wishes to retrieve. \\
$\numd$ & Block length, namely, the number of data-acquisition qubits per block. \\
$\nume$ & Size of each party's sifted sample used for channel error estimation; Step~\ref{QSPIR:rate_estimation} sacrifices $4\nume$ QKD-mode qubits. \\
$\nums$ & Number of QKD-mode qubits used to generate the shared key of Step~\ref{QSPIR:standard_secure_bits}, which pads the database and the shift; the protocol's key overhead, Eq.~\eqref{eq:nums_sifting}. \\
$\ell$ & Length of that key, $\ell=\num+\lceil\log_2\num\rceil$ [Eq.~\eqref{eq:bnums_nums}]. \\
$\keydb$, $\keysh$ & Its two parts: the $\num$-bit pad of the database (Step~\ref{QSPIR:final_step}) and the $\lceil\log_2\num\rceil$-bit pad of the shift (Step~\ref{QSPIR:shift_idx}). \\
$N$ & Number of retained data-acquisition positions, $N=\numd\num$ (Step~\ref{QSPIR:shuffle}); in Section~\ref{sec:Finite_key_analysis_qkd}, the sifted raw-key length. \\
$\hmarg$, $\gamma$, $t$ & Number of check symbols published per block [Eq.~\eqref{eq:syndrome_length}], decoding margin, and decoding radius $t=\lfloor(\tilde{e}_b+\gamma)\cthr\rfloor$. \\
$\Psubstr{A},\,\Psubstr{B}$ & Alice's and Bob's data-acquisition bit strings; $\Psubstr[i]{A}$ is the $i$th block. \\
$\ebit,\ \ephase$ & Channel bit-flip and phase-flip error rates. \\
$\keyrate$ & Shor--Preskill key rate, $\keyrate=1-H_2(\ebit)-H_2(\ephase)$. \\
$\USDprob$ & Single-copy USD success probability (ideally $1-1/\sqrt{2}\approx0.293$). \\
$\tilde{e}_b$ & Post-USD error rate, the fraction of Bob's conclusive bits that are wrong (Appendix~\ref{sec:Depolarizing_and_Dephasing_Noise}). \\
$\CSSprob$, $\cthr$ & User-defined slack and size of the announced set, $\cthr=\lfloor(\USDprob-\CSSprob)\numd\rfloor$ (Step~\ref{QSPIR:code_sizing}). \\
$\aset$ & Set of $\cthr$ within-block locations announced by Bob (Step~\ref{QSPIR:shift_idx}); $\Psubstr[j]{A}|_{\aset}$ denotes the bits of block $j$ at these locations. \\
$\dstr{H}$, $\dstr{z}$, $\sigma_j$ & Public $\hmarg\times\cthr$ check matrix, public random seed, and the check symbols $\sigma_j=\dstr{H}\,\Psubstr[j]{A}|_{\aset}$ of block $j$ (Step~\ref{QSPIR:final_step}). \\
$K_j$, $a_j$, $o_{jj'}$ & Locations of block $j$ whose bit Bob knows; number $a_j=|\aset\setminus K_j|$ of announced locations he misses in block $j$; overlap $o_{jj'}=|K_j\cap K_{j'}|$ of two blocks. \\
$x_j$, $W$ & Leakage score $x_j=\tfrac12\min(1,2^{\hmarg-a_j})$ of block $j$ (Lemma~\ref{lem:seeded_parity}) and leakage weight $W=\sum_jx_j-\max_jx_j$ [Eq.~\eqref{eq:epsB_main}]. \\
$\pknown$ & Known-rate: upper bound on the mean probability that Bob knows the bit of a correctly guessed, untested position ($\USDprob$ for the prescribed measurement over a noiseless channel). \\
$H_2(\cdot)$ & Binary entropy function. \\
$\ECprob$ & Tolerated probability that a wrong string matches the check symbols (Lemma~\ref{lem:decodability}). \\
$\epsilon_\perp$ & Budget for the probability that more than $t$ of Bob's conclusive bits on $\aset$ are wrong, the Chernoff term of Lemma~\ref{lem:decodability}; set equal to $\ECprob$ in Section~\ref{sec:operating_window}. \\
$\Secret,\ \Correct$ & Secrecy and correctness parameters of the protocol. \\
$\SecretS$ & Secrecy and correctness parameter of the QKD key of Step~\ref{QSPIR:standard_secure_bits}. \\
$\SecretA$ & Privacy of Bob's index $\pidx$ against Alice. \\
$\SecretB$ & Privacy of the unqueried database bits against Bob. \\
$\Abortprob$ & Probability that the protocol aborts. \\
$\epsilon_\text{abort}$ & Probability that an honest Bob aborts in Step~\ref{QSPIR:block_draw}. \\
$L_{\mathrm{com}},\ q_{\mathrm{com}}$ & Auxiliary qubits per commitment and the adversary's retained-qubit bound. \\
$D$, $\tau_D$, $\epsilon_D$ & Number of data-acquisition positions at which two or more of Bob's detectors click, its abort threshold (Step~\ref{QSPIR:initial}), and the failure probability of this count (Appendix~\ref{app:alice_multiphoton}).  The count stops Alice once her expected number of multi-clicks reaches $B_D$ (Lemma~\ref{lem:multiclick_count}), which gives the term $0.78B_D/\num$ of Eq.~\eqref{eq:epsA_main} (Lemma~\ref{lem:photon_number}). \\
$\epsilon_{\mathrm{rnd}}$ & Statistical error of the joint permutation of Step~\ref{QSPIR:shuffle}, including commitment errors. \\
$n_\text{t}$ & Unsifted data-acquisition positions sacrificed to test Bob's measurement before Alice's set disclosure (Step~\ref{QSPIR:usd_test}). \\
$\delta_\text{test}$ & Per-cell tolerance of that test, which fixes the effective rate $\USDprob^\text{test}$ of Eq.~\eqref{eq:usd_test_slack}. \\
$\USDprob^\text{cert}$, $\epsilon_\text{cert}$ & Certified known-rate of a cheating Bob and the failure probability of its certificate (Appendix~\ref{app:usd_certification}). \\
$\pknown^{\mathrm{WCP}}$, $\epsilon_{\mathrm{PE}}$ & Certified known-rate under weak coherent pulses [Eq.~\eqref{eq:pk_wcp}] and the failure probability of the declaration-frequency intervals (Appendix~\ref{app:decoy_receiver_certificate}). \\
\end{tabular}
\end{ruledtabular}

\end{table*}

\section{Security Against Eve}
\label{sec:security_eve}

Recall the convention of Section~\ref{subsec:objectives}: Eve is the only adversary here, while Alice and Bob follow the protocol.  Everything Eve learns about $\dstr{d}$ and $\pidx$ passes through two messages, the padded database $\dstr{d}\oplus\Psubstr{s}\oplus\keydb$ of Step~\ref{QSPIR:final_step} and the padded shift of Step~\ref{QSPIR:shift_idx}.  Both are one-time pads with the key of Step~\ref{QSPIR:standard_secure_bits}, generated by standard BB84 QKD, whose $\SecretS$-secrecy and $\CorrectS$-correctness against Eve are well established in the asymptotic and the finite-key regime~\cite{Shor_Preskill_Proof1, Lo_Chau_proof, Finite_key_analysis1, Scarani_QKD_review, Xu_QKD_review}.  Since the secrecy definition is composable~\cite{Renner_Koenig_composable, Portmann_Renner_review}, the protocol is $\Secret$-secret against Eve with $\Secret=\SecretS$, whatever Eve learns from the data qubits, the check symbols or the seed.  Note that no other message depends on $\pidx$: the permutation is uniform, and the set $\aset$ of Step~\ref{QSPIR:shift_idx} is drawn from a block that Bob chooses independently of $\pidx$.  Whatever Eve learns about that block is useless to her without the padded shift.

\section{Security Against Bob}
\label{sec:Bob_Cheating}

Can Bob learn more about Alice's database than the single bit he is entitled to?  In this section, we show that the answer is statistically no.  By the convention of Section~\ref{subsec:objectives}, Bob is the only dishonest party here: Alice follows the protocol and Eve does not collude with him.  More precisely, we show that whatever set Bob announces, he learns about the database at most the one bit he is entitled to, up to an explicit $\SecretB$ whose leading term decays exponentially in $\cthr$ (Theorem~\ref{thm:success_protocol}).  If he measures differently, the test of Step~\ref{QSPIR:usd_test} either stops him or certifies a rate for which the same bound holds.

Before turning to a dishonest Bob, note that even an honest Bob may fail to obtain $\dbit{d}_\pidx$.  This happens in exactly three ways, which Theorem~\ref{thm:success_protocol}(i) bounds.
\begin{enumerate}
    \item[(a)] \emph{Abort.}  The block Bob draws in Step~\ref{QSPIR:block_draw} carries fewer than $\cthr$ conclusive positions.  He stops before announcing $\aset$, so the abort reveals nothing about $\pidx$.  Its probability is at most $e^{-2\CSSprob^{2}\numd}$.
    \item[(b)] \emph{Declared failure, $\perp$.}  Bob has announced $\aset$, but channel noise has flipped some of his conclusive bits there.  In Step~\ref{QSPIR:final_step} he finds either no string within distance $t$ of his own that matches $\sigma_\pidx$, because more than $t$ bits are flipped, or more than one, because a wrong string matches by chance.  He knows that he has failed and outputs $\perp$ instead of a bit, and he keeps this to himself.  Its probability is at most $\ECprob+e^{-\cthr D(\tilde{\ebit}+\gamma\Vert\tilde{\ebit})}$, with $D$ the binary relative entropy.
    \item[(c)] \emph{Wrong bit.}  Bob finds a unique match, but it is not Alice's string, or his copy of the key $\keydb$ differs from Alice's.  He then outputs a wrong $\dbit{d}_\pidx$ without noticing.  Its probability is at most $\ECprob+\CorrectS$.
\end{enumerate}
Note that (b) differs from (c) in what Bob knows: in (b) he knows he has no bit, in (c) he does not know his bit is wrong.  We remark that if all agents are honest and the communication channels are noiseless, then all his bits on $\aset$ are correct, so with $\gamma=0$ case (b) cannot occur and (c) reduces to a key mismatch.

We proceed in the following order.  We first describe the test of Step~\ref{QSPIR:usd_test}, which forces Bob to use the prescribed measurement.  We then show that the ambiguous outcomes of that measurement carry no information, so that every position Bob misses is an erasure.  Next we show that the check symbols of Step~\ref{QSPIR:final_step} let Bob decode the block he announced his set from, while the parity of any other block leaks an explicitly bounded amount, and combine the two into the main theorem.

\subsection{Verifying That Bob Measures As Prescribed}
\label{subsec:usd_test}

Every bound of this section takes Bob's per-qubit knowledge to be that of the unambiguous measurement of Step~\ref{QSPIR:initial}.  But what forces him to use it?  Nothing in the quantum phase does.  Note that the bit convention of Step~\ref{QSPIR:initial} places the two states carrying a given bit value at right angles on the Bloch sphere, so the two bit-value mixtures $\rho_0=\tfrac12(\ket0\bra0+\ket-\bra-)$ and $\rho_1=\tfrac12(\ket1\bra1+\ket+\bra+)$ satisfy $\tfrac12\norm{\rho_0-\rho_1}_1=1/\sqrt2$.  A single fixed measurement along $(\sigma_Z-\sigma_X)/\sqrt2$ therefore returns the correct bit of every retained qubit with probability $\tfrac12+\tfrac{1}{2\sqrt2}\approx0.854$, without any quantum memory and without knowing which set was prepared.  Such a Bob holds no ambiguous positions at all, only positions whose bit he guesses correctly with probability $0.854$.  The erasure property on which every bound of this section rests (Section~\ref{subsec:soft_info}) then fails.  More importantly, the no-long-term-memory assumption does not exclude this.

The ordering of Steps~\ref{QSPIR:usd_test} and~\ref{QSPIR:set_disclosure} is essential. Bob first announces all guessed sets, Alice then samples and tests unsifted positions, and only after the test does she disclose which guesses were correct among the non-test positions.  Thus Table~\ref{tab:usd_test} describes the tested population before conditioning on correct set guesses.  The test positions are removed before the $\numd\num$ positions for blocking are selected.

Step~\ref{QSPIR:usd_test} closes this gap by testing the receiver, just as Step~\ref{QSPIR:rate_estimation} tests the channel.  Why does reproducing Table~\ref{tab:usd_test} prevent Bob's cheating measurement?  The four BB84 states span the real qubit operator space, so the four numbers $\Pr[d\mid s]$ determine the POVM element behind a declaration $d$ up to a $\sigma_Y$ component, which positivity excludes when the table matches a rank-one honest element, and the honest elements are rank one and hence admit no finer splitting that could carry extra information.  The zero entries carry most of the weight: an honest conclusive claim comes from an element annihilating the alternative state, whereas a minimum-error outcome is wrong in every direction.

The test has a finite tolerance, and the freedom it leaves is what the size rule of Eq.~\eqref{eq:size_rule} must be evaluated against.  Maximizing Bob's probability of guessing a retained bit over all single-copy measurements whose declarations match Table~\ref{tab:usd_test} to within $\delta_\text{test}$ is a semidefinite program in the POVM elements; restricted to rank-one elements with Bloch vectors in the $XZ$ plane, which suffices here, it becomes a linear program.  Writing its value as the effective rate $\USDprob^\text{test}$ whose record would guess equally well, we find
\begin{equation}
    \USDprob^\text{test}-\USDprob\ \approx\ C\sqrt{\delta_\text{test}},\qquad C\approx1.1,
    \label{eq:usd_test_slack}
\end{equation}
numerically $+0.011$, $+0.037$ and $+0.133$ at $\delta_\text{test}=10^{-4},10^{-3},10^{-2}$ (the constant drifts to $\approx1.3$ at the upper end).  The square root is generic rather than accidental: the zero entries of Table~\ref{tab:usd_test} are quadratic minima under a tilt of the measurement, while Bob's payoff grows linearly in that tilt.  As $\delta_\text{test}\to0$ the excess vanishes and only the prescribed measurement survives.

Appendix~\ref{app:usd_certification} turns this test into a theorem.  For a Bob who measures the qubits individually, the erasure hypothesis of Theorem~\ref{thm:success_protocol} is replaced by a certified rate $\USDprob^\text{cert}$, at an additive cost $\epsilon_\text{cert}$ in $\SecretB$.  The price is the test size: the slack $\USDprob^\text{cert}-\USDprob$ falls only as $n_\text{t}^{-1/4}$.  The appendix assumes that Bob measures the qubits one at a time.

\subsection{Ambiguous Measurement Results Carry No Information On The Bit Value}
\label{subsec:soft_info}
Does Bob gain any partial knowledge on the states that he fails to unambiguously identify?  For the measurement of Step~\ref{QSPIR:initial}, the answer is no.  Suppose Bob guessed $S_0=\{\ket{0},\ket{+}\}$ correctly and obtained the inconclusive outcome.  Its POVM element is proportional to $\ket{\phi_0}\bra{\phi_0}$ with $\ket{\phi_0}=\cos\frac{\pi}{8}\ket{0}+\sin\frac{\pi}{8}\ket{1}$, which lies halfway between the two states of the set:
\begin{equation}
    |\braket{\phi_0}{0}|^2=|\braket{\phi_0}{+}|^2=\cos^2\tfrac{\pi}{8}.
\end{equation}
Both states therefore produce the inconclusive outcome equally often, and by Bayes' theorem
\begin{equation}
\begin{split}
    &P(\text{A sends $\ket{0}$} \mid \text{inconclusive})\\
    =&P(\text{A sends $\ket{+}$} \mid \text{inconclusive})\\
    =&\tfrac{1}{2}.
\end{split}
\end{equation}
The same holds for $S_1$ with $\ket{\phi_1}=\sin\frac{\pi}{8}\ket{0}-\cos\frac{\pi}{8}\ket{1}$.  Thus $p_a$, the larger of Bob's two posterior probabilities given the inconclusive outcome, is $1/2$: every ambiguous position is a genuine erasure, while the success probability is the optimal $\USDprob=1-|\braket{0}{+}|\approx0.293$.  The POVM needs one ancilla mode (Appendix~\ref{app:USD_POVM}) but no quantum memory, since Bob measures each qubit on arrival.

Note that the halfway property survives channel noise that is symmetric under interchanging the $\sigma_Z$ and $\sigma_X$ bases, such as the depolarizing channel (Appendix~\ref{sec:Depolarizing_and_Dephasing_Noise}).  Other noise tilts the posterior slightly; for example, a coherent rotation by $0.05$~rad ($0.1$~rad on the Bloch sphere) keeps $e_b=e_p$ yet gives $p_a\approx0.52$.  More importantly, a cheating Bob need not use this measurement at all; this is what the test of Section~\ref{subsec:usd_test} guards against.

\subsection{Why Bob Announces Locations}
\label{subsec:reconciliation}
We first recall the objects of Steps~\ref{QSPIR:code_sizing}--\ref{QSPIR:final_step}.  After the joint permutation, Alice's data string is divided into $\num$ blocks $\Psubstr[j]{A}$ of $\numd$ bits each, and Bob holds the corresponding blocks $\Psubstr[j]{B}$, of which he knows only the conclusive positions.  He draws one block privately, and announces a set $\aset\subset\{1,\ldots,\numd\}$ of $\cthr=\lfloor(\USDprob-\CSSprob)\numd\rfloor$ conclusive locations of that block, where $\CSSprob>0$ is a small user defined value.  Alice then applies the same $\aset$ to every block: she publishes a random $\hmarg\times\cthr$ binary matrix $\dstr{H}$, the check symbols $\sigma_j=\dstr{H}\,\Psubstr[j]{A}|_{\aset}$ of every block $j$, and a random seed $\dstr{z}\in\{0,1\}^{\cthr}$, and the bit of block $j$ is the parity $\dstr{z}\cdot\Psubstr[j]{A}|_{\aset}$.  The number $\hmarg$ of check symbols is fixed by Eq.~\eqref{eq:syndrome_length}: it pays for the errors among Bob's conclusive bits and is much smaller than $\cthr$.  Throughout, $K_j\subset\{1,\ldots,\numd\}$ denotes the set of locations of block $j$ whose bit Bob knows, $a_j:=|\aset\setminus K_j|$ the number of announced locations he misses in block $j$, and $o_{jj'}:=|K_j\cap K_{j'}|$ the overlap of two blocks.
{

Why does Bob announce locations?  In a block he has not drawn, the set $\aset$ falls on positions of which he knows only a fraction $\USDprob$.  He therefore misses about $(1-\USDprob)\cthr$ of its bits, far more than the $\hmarg$ check symbols can supply.  In his own block he misses none.  The separation between the two is linear in $\numd$.  Two elementary facts carry the whole section.

\begin{lemma}[Decodability]
\label{lem:decodability}
Let the conclusive bits of Bob's block be in error independently with probability at most $\tilde{\ebit}$.  For a uniformly random $\dstr{H}$ sized by Eq.~\eqref{eq:syndrome_length}, Bob recovers $\Psubstr[\pidx]{A}|_{\aset}$ except with probability at most $\ECprob+e^{-\cthr D(\tilde{\ebit}+\gamma\Vert\tilde{\ebit})}$, and outputs a wrong string with probability at most $\ECprob$.  Here $D(x\Vert y)=x\ln\frac{x}{y}+(1-x)\ln\frac{1-x}{1-y}$ is the binary relative entropy.
\end{lemma}
\begin{proof}
The error weight exceeds $t$ with probability at most $e^{-\cthr D(\tilde{\ebit}+\gamma\Vert\tilde{\ebit})}$ (Chernoff), and never if $\tilde{\ebit}=0$.  Otherwise the true error pattern is among the at most $\mathrm{Vol}(\cthr,t)\le2^{\cthr H_2(\tilde{\ebit}+\gamma)}$ patterns of weight at most $t$, where $\mathrm{Vol}(c,t)=\sum_{i\le t}\binom{c}{i}$~\cite{MacWilliams_Sloane} and $\tilde{\ebit}+\gamma\le\tfrac12$.  Any other one differs from it by a nonzero vector, whose check symbols are uniform, so it matches $\sigma_\pidx$ with probability $2^{-\hmarg}$.  A union bound and Eq.~\eqref{eq:syndrome_length} give $\ECprob$ for the event that a wrong pattern of weight at most $t$ matches, whatever the true weight.
\end{proof}

A random $\dstr{H}$ is not required.  The next lemma holds for every $\dstr{H}$, since it uses only the number of its rows and the uniform seed.  Alice may thus use any efficiently decodable code, supplemented by $\lceil\log_2(1/\ECprob)\rceil$ random check rows drawn after $\aset$, which a miscorrected block fails except with probability $\ECprob$.

\begin{lemma}[Leakage of the seeded parity]
\label{lem:seeded_parity}
Fix a block $j$, let $E_j:=\aset\setminus K_j$ be the announced locations that Bob misses, so that $a_j=|E_j|$, and write $\dstr{H}_{E_j}$ for the columns of $\dstr{H}$ indexed by $E_j$.  Let the values of the blocks on the sets $E_j$ be uniform given Bob's entire record $V$, the set $\aset$ included, and independent across blocks (the erasure hypothesis; Section~\ref{subsec:soft_info} shows that the measurement of Step~\ref{QSPIR:initial} satisfies it, and Appendix~\ref{app:usd_certification} restores it, at a certified known-rate, for a Bob who measures differently).  Let $L_j=\dstr{z}\cdot\Psubstr[j]{A}|_{\aset}$ with $\dstr{z}$ uniform and drawn after $\aset$.  Then, with $V$ now including $\sigma_j$, $\dstr{H}$ and $\dstr{z}$, and for every $\dstr{H}$,
\begin{equation}
    \begin{split}
    \tfrac12\big\Vert\rho_{L_jV}-\omega_L\otimes\rho_V\big\Vert_1&=\tfrac12\,2^{\operatorname{rank}\dstr{H}_{E_j}-a_j}\\
    &\le\;x_j:=\tfrac12\min\!\big(1,2^{\hmarg-a_j}\big),
    \end{split}
    \label{eq:seeded_leak}
\end{equation}
where $\omega_L$ is the uniform bit.
\end{lemma}
\begin{proof}
Given $V$, the unknown values $\dstr{x}\in\{0,1\}^{a_j}$ are uniform on the affine space $\{\dstr{x}:\dstr{H}_{E_j}\dstr{x}=\sigma'\}$, where $\sigma'$ is $\sigma_j$ corrected by the known positions, and $L_j=L_j^{0}\oplus\dstr{z}_{E_j}\cdot\dstr{x}$, where the parity $L_j^{0}$ over the known locations is known to Bob.  If $\dstr{z}_{E_j}$ lies in the row space of $\dstr{H}_{E_j}$, then $\dstr{z}_{E_j}\cdot\dstr{x}$ is constant on this space and Bob knows $L_j$.  Otherwise $\dstr{z}_{E_j}\cdot\dstr{x}$ is a nonconstant linear functional on the kernel of $\dstr{H}_{E_j}$, hence balanced, and $L_j$ is uniform and independent of $V$.  The first case has trace distance $\tfrac12$, the second $0$, and over the uniform $\dstr{z}_{E_j}$ the first occurs with probability $2^{\operatorname{rank}\dstr{H}_{E_j}}/2^{a_j}$.  Finally, $\operatorname{rank}\dstr{H}_{E_j}\le\min(\hmarg,a_j)$.  Appendix~\ref{app:LHL_application} restates this in classical--quantum notation.
\end{proof}

Note that the bound is exact and needs no leftover hashing.  The $\hmarg$ check symbols can fill at most $\hmarg$ missed locations, and every further one halves the leak.  We call $x_j$ the leakage score of block $j$.  Channel noise only weakens Bob, for giving him the error pattern of his conclusive bits reduces his view to an erasure view, as in the noiseless case.
}

\subsection{Bob Gains At Most One Bit}
\label{sec:Bob_get_more_info}
\label{subsec:PA_uncorrectable}

Can Bob learn more than the single bit he is entitled to?  Two things could go wrong.  A second block could be decodable, or the blocks he cannot decode could still leak part of their parity bits.  Theorem~\ref{thm:success_protocol}(ii) below covers both at once.  By Lemma~\ref{lem:seeded_parity}, block $j$ leaks its parity bit with trace distance at most its leakage score $x_j$.  Extracting the bit by a public random parity is thus itself a privacy-amplification step, and an exact one; Theorem~\ref{thm:PA_block} of Appendix~\ref{app:LHL_application} restates it in classical--quantum notation.  Grant Bob the block with the largest score, since he is entitled to one bit.  What he gains beyond it is the leakage weight $W=\sum_jx_j-\max_jx_j$ (Proposition~\ref{prop:simulator}).

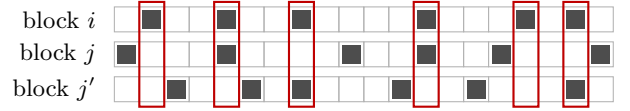
\begin{figure}
    \centering
    \begin{tikzpicture}[x=0.33cm,y=0.33cm]
      \foreach \r/\lab/\pat in {2/{block $i$}/{1,4,7,12,16,18},
                                 1/{block $j$}/{0,4,9,12,15,19},
                                 0/{block $j'$}/{2,5,7,11,14,18}}{
        \node[anchor=east] at (-0.4,\r*1.4+0.5) {\lab};
        \foreach \c in {0,...,19} \draw[gray!60] (\c,\r*1.4) rectangle ++(1,1);
        \foreach \c in \pat \fill[black!70] (\c+0.12,\r*1.4+0.12) rectangle ++(0.76,0.76);
      }
      \foreach \c in {1,4,7,12,16,18} \draw[red!75!black,thick] (\c,-0.2) rectangle ++(1,4.2);
    \end{tikzpicture}
    \caption{\label{fig:common_set_cartoon}Why one announced set cannot serve two blocks.  Each row is a block of $\numd$ positions after the joint permutation of Step~\ref{QSPIR:shuffle}; a filled cell is a location whose bit Bob knows ($K_i$ in the top row, $K_j$ and $K_{j'}$ below).  The red boxes mark the set $\aset$, which Bob draws from the known locations of his block $i$.  Applied to any other block, the same columns hit a known location only inside that block's overlap with $K_i$, which the permutation keeps near $\pknown^{2}\numd$; every other box of the lower rows is one of the $a_j$ (or $a_{j'}$) locations that Bob misses there, and the $\hmarg$ check symbols can fill at most $\hmarg$ of them.}
\end{figure}

Why is $W$ small for every set Bob may announce?  A location of $\aset$ outside $K_j\cap K_{j'}$ is missed in at least one of the two blocks.  Hence $a_j+a_{j'}\ge\cthr-o_{jj'}$, and the worse of the two blocks misses at least half of that.  Now the overlap $o_{jj'}$ does not depend on $\aset$, and the joint permutation of Step~\ref{QSPIR:shuffle} makes it small.  Write $\pknown$ for the known-rate, an upper bound on the mean probability that Bob knows the bit of a correctly guessed, untested position; for the measurement of Step~\ref{QSPIR:initial}, $\pknown=\USDprob$ over a noiseless channel and $\pknown=\USDprob(1+p)$ over a depolarizing channel of strength $p$, and the test of Step~\ref{QSPIR:usd_test} certifies a value for a Bob who measures differently (Section~\ref{subsec:usd_test} and Appendix~\ref{app:usd_certification}).  Since the permutation is independent of Bob's record, a location is known in two given blocks with probability about $\pknown^{2}$.  Two blocks therefore share only about $\pknown^{2}\numd$ known locations, against the $\cthr\approx\USDprob\numd$ that the set must contain.  Figure~\ref{fig:common_set_cartoon} shows the situation: the boxes of $\aset$ fall on known locations of block $i$ by construction, but in every other block they hit known locations only inside that block's overlap with $K_i$.  Note that the permutation is essential here.  Were Bob able to align his known positions across two blocks, the overlap would be $\pknown\numd$ and one set would serve both.  Appendix~\ref{app:simulator_proof} turns this argument into Eq.~\eqref{eq:EW_common} and gives the parameters in closed form.  For $\num=10^{4}$ and $\SecretB=\ECprob=10^{-6}$, a block of $\numd\approx900$ suffices over a noiseless channel and $\numd\approx4000$ at $\tilde{\ebit}=2\%$ (Table~\ref{tab:common_set} and Fig.~\ref{fig:CS_window}).

Note that exact erasure matters.  Were Bob's posterior on an ambiguous position even moderately above $1/2$, the positions he misses would no longer be uniform and Lemma~\ref{lem:seeded_parity} would not apply; this is why Step~\ref{QSPIR:usd_test} tests his measurement (Section~\ref{subsec:usd_test}).

\begin{theorem}[Correctness, database privacy and resources]
\label{thm:success_protocol}
Fix $\num$, $\numd$, $\CSSprob$, $\gamma$ and $\ECprob$, and set the number $\hmarg$ of check symbols by Eq.~\eqref{eq:syndrome_length}.  Let the locations known to Bob carry independent flags whose mean rate over the positions from which the blocks are drawn uniformly, the correctly guessed untested ones, is at most the known-rate $\pknown$.  For the measurement of Step~\ref{QSPIR:initial}, $\pknown=\USDprob$ over a noiseless channel and $\pknown=\USDprob(1+p)$ over a depolarizing channel of strength $p$, where Bob is granted his error pattern.  Let $\epsilon_{\mathrm{rnd}}$ be the statistical error of the joint permutation of Step~\ref{QSPIR:shuffle}, in the sense of Eq.~\eqref{eq:permutation_abort_bound}, commitment errors included (Appendix~\ref{app:storage_commitment}).  Under the erasure hypothesis of Lemma~\ref{lem:seeded_parity}, the protocol of Section~\ref{subsec:protocol}
\begin{enumerate}
    \item[(i)] outputs a wrong bit with probability at most $\ECprob+\CorrectS$, so that it is $\Correct$-correct with $\Correct=\ECprob+\CorrectS$, and $\perp$ with probability at most $\ECprob+e^{-\cthr D(\tilde{\ebit}+\gamma\Vert\tilde{\ebit})}$, and an honest Bob aborts in Step~\ref{QSPIR:block_draw} with probability $\epsilon_\text{abort}\le e^{-2\CSSprob^{2}\numd}$;
    \item[(ii)] is $\SecretB$-secret against Bob with
    \begin{equation}
        \begin{split}
        \SecretB&\;\le\;\EX[W]+\epsilon_{\mathrm{rnd}},\\
        W&:=\sum_j x_j-\max_j x_j,
        \end{split}
        \label{eq:epsB_main}
    \end{equation}
    where $W$ is the leakage weight and, for \emph{every} rule by which Bob chooses $\aset$,
    \begin{equation}
        \EX[W]\;\le\;\frac{\num^{2}}{4}\;2^{-(\cthr-2\hmarg)/2}\,\Big[1+(\sqrt2-1)\,\pknown^{2}\Big]^{\numd};
        \label{eq:EW_common}
    \end{equation}
    for a Bob who measures differently, under assumptions (A1)--(A3) of Appendix~\ref{app:usd_certification}, both displays hold with $\pknown=\USDprob^\text{cert}$ and $\epsilon_\text{cert}$ added to $\SecretB$, where $\USDprob^\text{cert}$ is the known-rate certified by the test of Step~\ref{QSPIR:usd_test} and $\epsilon_\text{cert}$ is the probability that the test passes although Bob's mean known-rate exceeds $\USDprob^\text{cert}$ [Eq.~\eqref{eq:p_cert} and Proposition~\ref{prop:cert_rate}];
    \item[(iii)] is $\Secret$-secret against Eve with $\Secret=\SecretS$ (Section~\ref{sec:security_eve});
    \item[(iv)] consumes $\ell=\num+\lceil\log_2\num\rceil$ bits of QKD key, i.e.\ $\nums=2\ell/\keyrate$ QKD-mode qubits asymptotically [Eqs.~\eqref{eq:bnums_nums} and~\eqref{eq:nums_sifting}] (Section~\ref{sec:Finite_key_analysis_qkd} gives the finite-key count), $\hmarg$ public bits per block, and $\numd$ upstream bits for $\aset$.
\end{enumerate}
\end{theorem}
\begin{proof}
(i) is Lemma~\ref{lem:decodability} together with the correctness of the QKD key; the abort bound is Hoeffding's inequality~\cite{Hoeffding_1963}, for every position of an honest Bob is conclusive with probability at least $\USDprob$.  The first display of (ii) is Lemma~\ref{lem:seeded_parity} together with Proposition~\ref{prop:simulator}; Eq.~\eqref{eq:EW_common} is proved in Appendix~\ref{app:simulator_proof}, and the certified case in Theorem~\ref{thm:cert_privacy} of Appendix~\ref{app:usd_certification}.  (iii) and (iv) are read off Steps~\ref{QSPIR:standard_secure_bits} and~\ref{QSPIR:final_step}.
\end{proof}

Equation~\eqref{eq:EW_common} turns a target $\SecretB$ into a size rule for the announced set (Appendix~\ref{app:common_set}),
\begin{equation}
    \begin{split}
    \cthr\;\ge\;&2\hmarg+2\numd\log_2\!\Big[1+(\sqrt2-1)\,\pknown^{2}\Big]\\
    &+4\log_2\num+2\log_2\frac{1}{\SecretB}-4 ,
    \end{split}
    \label{eq:size_rule}
\end{equation}
where $\SecretB$ stands for the part of the budget left for $\EX[W]$, and whose second term equals $0.1007\,\numd$ at $\pknown=1-1/\sqrt2$.  Note that the check symbols enter twice, once for each of the two blocks a cheating Bob must fit into one set.  Note also that at fixed $\hmarg$ every two further locations halve the bound.  More importantly, the rule is compatible with $\cthr<\USDprob\numd$ only if $\USDprob\big[1-2H_2(\tilde{\ebit}+\gamma)\big]$ exceeds $2\log_2[1+(\sqrt2-1)\pknown^{2}]$.  Over a depolarizing channel, where $\pknown=\USDprob(1+p)$ and $\tilde{\ebit}=p/(1+p)$, this allows $\tilde{\ebit}$ up to $5.5\%$ as $\numd\to\infty$.

Appendix~\ref{sec:Depolarizing_and_Dephasing_Noise} expresses $\tilde{\ebit}$ through the estimates of Step~\ref{QSPIR:rate_estimation}: for a Pauli channel $\tilde{\ebit}=(\ebit+\ephase)/(1+\ebit+\ephase)$, about twice the channel error rate.

The trace-distance parameter $\SecretB$ also controls the Shannon information.  Let $\mathbf{x'}$ be the $\num-1$ database bits outside the one Bob is entitled to, for a uniformly random database, and $\mathbf{y'}$ his estimates of them.  Applying the Alicki--Fannes--Winter continuity bound~\cite{Alicki_Fannes, Winter_continuity} to the $(\num-1)$-bit register $\mathbf{x'}$, as in Corollary~\ref{cor:Shannon_from_trace} of Appendix~\ref{app:LHL_application}, gives $I(\mathbf{x'};\mathbf{y'})\le\epsilon_I$ with
\begin{equation}
    \epsilon_I\;=\;2(n-1)\,\SecretB+(1+\SecretB)\,H_2\!\left(\frac{\SecretB}{1+\SecretB}\right).
    \label{eq:epsI_epsB}
\end{equation}
Since $\SecretB$ is exponentially small in $\cthr$, so is $\epsilon_I=O\big(\SecretB\,[n+\log_2(1/\SecretB)]\big)$.

\section{Security Against Alice}
\label{sec:Alice_cheating}
From the discussions in Sec.~\ref{subsec:objectives}, Alice is now the only dishonest party. Bob follows the protocol honestly, and the quantum channel is controlled by Eve, who does not collude with Alice.  A dishonest Alice controls the data states she sends in Step~\ref{QSPIR:initial}, the auxiliary commitment states of Step~\ref{QSPIR:shuffle}, and her half of the classical conversation.  Bob, in turn, sends only two messages that depend on his private record of the positions that enter the blocks, both in Step~\ref{QSPIR:shift_idx}: the set $\aset$ and the cyclic shift, the latter computed from $\pidx$ itself.  Index privacy therefore hinges on one question: from Alice's point of view, can either message be correlated with $\pidx$?  If Alice prepares the legitimate states, the answer is no, by symmetry.  Every kept qubit is unambiguously determined with one and the same probability $\USDprob$ no matter which of the four BB84 states was sent.  Hence, to Alice the conclusive pattern is an i.i.d.\ coin-flip pattern that she can neither observe nor bias, and all $\num$ substrings look statistically identical.  The block $i$ is drawn uniformly, the set $\aset$ is a uniformly random set of $\cthr$ locations whichever block it comes from, and the announced shift $\pidx-i$ is a one-time-padded copy of $\pidx$.  Against an honestly preparing Alice, index privacy therefore holds exactly, over any channel whose noise does not distinguish the states she sent (Proposition~\ref{prop:index_privacy_honest} of Appendix~\ref{app:alice_audit}).

Can a dishonest Alice break this symmetry?  Apart from her confirmations in Step~\ref{QSPIR:set_disclosure}, her classical messages cannot.  The error rates estimated in Step~\ref{QSPIR:rate_estimation} enter the protocol only through global parameters, the key rate and the number of check symbols of Eq.~\eqref{eq:syndrome_length}, which are common to all $\num$ blocks, as are $\dstr{H}$ and $\dstr{z}$; overstating a symmetric error rate merely inflates the overhead.  Wrongly prepared states are another matter.  They can bias which of Bob's positions are conclusive, and Bob's own record cannot detect this (Appendix~\ref{app:alice_audit}).  A false confirmation in Step~\ref{QSPIR:set_disclosure}, that a wrong set guess was correct, has the same effect without any tampering: $\ket0$ kept under the guess $S_1$ is conclusive with probability $3\USDprob\approx0.88$.  We therefore assume
\begin{itemize}
    \item[(P)] At every position, every photon that Alice emits carries one and the same of the four BB84 states, and Alice announces truthfully in Step~\ref{QSPIR:set_disclosure} which set guesses were correct and which positions are retained.  The number of photons per position is arbitrary.
\end{itemize}
The parameter estimation of Step~\ref{QSPIR:rate_estimation} supports the first half of (P).  There Bob picks his own sample first, from all QKD mode qubits after the transmission, so Alice cannot keep the positions she tampers with out of it, and she announces her bits on it without seeing his outcomes.  A preparation that departs from the BB84 states therefore raises Bob's estimates, the protocol uses the larger of the two parties' estimates, and a large departure makes it abort.  The second half of (P) is not tested, and the estimation does not prove the first half either.  A small rotation biases the conclusive probability at first order in its angle but the error rates only at second order [Eq.~\eqref{eq:tamper_linear_gain}], so it can hide in the statistical fluctuations of the estimates and in the channel noise.  We prove no bound against such an Alice (Appendix~\ref{app:alice_audit}).

Several photons per position are a separate danger, for Bob's detectors cannot tell one photon from several.  A bright pulse of a legitimate state is almost always conclusive for the measurement of Step~\ref{QSPIR:initial}, and (P) does not exclude it.  It does, however, often make two detectors click.  On average, such a pulse changes Bob's conclusive record at its position with probability at most $0.39$ times its multi-click probability, and the count $D$ of Step~\ref{QSPIR:initial} limits the total multi-click probability of Alice's pulses.  Since the block of Step~\ref{QSPIR:block_draw} is drawn uniformly and privately, whatever bias remains acts through one random block only (Appendix~\ref{app:alice_multiphoton}).  Hence, under (P) and over a noiseless channel,
\begin{equation}
    \SecretA\;\le\;\frac{0.78\,B_D}{\num}\;+\;\epsilon_D
    \label{eq:epsA_main}
\end{equation}
(Theorem~\ref{thm:alice_multiphoton} of Appendix~\ref{app:alice_audit}).  The count stops Alice, except with probability $\epsilon_D$, once her expected number of multi-clicks reaches $B_D=\tau_D+\Lambda_D+\sqrt{2\tau_D\Lambda_D+\Lambda_D^{2}}$ with $\Lambda_D=\ln(1/\epsilon_D)$ (Lemma~\ref{lem:multiclick_count}).  The bound is $0.002$ at $\num=10^{4}$, $\tau_D=0$ and $\epsilon_D=10^{-6}$.  With one photon per position, $\SecretA=0$ over any channel whose noise does not distinguish the four states and acts alike at every position (Proposition~\ref{prop:index_privacy_honest}).

The joint generation in Step~\ref{QSPIR:shuffle} keeps a dishonest Bob from choosing the permutation.  A permutation chosen by Bob alone would let him align his conclusive positions across several blocks, so that one set serves all of them.  Alice's rank is bound before honest Bob supplies a fresh uniform rank, so Alice cannot choose the resulting permutation.  Conversely, hiding prevents a dishonest Bob from learning honest Alice's rank before supplying his own contribution.  Appendix~\ref{app:storage_commitment} states the statistical error and the abort qualification; the sampler does not guarantee a uniform distribution among adversarially selected completed runs.

\section{Parameters, Resources And Numerical Validation}
\label{sec:numerics}
In this section, we choose the parameters of the protocol, estimate its cost, and validate numerically the block statistics that underlie the security analysis of Section~\ref{sec:Bob_Cheating}.  The simulations are conducted in Python using NumPy, and the graphs are plotted with Matplotlib.

\subsection{Block Length And Operating Window}
\label{sec:operating_window}

The one design decision of the protocol is the size $\cthr$ of the announced set.  Two constraints pull in opposite directions.  Correctness bounds $\cthr$ from above, Bob can only put locations he knows into $\aset$, so the block he draws must carry at least $\cthr$ conclusive positions.  That count is $\mathrm{Bin}(\numd,\pknown)$ with mean $\pknown\numd\approx0.293\,\numd$, and $\cthr$ must sit below the value that the count exceeds except with the abort probability $\epsilon_\text{abort}$ of Step~\ref{QSPIR:block_draw}.  Database privacy bounds $\cthr$ from below through the size rule of Eq.~\eqref{eq:size_rule}, a smaller set could be fitted by a cheating Bob into the locations he knows in \emph{two} blocks, which share about $\pknown^{2}\numd$ of them, with the $\hmarg$ check symbols making up the difference in each block.  The rule therefore charges $2\hmarg$ plus about $0.1\,\numd$ plus terms logarithmic in $\num$ and $1/\SecretB$.  Both limits grow linearly in $\numd$, the upper one at rate $\USDprob\approx0.29$ and the lower one at rate about $0.10$ plus the check-symbol rate $2\hmarg/\numd$, so a long enough block always leaves room between them, while channel noise narrows it: erroneous conclusive bits raise $\pknown$ slightly, and, more importantly, they require more check symbols, $\hmarg\approx\cthr H_2(\tilde{\ebit}+\gamma)$, which are counted twice.

Figure~\ref{fig:CS_window} evaluates the two constraints exactly, with the binomial law of the overlap $o_{jj'}$ in place of the moment bound of Lemma~\ref{lem:overlap_mgf}.  Panel~(a) fixes $\num=10^{4}$ and $\numd=3000$ and plots both limits against the post-USD error rate $\tilde{\ebit}$, the fraction of Bob's conclusive bits that are wrong ($\tilde{\ebit}=p/(1+p)$ for a depolarizing channel of strength $p$, Appendix~\ref{sec:Depolarizing_and_Dephasing_Noise}).  The blue line is the largest $\cthr$ for which an honest Bob aborts with probability at most $10^{-3}$, about $800$ over a noiseless channel; it rises slowly with noise because noise creates extra conclusive outcomes.  The red line is the smallest $\cthr$ allowed by the size rule at $\SecretB=10^{-6}$, about $430$ over a noiseless channel ($2\hmarg=40$, $0.1007\numd=302$, $4\log_2\num=53$ and $2\log_2(1/\SecretB)=40$, less the constant $4$); it rises steeply with noise through $\hmarg$.  The shaded region between them is the operating window: every $\cthr$ in it gives both correctness and privacy.  The two lines cross at $\tilde{\ebit}\approx1.5\%$, beyond which no set size works at this block length.

Panel~(b) shows the remedy.  Since the window widens with $\numd$, it plots the smallest block length for which it is open, against $\tilde{\ebit}$ and for three values of $\SecretB$: about $900$, $2400$, $4000$ and $7100$ positions at $\tilde{\ebit}=0$, $1\%$, $2\%$ and $3\%$ for $\SecretB=10^{-6}$.  The three curves lie close together because $\SecretB$ enters the size rule only through the additive term $2\log_2(1/\SecretB)$, so the block length grows only logarithmically in $1/\SecretB$.  All of them diverge, beyond the plotted range, near $5.8\%$, where the gain $\pknown[1-2H_2(\tilde{\ebit}+\gamma)]$ per location no longer exceeds the cost $2\log_2[1+(\sqrt2-1)\pknown^{2}]$ and no block length suffices; the figure lets $\cthr$ reach the honest count $\pknown\numd$, and for $\cthr<\USDprob\numd$ this is the threshold of about $5.5\%$ of Section~\ref{sec:Bob_get_more_info}.  To use the figure, one reads the block length off panel~(b) for the noise level at hand and then takes any $\cthr$ inside the window of panel~(a), recomputed at that block length; Appendix~\ref{app:simulator_proof} gives the same choice in closed form, and Table~\ref{tab:common_set} lists it next to the exact evaluation.  Figure~\ref{fig:USD_number_per_block_histogram} below shows the conclusive counts of simulated blocks against such a $\cthr$.

\begin{figure*}
    \centering
    \includegraphics[width=0.9\textwidth]{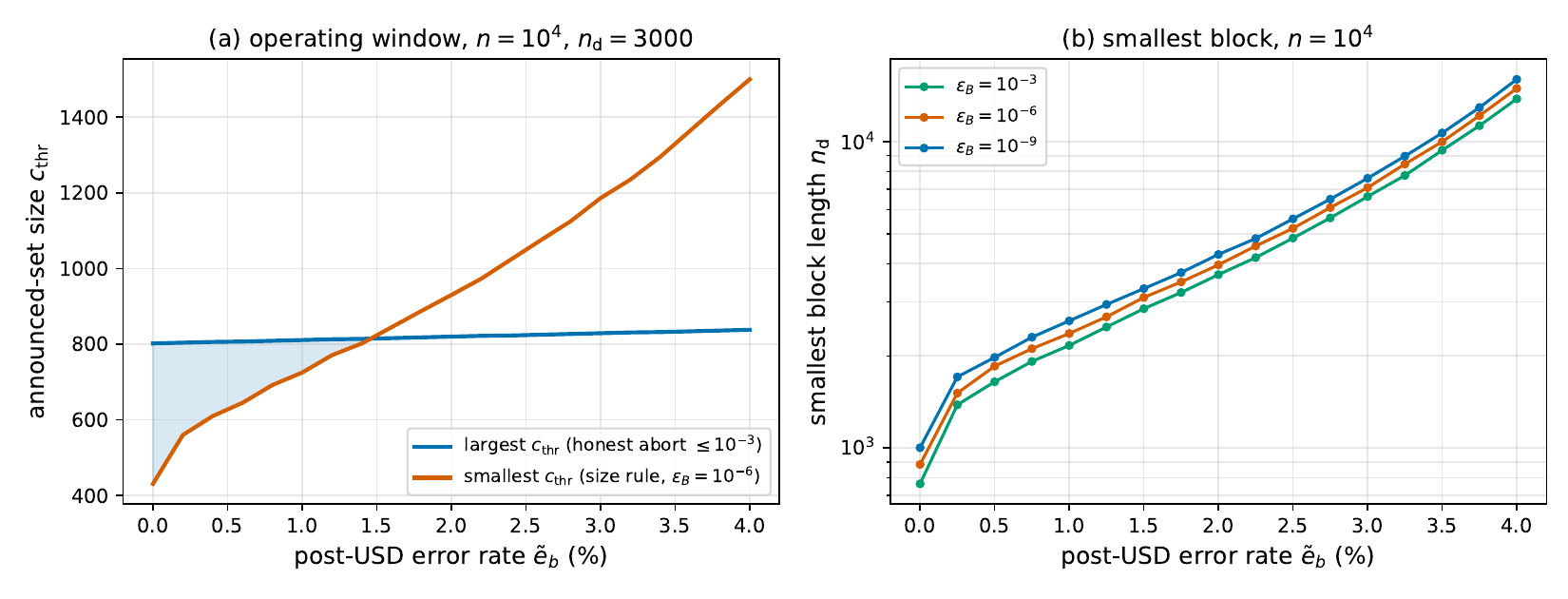}
    \caption{Operating window of the announced set at $\num=10^{4}$, for a depolarizing channel with post-USD error rate $\tilde{\ebit}$, known-rate $\pknown=\USDprob/(1-\tilde{\ebit})$, and $\ECprob=\epsilon_\perp=10^{-6}$ for the decoding-failure and $\perp$ probabilities.  (a)~At $\numd=3000$: the smallest $\cthr$ allowed by the size rule at $\SecretB=10^{-6}$ (red) and the largest one for which an honest Bob aborts with probability at most $10^{-3}$ in Step~\ref{QSPIR:block_draw} (blue); the protocol operates in the shaded region.  (b)~The smallest block length $\numd$ for which the window is open.  Both panels use the exact binomial law of the overlap $o_{jj'}$, the decoding radius $t$ at the exact binomial quantile, and $\hmarg=\lceil\log_2\mathrm{Vol}(\cthr,t)+\log_2(1/\ECprob)\rceil$.}
    \label{fig:CS_window}
\end{figure*}

\begin{table}[h]
\caption{\label{tab:common_set}Smallest block at $\num=10^{4}$, $\SecretB=\ECprob=\epsilon_\perp=10^{-6}$ ($\epsilon_\perp$ bounds the Chernoff term of Lemma~\ref{lem:decodability}), $\epsilon_\text{abort}=10^{-3}$ and $\pknown=\USDprob/(1-\tilde{\ebit})$.  Left: the smallest multiple of ten satisfying Eq.~\eqref{eq:size_rule}, with $\CSSprob$ and $\gamma$ chosen as in Appendix~\ref{app:simulator_proof} and $\hmarg$ from Eq.~\eqref{eq:syndrome_length}.  Right: exact evaluation.}
\begin{ruledtabular}
\begin{tabular}{c|ccc|cccc}
$\tilde{\ebit}$ & $\numd$ & $\cthr$ & $\hmarg$ & $\numd$ & $\cthr$ & $t$ & $\hmarg$\\
\hline
$0$ & $980$ & $228$ & $20$ & $884$ & $218$ & $0$ & $20$\\
$1\%$ & $4360$ & $1154$ & $308$ & $2365$ & $632$ & $21$ & $150$\\
$2\%$ & $6800$ & $1838$ & $518$ & $3962$ & $1096$ & $47$ & $296$\\
$3\%$ & $12330$ & $3405$ & $998$ & $7079$ & $2019$ & $100$ & $590$\\
\end{tabular}
\end{ruledtabular}
\end{table}

\subsection{Key Overhead With Finite Keys}
\label{sec:Finite_key_analysis_qkd}
So far, our estimate of the overhead $\nums$ in Theorem~\ref{thm:success_protocol} rests on the asymptotic key rate.  What happens if Alice sends only finitely many qubits and real codes are used?  Interestingly, Theorem~\ref{thm:success_protocol} carries over.  All we need is to replace the asymptotic key rate $\keyrate$ of Eq.~\eqref{eq:SP_rate_split} by $r_{\text{real}}$, and the condition of Step~\ref{QSPIR:rate_estimation} by $r_{\text{real}}>0$, which defines a finite-key threshold error rate $e_{\text{thr}}^{\text{real}}$.

We quote the approximate realistic key rate of the finite-key analysis of Ref.~\cite{Finite_key_analysis1} (see also Ref.~\cite{Finite_key_analysis2}, and Refs.~\cite{Scarani_QKD_review, Xu_QKD_review} for reviews):

\begin{equation}
    r_{\text{real}}\approx q-H_2(e_{p,\text{tol}}+\xi)-\lambda H_2(e_{b,\text{tol}})-\frac{1}{N}\log_2 \frac{2}{\epsilon_{\text{sec}}^2\epsilon_{\text{cor}}}
    \label{eq:realistic_rate}
\end{equation}
Here the statistical-fluctuation slack $\xi$ (denoted $\mu$ in Ref.~\cite{Finite_key_analysis1}; we reserve $\mu$ for photon intensities) is
\begin{equation}
    \xi := \sqrt{\frac{N+\nume}{N\nume}\,\frac{\nume+1}{\nume}\,\ln\frac{4}{\epsilon_{\text{sec}}}},
    \label{eq:xi_slack}
\end{equation}
with $N$ the raw-key length of the QKD-mode transmission and $\nume$ the size of each parameter-estimation sample, both counted after sifting.  Since the larger of the two estimates of Step~\ref{QSPIR:rate_estimation} is used, one honest estimate suffices, and Eq.~\eqref{eq:xi_slack} applies with the size of one sample.  Here $\epsilon_{\text{sec}}$ and $\epsilon_{\text{cor}}$ are the secrecy and correctness parameters of the QKD key of Step~\ref{QSPIR:standard_secure_bits}, which Table~\ref{tab:notation} writes $\SecretS$ for both.  The argument $4/\epsilon_{\text{sec}}$ follows from Eqs.~(S3)--(S4) of Ref.~\cite{Finite_key_analysis1}, in its updated version arXiv:1103.4130v2 that uses the secrecy criterion of Section~\ref{subsec:definitions}, with $\epsilon=\epsilon_{\text{sec}}/4$.  The rate is per raw-key bit; the sifting factor of Eq.~\eqref{eq:nums_sifting} is applied separately.

Here $q$ is the source preparation quality, defined by $q:=-\log_2\max|\bra{\psi_X}\ket{\psi_Z}|^{2}$ in terms of the maximal overlap between the source's $X$- and $Z$-basis states~\cite{Finite_key_analysis1}.  For ideal BB84 state preparation, $q=1$.  %
The phase-error tolerance $e_{p,\text{tol}}$ carries the statistical-fluctuation penalty $\xi$ [Eq.~\eqref{eq:xi_slack}].  In contrast, the bit-error tolerance $e_{b,\text{tol}}$ multiplies the error-correction inefficiency $\lambda$ because reconciliation acts on the directly measured bit values.  For a symmetric channel, set $e_{b,\text{tol}}=e_{p,\text{tol}}=e_{\text{tol}}$.  Last but not least, $\lambda\ge1$ is the inefficiency ratio of the error-correction code actually used, $\text{leak}_{\text{EC}}\approx\lambda N H_2(e_{b,\text{tol}})$~\cite{Finite_key_analysis1}, with $\lambda=1$ the Shannon limit.  Practical codes reach $\lambda\approx1.1$--$1.2$~\cite{Rates_for_realistic_code2, Rates_for_realistic_code3}.  In all numerical evaluations below, we set $\lambda=1.16$~\cite{Lim_decoy_finite}.

The finite-key threshold is now a \emph{curve} in the $(e_b, e_p)$ plane, obtained by setting
\begin{equation}
    r_{\text{real}}>0
\end{equation}
This condition defines the threshold curve $(e_{b,\text{thr}}^{\text{real}}, e_{p,\text{thr}}^{\text{real}}) = f(\Phi[N, \nume, \ell, e_{\text{tol}}, \epsilon_{\text{cor}}, \text{leak}_{\text{EC}}], q, \epsilon_{\text{sec}})$, where $\Phi$ collects the finite-key protocol parameters, including raw-key size, parameter-estimation sample, key length, error tolerance, correctness, and error-correction leakage~\cite{Finite_key_analysis1}.  We write $e_{\text{thr}}^{\text{real}}$ for its symmetric point $e_b=e_p$.  Written out for the key length, with $\xi$ evaluated at the sample size $\nume$, the condition reads

\begin{equation}
    \begin{split}
        &N(q-H_2(e_{p,\text{thr}}^{\text{real}}+\xi)-\lambda H_2(e_{b,\text{thr}}^{\text{real}}))\\
        &\qquad>\log_2 \frac{2}{\epsilon_{\text{sec}}^2\epsilon_{\text{cor}}}\\
        &H_2(e_{p,\text{thr}}^{\text{real}}+\xi)+\lambda H_2(e_{b,\text{thr}}^{\text{real}})\\
        &\qquad<q-\frac{1}{N}\log_2 \frac{2}{\epsilon_{\text{sec}}^2\epsilon_{\text{cor}}}\\
    \end{split}
\end{equation}
On the symmetric slice $e_{b,\text{thr}}^{\text{real}}=e_{p,\text{thr}}^{\text{real}}=e_{\text{thr}}^{\text{real}}$, it is a condition on $e_{\text{thr}}^{\text{real}}$ alone.  Two trends are apparent.  First, larger $\epsilon_{\text{sec}}$ and $\epsilon_{\text{cor}}$ raise $e_{\text{thr}}^{\text{real}}$.  Second, $e_{\text{thr}}^{\text{real}}$ increases with $N$ and $\nume$, and approaches the asymptotic limit $e_{\text{thr}}\approx9.8\%$, the root of $(1+\lambda)H_2(e)=1$.  Note that $N$ is not tied to the data transmission.  It is fixed by the key length alone, as we now show.

Next, we estimate the overhead under finite-key effects.  By Theorem~\ref{thm:success_protocol}(iv), the key length is $\ell=\num+\lceil\log_2\num\rceil$ bits.  Let $N$ be the smallest raw-key length whose finite-key output on the symmetric slice $e_b=e_p=e$, $N\big(q-H_2(e+\xi)-\lambda H_2(e)\big)-\log_2\frac{2}{\epsilon_{\text{sec}}^2\epsilon_{\text{cor}}}$ [Eq.~\eqref{eq:realistic_rate}] reaches $\ell$.  Since sifting halves both the key-generating qubits and the samples, the QKD mode needs $\nums=2N$ qubits in Step~\ref{QSPIR:standard_secure_bits} and $4\nume$ in Step~\ref{QSPIR:rate_estimation}, two samples of $2\nume$.  Table~\ref{tab:overhead} lists their sum, minimized over $\nume$.

\begin{table}[t]

\caption{\label{tab:overhead}Minimum number of QKD-mode qubits $4\nume+\nums=2(N+2\nume)$ for the key length of Eq.~\eqref{eq:bnums_nums}, on the symmetric slice $e_b=e_p=e$ with $\lambda=1.16$, $\epsilon_{\text{sec}}=\epsilon_{\text{cor}}=0.01$ and the sample size $\nume$ optimized.}
\begin{ruledtabular}
\begin{tabular}{lcccc}
$\num$ & $e=0.005$ & $e=0.01$ & $e=0.02$ & $e=0.05$ \\
\colrule
$10^{4}$ & $4.1\times10^{4}$ & $4.4\times10^{4}$ & $5.3\times10^{4}$ & $1.0\times10^{5}$ \\
$10^{5}$ & $3.0\times10^{5}$ & $3.3\times10^{5}$ & $3.9\times10^{5}$ & $7.2\times10^{5}$ \\
$10^{6}$ & $2.6\times10^{6}$ & $2.8\times10^{6}$ & $3.3\times10^{6}$ & $6.1\times10^{6}$ \\
\end{tabular}
\end{ruledtabular}
\end{table}

The overhead is linear in the database size: at large $\num$ it approaches $2\num/r_{\text{real}}$, while at small $\num$ the finite-key terms of Eq.~\eqref{eq:realistic_rate} dominate, so that its ratio to $\num$ falls from $4.1$ at $\num=10^{4}$ to $2.6$ at $\num=10^{6}$ for $e=0.005$.  Note that a fixed sample $\nume=\sqrt{N}$ is a poor choice at small $\num$: it costs $7.7\times10^{4}$ qubits instead of $4.1\times10^{4}$ at $\num=10^{4}$ and $e=0.005$.  As in QKD, a channel error rate near the threshold drives the overhead up sharply.  These are the costs of the two pads only.  The data transmission adds $\numd\num$ retained positions, and the measurement test and the commitment add their own samples (Appendices~\ref{app:usd_certification} and~\ref{app:storage_commitment}).

\subsection{Conclusive Count Per Block}
\label{sec:sim_histogram}

\begin{figure}
    \centering
    \includegraphics[width=\linewidth]{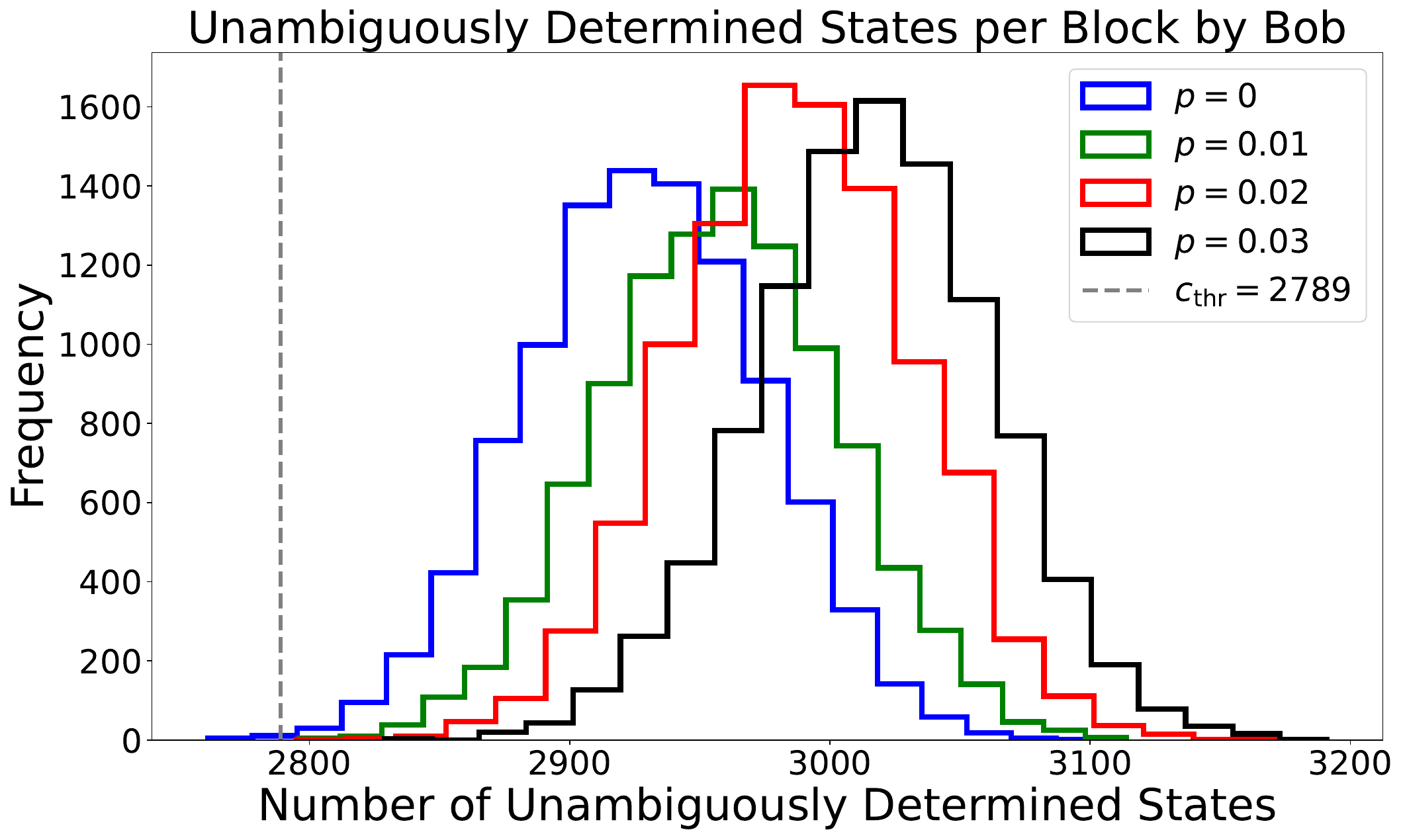}
    \caption{Histogram of the number of conclusive positions per block in a typical run, for $\num=\numd=10^{4}$ and depolarizing noise $p\in\{0, 0.01, 0.02, 0.03\}$.  The dashed line is the set size $\cthr=2789$, for which an honest Bob aborts with probability $10^{-3}$ over a noiseless channel; it obeys the size rule of Eq.~\eqref{eq:size_rule} at $\SecretB=10^{-6}$ for all four noise levels, with $\hmarg$ as in the exact evaluation of Table~\ref{tab:common_set}.}
    \label{fig:USD_number_per_block_histogram}
\end{figure}
In Fig.~\ref{fig:USD_number_per_block_histogram}, the histogram shows the number of conclusive positions per block, for $\num=\numd=10^{4}$ and depolarizing noise $p\in\{0, 0.01, 0.02, 0.03\}$, that is, $\tilde{\ebit}$ up to $2.9\%$, inside the operating window at this block length (Appendix~\ref{app:simulator_proof}). When $p=0$, we can clearly see that the histogram peaks at around $2929$, which means that most of the blocks have around a fraction $\USDprob=1-1/\sqrt{2}$ of their states unambiguously determined. The peak of the number of unambiguously determined states shifts to the right when we introduce noise in the simulation, to $\USDprob(1+p)\numd$; the analysis is in Appendix~\ref{sec:Depolarizing_and_Dephasing_Noise}.

Regardless of the noise level, the count is binomial, and approximately Gaussian near its mean.  This is what the protocol uses.  A block drawn at random carries at least $\cthr<\USDprob\numd$ conclusive positions with high probability, so an honest Bob rarely aborts, while the conclusive patterns of different blocks are independent, so that two blocks overlap in only about $\pknown^{2}\numd$ known locations.

The per-block conclusive count over $10^{4}$
noiseless blocks of the measurement in Step~\ref{QSPIR:initial} has
empirical mean $2928.6$ against the binomial mean $\numd\USDprob=2928.9$ and variance ratio
(empirical over binomial $\numd\USDprob(1-\USDprob)$) of $1.000$: the count is
$\mathrm{Bin}(\numd,\USDprob)$ to within sampling error, with no over-dispersion.  Note, however, that Eq.~\eqref{eq:EW_common} rests on the overlap of two blocks rather than on the count of one.  We test the overlap directly in Section~\ref{sec:sim_pair_attack}.

\subsection{Overlap Of Two Blocks And Tightness Of The Size Rule}
\label{sec:sim_pair_attack}

\begin{figure*}
    \centering
    \includegraphics[width=\textwidth]{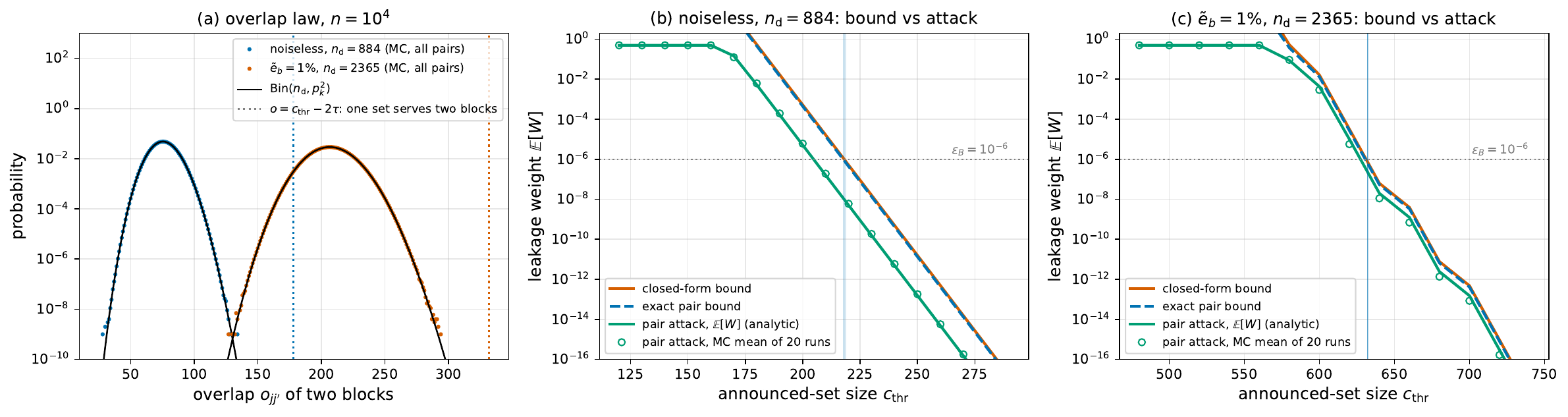}
    \caption{Overlap law and tightness of Eq.~\eqref{eq:EW_common} at $\num=10^{4}$.  (a)~Overlap of all pairs of blocks (dots) against $\mathrm{Bin}(\numd,\pknown^{2})$ (lines); dotted lines mark $o_{jj'}=\cthr-2\hmarg$, where one set would serve two blocks.  (b),~(c)~Leakage weight of the pair attack (green: circles are means over the $20$ databases, the line averages over the law of the largest overlap) against the bound (orange: Eq.~\eqref{eq:EW_common}; dashed: its exact form of Appendix~\ref{app:simulator_proof}), over a noiseless channel and at $\tilde{\ebit}=1\%$.  The vertical line is the smallest $\cthr$ allowed by the size rule.}
    \label{fig:CS1_pair_attack}
\end{figure*}

Is the overlap law behind Eq.~\eqref{eq:EW_common} borne out, and how much does the bound give away?  We simulate $20$ databases of $\num=10^{4}$ blocks at each of the exact-evaluation blocks of Table~\ref{tab:common_set} for $\tilde{\ebit}=0$ and $1\%$, $\numd=884$ and $2365$, each permuted explicitly as in Step~\ref{QSPIR:shuffle}, and record the overlap $o_{jj'}$ of all $\binom{\num}{2}$ pairs, $10^{9}$ per block length.  Figure~\ref{fig:CS1_pair_attack}(a) shows that it follows $\mathrm{Bin}(\numd,\pknown^{2})$ down to probabilities of $10^{-9}$.  Its mean is $75.84$ against $\numd\pknown^{2}=75.84$ with variance ratio $1.000$ over a noiseless channel, and $207.01$ against $207.00$ with ratio $0.999$ at $\tilde{\ebit}=1\%$.  The moment $\log_2\EX\,2^{o_{jj'}/2}$ that enters the bound is $44.55$ and $121.36$, against $44.53$ and $121.51$ from Lemma~\ref{lem:overlap_mgf}.

To test the tightness we let Bob cheat.  He searches the pair of blocks with the largest overlap, puts their common known locations into $\aset$, and splits the rest evenly between the locations he knows in one of the two blocks only.  This meets Lemma~\ref{lem:pair_bound} with equality: both blocks miss about $(\cthr-o_{jj'})/2$ locations.  Panels~(b) and~(c) compare his leakage weight, evaluated over all $\num$ blocks, with Eq.~\eqref{eq:EW_common}.  The circles are the means of $W$ over the $20$ databases, and the green line averages $W$ over the law of the largest overlap, computed as if the pairs were independent.  Below the size rule the attack succeeds.  For $\cthr\lesssim160$ over a noiseless channel it yields $W=\tfrac12$, that is, a second bit.  At the operating point the bound exceeds the line by a factor $2^{6.7}$, which costs about $13$ of the $218$ announced locations, and at $\tilde{\ebit}=1\%$ by $2^{1.7}$, or $3$ of $632$.  Over a noiseless channel the circles agree with the line.  At $\tilde{\ebit}=1\%$ they lie a factor $1.7$ below it near the operating point, and three further independent sets of $20$ databases lie lower still: in the four sets the bound exceeds the mean leak by $2^{2.5}$ to $2^{3.5}$, or $5$ to $7$ locations.  The mean is then dominated by rare databases whose best pair has an unusually large overlap.  Twenty databases seldom contain one, and the independence assumed by the line may overstate how often one occurs.  Either way, the size rule is nearly tight.  Note that an honest Bob leaks far less, $W\approx2^{-105}$, for his set misses about $(1-\USDprob)\cthr$ locations of every other block.

\section{SPIR With Decoy States}
\label{sec:decoy_state}
Our analysis so far assumes a single-photon QKD setting, whereas in most realistic experimental QKD settings, to avoid photon-number-splitting (PNS) attacks, decoy states are the standard defence.  We now show that our scheme extends to the decoy-state setting, at a lower signal intensity than in QKD. In this section, without loss of generality, we demonstrate the key idea using a specific BB84 QKD protocol where the raw key is generated by the $X$ measurement results, and channel phase-error checks are done by the $Z$ measurement results. The argument generalizes easily to other BB84 QKD schemes.

The decoy-state estimates below concern the data and QKD-key transmissions.  The auxiliary commitment transmission must independently satisfy a noise- and loss-tolerant storage-model analysis~\cite{Ng_noisy_storage_commitment}; the data-mode decoy bounds alone do not certify it.  Its additional resources are described in Appendix~\ref{app:storage_commitment}.

Throughout the discussion of decoy-state SPIR, we assume that the parties have the following powers.

Dishonest Bob, Alice's database-privacy adversary, is all-powerful except for long-term quantum memory. He is thus able to perform photon-number resolution, lossless collection, optimal multi-copy photon discrimination, and so on.  As in assumption (A2) of Appendix~\ref{app:usd_certification}, he measures the pulses one at a time, possibly adaptively; within one pulse his measurement is arbitrary.

Eve, the protocol's secrecy adversary, has full channel control, including the PNS attack; her power is handled by the standard decoy-state estimates~\cite{Hwang_decoy, Lo_Ma_Chen_decoy, Wang_decoy, Ma_Qi_Zhao_Lo_decoy}, exactly as in decoy-state QKD (reviewed in Ref.~\cite{Xu_QKD_review}).

Dishonest Alice, Bob's index-privacy adversary, controls the source.  We assume (L) that it is a laser: every pulse is a coherent state whose intensity she may choose slot by slot, with one of the four BB84 polarizations as (P) of Section~\ref{sec:Alice_cheating} requires, but she learns nothing about its photon number.  Without (L), a source that told her the photon number of every pulse while emitting the honest ensemble could not be detected by any test, and would give her an estimated $\SecretA\approx0.07$--$0.10$ at $\mu_1=0.05$ even with the multi-click count of Step~\ref{QSPIR:initial}.

Under (L), the danger is pulses brighter than $\mu_1$, which the count of Step~\ref{QSPIR:initial} addresses as in Appendix~\ref{app:alice_multiphoton}.  An honest source, however, already produces multi-clicks, a fraction $\USDprob(1-\USDprob)\eta\mu_1\approx10^{-3}$ of the detections at the transmittance $\eta=0.1$ and signal intensity $\mu_1=0.05$ used below.  Bob therefore replaces the count $D$ by $\tilde D=D-N_\mathrm{R}\hat P_C\hat P_I$, taken over the $N_\mathrm{R}$ data-acquisition positions with a correct set guess once Alice has disclosed them (Step~\ref{QSPIR:set_disclosure}), where $\hat P_C$ and $\hat P_I$ are the fractions of these positions at which a conclusive and the inconclusive detector clicked.  For pulses of one intensity and polarization the two clicks are independent (Poisson thinning), so $\EX\tilde D=0$ for an honest source, whereas pulses of different intensities raise $\EX\tilde D$ (Chebyshev's sum inequality).  Bob aborts there, before the permutation, if $\tilde D$ exceeds a threshold $\tau_D$ of a few $\sqrt D$.  Pulses much brighter than $\mu_1$ pay for their effect with multi-clicks, as in Theorem~\ref{thm:alice_multiphoton}.  Pulses only slightly brighter give Alice little information on Bob's block.  They change the conclusive probability of a position by only about $0.10\,\eta\,\Delta\mu$, with $\Delta\mu$ the excess intensity.  For example, making half the pulses $10\%$ brighter and half $10\%$ dimmer passes the count, yet moves her posterior by only about $0.002$.  A Monte-Carlo simulation near the operating point supports the rule ($\num=10^{4}$, $\numd=2610$, $\eta=0.1$; honest $\EX D=2.7\times10^{4}$ and $\tau_D=508$ for an honest abort probability of $10^{-3}$).  The strongest attack by pulses of a legitimate polarization that still passes the count with probability $\tfrac12$ moves Alice's posterior on Bob's block by $0.03$ in total variation on average; it uses about $\tau_D$ very bright pulses, each of which is conclusive and clicks two detectors with certainty.  Without the count, the same attacks reach $0.8$ to $1$.

The honest receiver applies the measurement of Step~\ref{QSPIR:initial}, for the set he guessed, to the whole pulse with passive optics, so that every photon of the pulse meets the same POVM.  Bob calls a position conclusive if a conclusive detector fires (Step~\ref{QSPIR:initial}).  Each photon fails independently with probability $\nu=1/\sqrt2$, so he attains the rate $\pUSDN{N}=1-\nu^{N}$ of Lemma~\ref{lem:beff}, which is also the optimum.  The privacy analysis, by contrast, holds against the dishonest Bob described above.

\subsection{Decoy-State QKD For SPIR}
\label{sec:decoy_for_spir}

As shown in Section~\ref{subsec:reconciliation}, all we require is an announced set whose size obeys Eq.~\eqref{eq:size_rule} at the certified known-rate, check symbols sized by Eq.~\eqref{eq:syndrome_length} at the conclusive-bit error rate, and a QKD key of positive rate for the pads. We now show that this conclusion survives when the ideal single-photon source assumed in Sections~\ref{sec:Bob_Cheating} and~\ref{sec:Finite_key_analysis_qkd} is replaced by a realistic phase-randomized weak coherent pulse (WCP) source of mean photon number $\mu$, whose photon number is Poisson distributed, $p_N(\mu)=e^{-\mu}\mu^N/N!$.

We adopt the standard three-intensity (vacuum $+$ weak-decoy $+$ signal) layout~\cite{Ma_Qi_Zhao_Lo_decoy}: the signal intensity $\mu_1$ carries the $\numd n$ data (USD) pulses, while the weak-decoy and vacuum intensities $\mu_2,\mu_3$ are interspersed at random, sacrificed, and used only for parameter estimation. We take $\mu_1=0.05$, $\mu_2=\mu_1/10$, $\mu_3=0$. The signal intensity sits an order of magnitude below the QKD-conventional $\mu_1\approx0.5$ because in our protocol multi-photon pulses do more than reduce the key rate: they inflate the known-rate at which the size rule must be evaluated (at least the tagging bound $\pUSDplus$ of Lemma~\ref{lem:beff_plus}), and this inflation does \emph{not} vanish with channel loss. The rule has a solution only if $\pUSDplus\lesssim0.5$ [Eq.~\eqref{eq:mu1_constraint} of Section~\ref{sec:photon_tagging}], and the block length grows quickly as this limit is approached; the price of $\mu_1=0.05$ is a ${\sim}10\times$ throughput overhead relative to $\mu_1=0.5$. Because every \emph{data} pulse then carries the single intensity $\mu_1$, no intensity label correlates with Bob's conclusive pattern when Alice follows the protocol (a cheating Alice is treated at the start of this section); the decoy estimates below are drawn only from the sacrificed subset.

The decoy-state quantities are standard~\cite{Ma_Qi_Zhao_Lo_decoy, Decoy_state_key_rate1} (see Ref.~\cite{Xu_QKD_review} for a review); we summarize them here: the per-intensity gain $Q_{X,\mu}$ and bit-error rate $E_{X,\mu}$ at intensity $\mu$; the $N$-photon yields $Y_{X,N}$, with $Y_{X,0},Y_{X,1}$ the vacuum and single-photon yields; the single-photon bit error $e_{X,1}$ and phase error $e_p$; the asymptotic secret-key rate per emitted pulse $R_{\text{decoy}}$ and the rate per sifted detection $r_{\text{decoy}}:=R_{\text{decoy}}/\big(p_X^2\langle Q_{X,\mu}\rangle\big)$, where $p_X$ is the probability of choosing the $X$ basis and the intensity average $\langle\cdot\rangle$ is defined below; the number of $X$-sifted pulses $s_X$, with vacuum and single-photon counts $s_{X,0},s_{X,1}$; the intensity average $\langle f(\mu)\rangle:=\sum_\mu p_\mu f(\mu)$ over the three intensities, $p_\mu$ being the probability of intensity $\mu$; the worst-case (Hoeffding) fluctuated gain $Q^{\langle\langle\cdot\rangle\rangle}:=Q\pm\Delta Q$; and the security factor $\chi(3)=21$~\cite{Lim_decoy_finite}, whose argument is the number of intensities. The finite-key channel-error tolerance $e_{\text{tol}}$ and source preparation quality $q$ are those of Section~\ref{sec:Finite_key_analysis_qkd}.

The multi-photon tail $N\ge2$ of the Poisson distribution creates two logically distinct effects, which must be controlled separately because they bear on different security goals. First, it exposes the protocol to the photon-number-splitting (PNS) attack against Eve-secrecy, which is handled in the standard way by the decoy single-photon estimates~\cite{Hwang_decoy, Lo_Ma_Chen_decoy, Wang_decoy, Decoy_state_key_rate1}. Secondly --- and with no analogue in ordinary QKD --- an $N$-photon data pulse delivers $N$ identical copies of the announced state to Bob, enhancing his unambiguous discrimination and threatening database privacy. The second effect is handled as in the single-photon case: Section~\ref{sec:decoy_receiver_test} tests Bob's declarations intensity by intensity, and under assumptions (A1)--(A3) of Appendix~\ref{app:usd_certification}, read pulse by pulse, the test certifies the known-rate at which the size rule is evaluated.

\subsection{The Effective USD Rate Under Weak Coherent Pulses}

\begin{lemma}[Effective USD rate]
\label{lem:beff}
Let the announced alternatives $\{\ket{0},\ket{+}\}$ have overlap $\nu=|\bra{0}\ket{+}|=1/\sqrt2$ and single-copy USD ceiling $\USDprob=1-\nu\approx0.293$. If a data pulse delivers $N$ photons to Bob, the optimal $N$-copy unambiguous discrimination of $\ket{0}^{\otimes N}$ from $\ket{+}^{\otimes N}$ succeeds with probability
\begin{equation}
    \begin{split}
    \pUSDN{N}&=1-\nu^{N}=1-2^{-N/2},\\
    \pUSDN{0}&=0,\quad \pUSDN{1}=\USDprob,\quad \pUSDN{2}=\tfrac12 .
    \end{split}
    \label{eq:bN}
\end{equation}
With channel-plus-detector transmittance $\eta$, the per-pulse and per-detected-pulse probabilities that the honest receiver obtains a conclusive USD outcome are
\begin{equation}
\begin{split}
    \pUSDeff&=1-e^{-\USDprob\,\eta\mu_1},\\
    \pUSDsift&=\frac{1-e^{-\USDprob\,\eta\mu_1}}{1-e^{-\eta\mu_1}}\\
    &=\USDprob\Big[1+\tfrac{1-\USDprob}{2}\,\eta\mu_1+O\big((\eta\mu_1)^2\big)\Big]\ \ge\ \USDprob .
\end{split}
    \label{eq:beff}
\end{equation}
\end{lemma}

The proof (Appendix~\ref{app:decoy_proofs}) applies the Ivanovi\'c--Dieks--Peres bound to the $N$-copy overlap $\nu^{N}$ and averages over the Poisson photon-number statistics of the attenuated coherent state.

The fractional enhancement over the single-photon ceiling is $\tfrac{1-\USDprob}{2}\eta\mu_1=\tfrac{\nu}{2}\eta\mu_1\approx0.35\,\eta\mu_1$. Throughout, $\USDprob=1-\nu\approx0.293$ is both the single-copy optimum and the rate of the measurement of Step~\ref{QSPIR:initial}.

\subsection{Photon Tagging And The Limit On The Signal Intensity}
\label{sec:photon_tagging}

\begin{lemma}[Photon tagging]
\label{lem:beff_plus}
Let the yields and the gain below refer to one population of detected signal pulses, in which a vacuum event carries no information on the bit and Bob knows the bit of a single-photon event with rate at most $\USDprob$.

Let $Y_{X,0}^{-}$ and $Y_{X,1}^{-}$ be the decoy lower bounds on the vacuum and single-photon yields obtained from the standard three-intensity decoy estimates~\cite{Ma_Qi_Zhao_Lo_decoy, Decoy_state_key_rate1} evaluated at worst-case (Hoeffding) fluctuated gains, and let $Q_{X,\mu_1}$ be the observed signal-intensity gain, with $\Delta Q_{X,\mu_1}$ its Hoeffding fluctuation~\cite{Decoy_state_key_rate1}. Define
\begin{equation}
    \Delta^{+}:=\min\!\Big(1,\ \max\!\Big(0,\
    1-\frac{e^{-\mu_1}Y_{X,0}^{-}+\mu_1 e^{-\mu_1}Y_{X,1}^{-}}{Q_{X,\mu_1}+\Delta Q_{X,\mu_1}}\Big)\Big).
    \label{eq:Delta_plus}
\end{equation}
Then, except with probability at most $5\epsilon_{\text{sec}}/\chi(3)$ (one $\epsilon_{\text{sec}}/\chi(3)$ per one-sided Hoeffding event, enumerated in the proof), $\Delta^{+}$ upper-bounds the fraction of \emph{detected} signal pulses carrying two or more photons, and consequently the rate at which Bob knows the bit of a detected signal pulse is at most $\pUSDplus$; in particular
\begin{equation}
    \pUSDsift\ \le\ \pUSDplus:=\USDprob+\Delta^{+}(1-\USDprob).
    \label{eq:beff_plus}
\end{equation}
\end{lemma}

The proof (Appendix~\ref{app:decoy_proofs}) enumerates the five one-sided Hoeffding events behind the direction $\Delta^{+}\ge\Delta$ and then charges every multi-photon detection the worst case $\pUSDN{N}\le1$.

Note that the single-photon hypothesis cannot simply be assumed, even for the measurement of Step~\ref{QSPIR:initial}: a Bob who reports his inconclusive photons as lost knows every retained bit.  Nor need the QKD-mode yields describe the data positions, for Bob names the mode of a pulse after measuring it.  The test of Section~\ref{sec:decoy_receiver_test} settles both points.  It works on the retained data positions, where $\Delta^{+}$ becomes $\Delta_R^{+}$ of Eq.~\eqref{eq:decoy_retained_tag}, and it replaces $\USDprob$ by $\min\{1,B(\delta_1)\}$ [Eq.~\eqref{eq:pk_wcp}].

We now explain why the signal intensity of Section~\ref{sec:decoy_for_spir} must be small. Unlike ordinary decoy-state QKD, where multi-photon pulses are merely tagged and their key contribution discarded, our SPIR layer evaluates the size rule of Proposition~\ref{prop:db_privacy_decoy} at the certified known-rate of Proposition~\ref{prop:cert_rate_wcp}, which lies above the tagging bound $\pUSDplus$ of Lemma~\ref{lem:beff_plus} even for an infinite test, whereas the honest Bob is conclusive at the rate $\pUSDsift$ of Lemma~\ref{lem:beff}, only slightly above $\USDprob$ ($0.2934$ at $\eta\mu_1=0.005$), and announces a set sized from $\USDprob$, $\cthr<\USDprob\numd$.  The gap between $\pUSDplus$ and $\pUSDsift$ does \emph{not} vanish at lower transmittance: for a pure-loss channel the multi-photon detected fraction obeys
\begin{equation}
    \Delta^{+}\ \ge\ 1-\frac{\mu_1 e^{-\mu_1}\,\eta}{1-e^{-\eta\mu_1}}\ \xrightarrow[\ \eta\to0\ ]{}\ 1-e^{-\mu_1},
    \label{eq:Delta_plus_scaling}
\end{equation}
so, for exact statistics, $\pUSDplus-\pUSDsift\le\Delta^{+}(1-\USDprob)=O(\mu_1)$ for every $\eta$.  With finite samples the Hoeffding terms add a floor to $\Delta^{+}$ that does not vanish as $\mu_1\to0$.  Multi-photon pulses survive loss more often, so a lossier channel makes the gap larger, not smaller. The size rule of Eq.~\eqref{eq:size_rule} has a solution with $\cthr<\USDprob\numd$ only if
\begin{equation}
    2\log_2\!\Big[1+(\sqrt2-1)\big(\pUSDplus\big)^{2}\Big]\;<\;\USDprob\big[1-2H_2(\tilde{e}_b+\gamma)\big],
    \label{eq:mu1_constraint}
\end{equation}
that is, $\pUSDplus<0.51$ or $\Delta^{+}<0.30$ over a noiseless channel. At the QKD-conventional $\mu_1=0.5$ one has $\Delta^{+}\approx0.23$--$0.39$ and the condition fails or is nearly saturated. Figure~\ref{fig:DS4_known_rate} shows the price at $\eta=0.1$, $\num=10^{4}$ and $\SecretB=10^{-6}$ over a noiseless channel.  With the tagging bound $\pUSDplus$ as the known-rate, the closed-form choice of Appendix~\ref{app:simulator_proof} gives $\numd=1140$, $2340$ and $6700$ at $\mu_1=0.05$, $0.2$ and $0.3$, and no solution at $\mu_1\ge0.4$.  No finite test certifies this bound, however.  The known-rate that the test of Section~\ref{sec:decoy_receiver_test} certifies, $\pknown^{\mathrm{WCP}}$ of Proposition~\ref{prop:cert_rate_wcp}, stays above a floor even for an infinite test, because the weak decoy $\mu_2=\mu_1/10$ resolves the single-photon yields only to a relative error of order $\mu_2$, and the guessing bound turns this error into its square root (Appendix~\ref{app:decoy_receiver_certificate}).  At the floor the block is $\numd=1540$ at $\mu_1=0.05$ and $2660$ at $\mu_1=0.1$, a test of $10^{11}$ pulses gives $2610$ and $4540$, and at $\mu_1=0.2$ the floor already sits at the limit of Eq.~\eqref{eq:mu1_constraint}.  This fixes the operating point $\mu_1=0.05$, $\mu_2=\mu_1/10$, $\mu_3=0$ adopted in Section~\ref{sec:decoy_for_spir}, where the block is $1.6$ to $2.7$ times the single-photon one. The throughput price is a factor ${\sim}0.5/\mu_1\approx10$ in transmitted pulses per data symbol relative to $\mu_1=0.5$, with correspondingly larger Hoeffding fluctuation terms in the decoy estimates.

\begin{figure*}
    \centering
    \includegraphics[width=0.9\textwidth]{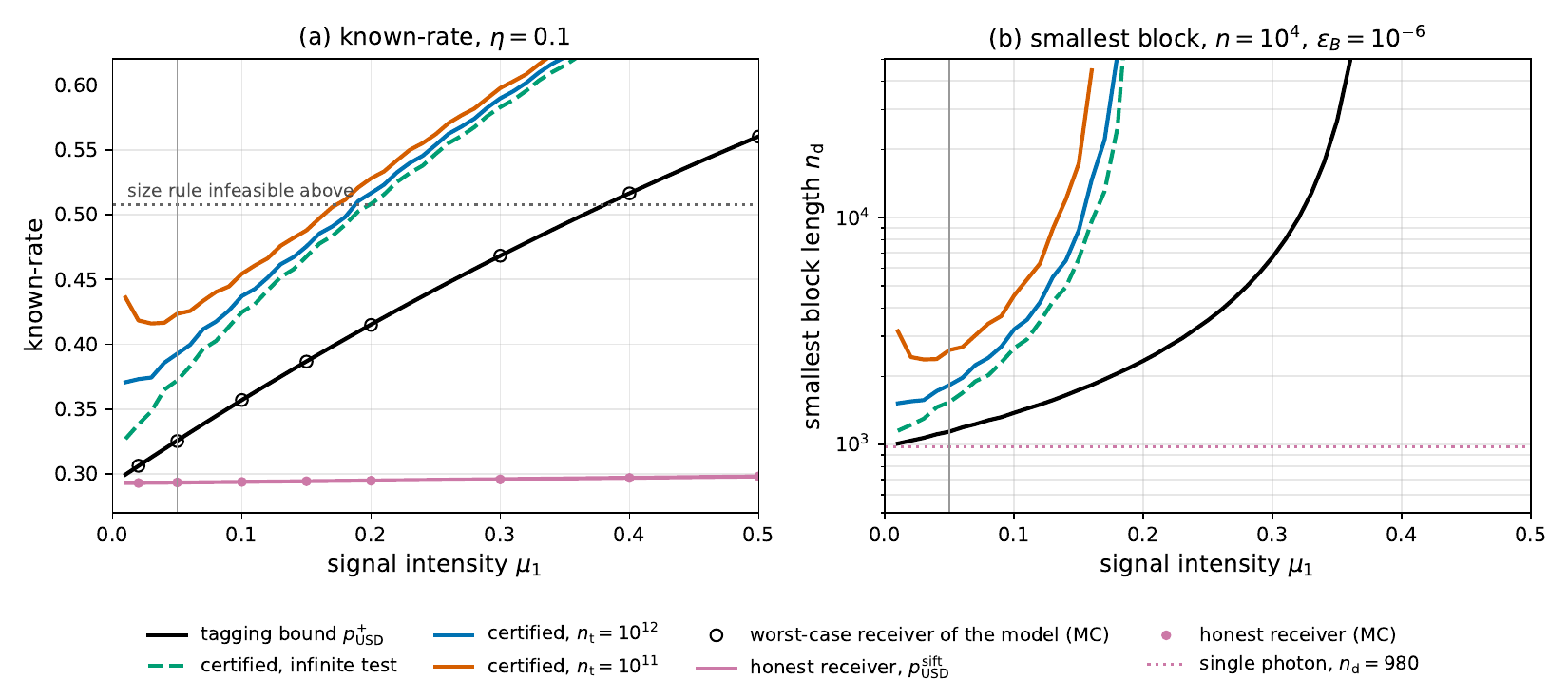}
    \caption{The price of weak coherent pulses at transmittance $\eta=0.1$ over a noiseless channel.  (a)~Known-rate against the signal intensity $\mu_1$: the tagging bound $\pUSDplus$ of Lemma~\ref{lem:beff_plus} (black line), attained by the worst-case receiver of the model, which collects every emitted photon and always knows the bit of a multi-photon pulse (circles, Monte-Carlo); the certified known-rate $\pknown^{\mathrm{WCP}}$ of Proposition~\ref{prop:cert_rate_wcp} for an infinite test (dashed) and for tests of $n_\text{t}=10^{12}$ and $10^{11}$ pulses, with $\mu_2=\mu_1/10$, intensity probabilities $(0.5,0.4,0.1)$, $\epsilon_{\mathrm{PE}}=\epsilon_\text{cert}=10^{-6}$ and the test counts at their expected values; and the conclusive rate $\pUSDsift$ of the honest receiver [pink line, Eq.~\eqref{eq:beff}; dots, photon-level Monte-Carlo].  Above the dotted line, $\pknown=0.51$, the size rule has no solution [Eq.~\eqref{eq:mu1_constraint}].  (b)~Smallest block from the closed-form choice of Appendix~\ref{app:simulator_proof} at $\num=10^{4}$ and $\SecretB=10^{-6}$ for each known-rate of (a), against the single-photon block $\numd=980$ of Table~\ref{tab:common_set}.  The vertical line marks the operating point $\mu_1=0.05$.}
    \label{fig:DS4_known_rate}
\end{figure*}

We comment on the origin of the gap $\pUSDplus-\pUSDsift$: it is not due to measurement ability. $\pUSDsift$ already credits Bob with the \emph{optimal} multi-copy discrimination: Lemma~\ref{lem:beff} evaluates the Poisson mixture of the $N$-copy optima $\pUSDN{N}$. What separates the two rates is instead the channel and the detection record. $\pUSDsift$ assumes an honest pure-loss channel: the channel's properties decide which photons reach Bob, and every click is reported truthfully. $\pUSDplus$ must additionally cover a Bob who collects all $N$ emitted photons losslessly, measures the photon number nondestructively, and post-selects which pulses to report as detected, disguising the discarded ones as channel loss. We must allow for this when Bob is dishonest.

Equation~\eqref{eq:beff_plus} is deliberately loose: it credits the non-multi-photon fraction with the full single-photon rate $\USDprob$ (ignoring that dark-count detections contribute $0$) and charges every multi-photon detection the worst case $\pUSDN{N}=1$. Both choices over-estimate $\pUSDsift$, which is exactly the conservative direction for database privacy. Crucially, Alice computes $\pUSDplus$ from the sacrificed decoy/test subset alone, she never observes which data pulses Bob actually resolved, so estimating the known-rate, rather than measuring Bob's conclusive pattern, is what preserves index privacy.

\subsection{Intensity-Resolved Measurement Test}
\label{sec:decoy_receiver_test}
The test below replaces Step~\ref{QSPIR:usd_test} when the source emits weak coherent pulses; its fixed-batch sampling and disclosure order replace the corresponding single-photon rules in Steps~\ref{QSPIR:initial}--\ref{QSPIR:usd_test}.  Bob fixes his losses, modes and set guesses for a fixed batch of emitted pulses, and he declares his outcomes on the sampled test slots before Alice reveals their states and intensities.  At a fixed photon number a pulse does not depend on its intensity~\cite{Ma_Qi_Zhao_Lo_decoy}.  Sorting the declarations by intensity, Alice can therefore apply the decoy equations to each declaration, bound the declaration statistics of single photons, and compare these bounds with Table~\ref{tab:usd_test}.  Their largest deviation $\delta_1$ from the table [Eq.~\eqref{eq:decoy_single_deviation}] limits how well Bob guesses the bit of a single photon (Lemma~\ref{lem:guessing_bound}), while every retained multi-photon position, a fraction at most $\Delta_R^+$ [Eq.~\eqref{eq:decoy_retained_tag}], is counted as known.  Together they give the certified guessing probability $g_R^{\mathrm{cert}}$ of Eq.~\eqref{eq:decoy_retained_guess}.  Appendix~\ref{app:decoy_receiver_test} states the test in full.

What does this give for the blocks?  By Lemma~\ref{lem:erasure_domination}, a bit that Bob guesses with probability $g$ is worth a known-flag of rate $2g-1$.  Proposition~\ref{prop:cert_rate_wcp} shows that, on passing the test, the mean flag rate of the retained positions is at most a certified rate $\pknown^{\mathrm{WCP}}$ [Eq.~\eqref{eq:pk_wcp_full}], which for a large test sample reads
\begin{equation}
    \pknown^{\mathrm{WCP}}\approx2g_R^{\mathrm{cert}}-1=(1-\Delta_R^{+})\min\{1,B(\delta_1)\}+\Delta_R^{+},
    \label{eq:pk_wcp}
\end{equation}
with $B$ of Eq.~\eqref{eq:decoy_G_bound}.  Alice accepts only if the size rule of Eq.~\eqref{eq:size_rule} holds at this rate, and then proceeds as in Step~\ref{QSPIR:set_disclosure}.  A mean rate is all that Lemma~\ref{lem:overlap_mgf} needs.  Note that Bob may measure adaptively; only joint measurements on several pulses remain excluded, as in assumption (A2).  Note also that $\pknown^{\mathrm{WCP}}$ tends to the rate $\pUSDplus$ of Lemma~\ref{lem:beff_plus} only as $\delta_1\to0$, which needs a vanishing weak decoy as well as a large test; at $\mu_2=\mu_1/10$ it stays above a floor (Fig.~\ref{fig:DS4_known_rate} and Appendix~\ref{app:decoy_receiver_certificate}).

\subsection{Switching From Single-Photon To Decoy-State Analysis}

\begin{proposition}[Four conservative substitutions]
\label{prop:four_subs}
The database-privacy conclusions here and in Theorem~\ref{thm:decoy_success} retain the additional hypotheses of Proposition~\ref{prop:db_privacy_decoy} and the underlying single-photon security analysis.

The single-photon analysis of Sections~\ref{sec:Bob_Cheating} and~\ref{sec:Finite_key_analysis_qkd} extends to the WCP source under the following substitutions, each in its worst-case direction, leaving every theorem statement and proof structurally unchanged.
\begin{enumerate}
    \item[\textbf{(S1)}] \emph{Key rate and threshold.} Three quantities must be kept apart: the key length $\ell$ of Eq.~\eqref{eq:bnums_nums}; the secret bits per \emph{emitted} pulse, $R_{\text{decoy}}$; and the secret bits per \emph{sifted detection}, $r_{\text{decoy}}=R_{\text{decoy}}/\big(p_X^{2}\langle Q_{X,\mu}\rangle\big)$, which is the quantity comparable to the single-photon $r$ of Theorem~\ref{thm:success_protocol}.  The QKD mode then needs $\ell/R_{\text{decoy}}$ emitted pulses, that is, $\ell/r_{\text{decoy}}$ sifted detections, and $e_{\text{thr}}\mapsto e_{\text{thr}}^{\text{decoy}}:=\sup\{e_p:R_{\text{decoy}}(e_p)>0\}$~\cite{Decoy_state_key_rate1, Ma_Qi_Zhao_Lo_decoy, Lim_decoy_finite}. The key requirement of Eq.~\eqref{eq:bnums_nums} is independent of the rate, so the smaller rate $r_{\text{decoy}}\le r$ only enlarges the raw overhead $\nums$.
    \item[\textbf{(S2)}] \emph{Known-rate.} $\pknown\mapsto\pknown^{\mathrm{WCP}}$, the certified rate of Proposition~\ref{prop:cert_rate_wcp}, which is at least $\pUSDplus=\USDprob+\Delta^{+}(1-\USDprob)$ of Lemma~\ref{lem:beff_plus} (Fig.~\ref{fig:DS4_known_rate}), as the known-rate $\pknown$ in the size rule of Eq.~\eqref{eq:size_rule}; the set size $\cthr$ remains fixed by $\USDprob$ (Step~\ref{QSPIR:code_sizing}); database privacy is preserved by Proposition~\ref{prop:db_privacy_decoy}.
    \item[\textbf{(S3)}] \emph{Statistical fluctuations.} All decoy estimates~\cite{Ma_Qi_Zhao_Lo_decoy, Decoy_state_key_rate1} are evaluated at the worst-case fluctuated gains $Q^{\langle\langle\cdot\rangle\rangle}$, and the budget $\epsilon_{\text{sec}}$ is partitioned by a union bound over the QKD-secrecy event, the at most six decoy-estimation statements (each allotted $\epsilon_{\text{sec}}/\chi(3)$), the database-privacy bound of Proposition~\ref{prop:db_privacy_decoy}, and correctness. The resulting certified counts $s_{X,0},s_{X,1}$ enter the finite-key rate of Eq.~\eqref{eq:r_decoy_real}.
    \item[\textbf{(S4)}] \emph{Two error rates.} The channel error rates $\ephase$ and $\ebit$ become the single-photon phase error $e_p$, which enters $R_{\text{decoy}}$, privacy amplification, and $e_{\text{thr}}^{\text{decoy}}$, and the single-photon bit error $e_{X,1}$, which enters error correction of the QKD key.  The check symbols of Eq.~\eqref{eq:syndrome_length} are sized with the conclusive-bit error rate $\tilde{\ebit}$.  It needs no certificate: if Bob overstates it, $\hmarg$ grows, and Alice checks Eq.~\eqref{eq:size_rule} with the $\hmarg$ she uses.  For correctness it is estimated directly on the conclusive declarations of the signal-intensity test slots of Section~\ref{sec:decoy_receiver_test}, not inferred from $e_{X,1}$, because single-photon errors do not bound the errors of multi-photon data pulses.
\end{enumerate}
\end{proposition}

The proof is deferred to Appendix~\ref{app:decoy_proofs}: (S1) is definitional --- the downstream analysis consumes the QKD layer only through the certified scalars --- while (S2) is Proposition~\ref{prop:cert_rate_wcp} together with Proposition~\ref{prop:db_privacy_decoy}.

For the test of Section~\ref{sec:decoy_receiver_test}, allocate $\epsilon_{\mathrm{PE}}$ separately; its $84$ declaration-frequency intervals are not among the decoy-estimation statements of (S3).

In the finite-key regime, substitution (S1) replaces the asymptotic preparation-quality rate by the decoy single-photon secret length~\cite{Decoy_state_key_rate1, Lim_decoy_finite}; the realistic rate of Eq.~\eqref{eq:realistic_rate} becomes, for the padding key generated from the $s_X$ sifted pulses,
\begin{equation}
\begin{split}
    r_{\text{decoy}}^{\text{real}}\approx{}&\frac{s_{X,0}+s_{X,1}\big[1-H_2(e_p)\big]-\lambda\,s_X H_2(e_{\text{tol}})}{s_X}\\
    &-\frac{6\log_2\frac{\chi(3)}{\epsilon_{\text{sec}}}+\log_2\frac{2}{\epsilon_{\text{cor}}}}{s_X},
\end{split}
    \label{eq:r_decoy_real}
\end{equation}
per sifted detection. Three conventions of Eq.~\eqref{eq:realistic_rate} carry over in modified form. The perfect-source fraction $q$ becomes the certified single-photon secret fraction $(s_{X,0}+s_{X,1}[1-H_2(e_p)])/s_X$. The statistical slack $\xi$ of $H_2(e_{\text{tol}}+\xi)$ is now carried by the phase error itself.  With $e_{Z,1}^{+}$ the decoy upper bound on the single-photon error rate in the $Z$ basis and $s_{Z,1}^{-},s_{X,1}^{-}$ the certified single-photon counts,
\begin{equation}
    e_p:=e_{Z,1}^{+}+\sqrt{\frac{s_{X,1}^{-}+s_{Z,1}^{-}}{s_{X,1}^{-}\,s_{Z,1}^{-}}\,\frac{s_{Z,1}^{-}+1}{s_{Z,1}^{-}}\,\ln\frac{\chi(3)}{\epsilon_{\text{sec}}}}\,,
    \label{eq:ep_margin}
\end{equation}
which is the margin of Eq.~\eqref{eq:xi_slack} applied to the single-photon events.  Note that it is defined for all counts.  The margin $\bar\gamma$ of Ref.~\cite{Decoy_state_key_rate1} is not: at $\epsilon_{\text{sec}}=0.01$ its logarithm turns negative once the counts exceed about $1.4\times10^{8}$ at a $1\%$ error rate.  The price is a larger margin where both are defined, $3.9\times10^{-3}$ against $3.1\times10^{-4}$ at $10^{6}$ counts and a $1\%$ error rate.  Together with the $\Delta Q$ corrections, this is why $\xi$ does not reappear in Eq.~\eqref{eq:r_decoy_real}. The $\lambda$ error-correction term keeps its convention of leaking on the whole sifted string, multi-photon included, while the constant log term acquires the decoy-estimation penalty $6\log_2(\chi(3)/\epsilon_{\text{sec}})$. Finally, the data pulses generate no key; they never enter Eq.~\eqref{eq:r_decoy_real} and are governed instead by Proposition~\ref{prop:db_privacy_decoy}.

\subsection{Database Privacy And The Existence Theorem}

\begin{proposition}[Database privacy under WCP]
\label{prop:db_privacy_decoy}
Assume (A1)--(A3) of Appendix~\ref{app:usd_certification}, read pulse by pulse, and let the test of Section~\ref{sec:decoy_receiver_test} pass.

Evaluate the size rule at the upper-bounded rate: let $\cthr$ obey Eq.~\eqref{eq:size_rule} with the known-rate $\pknown=\pknown^{\mathrm{WCP}}$, which is at least $\pUSDplus$ of Lemma~\ref{lem:beff_plus} even for an infinite test.  Then Theorem~\ref{thm:success_protocol}(ii) holds for the WCP source,
\begin{equation}
\begin{split}
    \SecretB\le{}&\epsilon_{\mathrm{PE}}+\epsilon_\text{cert}+\epsilon_{\mathrm{rnd}}\\
    &+\frac{\num^{2}}{4}\,2^{-(\cthr-2\hmarg)/2}\Big[1+(\sqrt2-1)\big(\pknown^{\mathrm{WCP}}\big)^{2}\Big]^{\numd} .
\end{split}
    \label{eq:p1_decoy}
\end{equation}
\end{proposition}

The proof (Appendix~\ref{app:decoy_proofs}) reduces the claim to the single-photon analysis: the flags of Lemma~\ref{lem:erasure_domination} have mean rate at most $\pknown^{\mathrm{WCP}}$ on the retained positions, and the right side of Eq.~\eqref{eq:EW_common} increases with the known-rate.

\begin{theorem}[Existence of a decoy-state SPIR scheme]
\label{thm:decoy_success}
Let Alice use a phase-randomized WCP source with intensities $\mu_1>\mu_2>\mu_3=0$, and let $r_{\text{decoy}},e_p,e_{X,1}$ be the decoy-certified secret-key rate, single-photon phase error and single-photon bit error, $\pknown^{\mathrm{WCP}}$ the certified known-rate of Proposition~\ref{prop:cert_rate_wcp}, which is at least the tagging bound $\pUSDplus$ of Lemma~\ref{lem:beff_plus} (Fig.~\ref{fig:DS4_known_rate}), and $\tilde{e}_b$ the conclusive-bit error rate of item (S4). Assume (A1)--(A3) of Appendix~\ref{app:usd_certification}, read pulse by pulse. Then for every single-photon phase error $e_p<e_{\text{thr}}^{\text{decoy}}$ the SPIR protocol of Section~\ref{sec:protocol}, run with these intensities, with the test of Section~\ref{sec:decoy_receiver_test} in place of Step~\ref{QSPIR:usd_test}, with the check symbols of Eq.~\eqref{eq:syndrome_length}, and with a set size obeying the rule at this known-rate,
\begin{equation}
\begin{split}
\cthr\;\ge\;&2\hmarg+2\numd\log_2\!\Big[1+(\sqrt2-1)\big(\pknown^{\mathrm{WCP}}\big)^{2}\Big]\\
&+4\log_2\num+2\log_2\frac{1}{\SecretB}-4 ,
\end{split}
    \label{eq:overhead_decoy}
\end{equation}
at the cost of $\big(\num+\lceil\log_2\num\rceil\big)/R_{\text{decoy}}$ emitted QKD-mode pulses, that is, $\big(\num+\lceil\log_2\num\rceil\big)/r_{\text{decoy}}$ sifted detections [item (S1) of Proposition~\ref{prop:four_subs}], is $\epsilon_{\text{sec}}$-secret against Eve, leaks to Bob at most $\SecretB+\epsilon_{\mathrm{PE}}+\epsilon_\text{cert}+\epsilon_{\mathrm{rnd}}$ [Eq.~\eqref{eq:p1_decoy}], where $\SecretB$ is the budget for $\EX[W]$ in Eq.~\eqref{eq:overhead_decoy}, and leaves Bob's index private against an Alice who uses the prescribed source (Proposition~\ref{prop:index_privacy_honest}); for an Alice under (P) and (L), see the multi-click count at the start of this section.
\end{theorem}

The proof (Appendix~\ref{app:decoy_proofs}) carries the substitutions of Proposition~\ref{prop:four_subs} through Theorem~\ref{thm:success_protocol}; the key requirement of Eq.~\eqref{eq:bnums_nums} is finite and its conversion to a number of pulses has its only pole at $R_{\text{decoy}}=0$, so a finite cost exists throughout $e_p<e_{\text{thr}}^{\text{decoy}}$.

The price of the substitution is the block length: the second term of the size rule grows from $0.101\,\numd$ at $\pknown=\USDprob$ to $0.124\,\numd$ at $\mu_1=0.05$, and to $0.207\,\numd$ at the rate $0.42$ certified by a test of $10^{11}$ pulses.  Correctness is not affected.  The set size $\cthr$ is fixed from the single-photon rate $\USDprob$, and the honest rate $\pUSDsift$, $0.2934$ at $\eta\mu_1=0.005$, lies slightly above it, so the honest abort probability can only fall.

\subsection{Numerical Evidence}
\label{sec:decoy_numerics}
We validate the decoy-layer analysis by a photon-number-level Monte-Carlo simulation of the WCP source through a pure-loss channel of transmittance $\eta$, simulating side by side the \emph{honest} receiver (the optimal $N$-copy unambiguous measurement of Eq.~\eqref{eq:bN} on the delivered photons, as in Step~\ref{QSPIR:initial}) and the strongest adversary charged by Lemma~\ref{lem:beff_plus}. Errors enter through the intrinsic misalignment (detector) error rate $e_{\text{det}}$: each \emph{detected} signal yields the wrong bit with probability $e_{\text{det}}$, independently of the channel loss, so the error-weighted gain follows the standard decoy channel model $E_{X,\mu}Q_{X,\mu}=e_0Y_{X,0}+e_{\text{det}}\,(1-e^{-\eta\mu})$~\cite{Ma_Qi_Zhao_Lo_decoy}, here with vanishing dark-count yield $Y_{X,0}=0$ (and $e_0=\tfrac12$). With $Y_{X,0}=0$, $e_{\text{det}}$ is simply the QBER of the detected signals and plays the role of the channel error rate $e=\ebit=\ephase$ of the symmetric single-photon analysis; we take $e_{\text{det}}=0.01$ in the Monte-Carlo simulations, whereas the block lengths of Section~\ref{sec:photon_tagging} are for a noiseless channel. The simulation methodology, together with the validation of the rate formulas of Lemmas~\ref{lem:beff} and~\ref{lem:beff_plus} (the latter is reached by the worst-case receiver of our simulation model, which collects every emitted photon; this shows that the bound cannot be improved within the model, not that such a receiver is physical), the scaling law of Eq.~\eqref{eq:Delta_plus_scaling}, and the binomial law of the per-block count, which fixes the honest abort probability (variance ratio $1.000\pm0.004$), is collected in Appendix~\ref{app:decoy_numerics_validation}. The operating point itself follows from Eq.~\eqref{eq:mu1_constraint}; no simulation is needed for it.

\section{Conclusion}
\label{sec:conclusion}

This paper proposes a one-way, \PM quantum protocol that realizes SPIR between a single client and a single classical database, with statistical security guarantees. When the client requests one bit of the database, our scheme delivers it without leaking extra information to the client, while keeping the index of the requested bit hidden from the server.

In particular, the construction combines a QKD-style transmission phase (to ensure security against eavesdropping), unambiguous state discrimination (to create client-side asymmetry in which bits are known), and a random blocking in which the client announces the locations he determined in one block and the server applies a short error-correcting hash and a random parity to those locations of every block, which ensures that the client can decode essentially only one requested item while learning at most $\SecretB$ of additional database information, together with explicit parameter trade-offs and a key cost linear in the database size.

Security against an all-powerful eavesdropper is obtained by adopting standard QKD steps (parameter estimation, reconciliation, and privacy amplification) to bound the adversary's information, and by one-time-padding the database and the shift with the generated secret key. The client's privacy rests on his private, uniform choice of block, which makes all blocks look alike to a server who prepares the prescribed states [assumption (P)], and on the multi-click count, which limits what she gains from pulses of several photons. Database privacy, by contrast, rests on the asymmetric model of Section~\ref{subsec:objectives}, where Eve is unbounded, while the parties are NISQ-limited. Our protocol makes constructive use of it: the client's inability to store photons forces him to measure every qubit on arrival, and the measurement test of Step~\ref{QSPIR:usd_test} then confines him to the unambiguous-discrimination rate on which the size of the announced set rests.

The analysis makes explicit the trade-offs among database size $n$, block size $\numd$, USD success probability $\USDprob$, the channel bit-flip and phase-flip error rates $(e_b, e_p)$, the size $\cthr$ of the announced set, the number of check symbols $\hmarg$, and the QKD-key overhead $\nums$, which is linear in $\num$.

The joint permutation of Step~\ref{QSPIR:shuffle} uses a storage-model string commitment, with an explicit statistical error and an auxiliary-transmission cost.  That cost is not small: Appendix~\ref{app:storage_commitment} estimates it at more than an order of magnitude above the data transmission, and the overheads of Section~\ref{sec:numerics} do not include it.

By choosing the size of the announced set according to Eq.~\eqref{eq:size_rule}, the server ensures that the client gains negligible information on other bits, for every set the client may announce, and at a post-USD error rate of up to a few percent.  The binomial law of the conclusive pattern assumed in the analysis is confirmed by Monte-Carlo simulation, and an explicit attack shows that the size rule is nearly tight.  Four limitations remain.  Index privacy is proved under assumption (P), for a server who prepares the prescribed states and confirms the client's guesses truthfully: exactly for one photon per position, and up to a small multi-photon term over a noiseless channel.  A server who tampers with her states, for example by small rotations hidden in the channel noise, or who falsely confirms guesses, is not covered, and for weak coherent pulses index privacy rests on the laser assumption (L) and on simulations.  And a client with long-term quantum memory who measures after the announcements of Step~\ref{QSPIR:final_step}, or jointly before them at the present set sizes, is not covered by our bounds (Appendix~\ref{sec:Bob_delayed_measurement}), so the no-long-term-quantum-memory assumption is load-bearing.  The certificate of the client's measurement (Appendix~\ref{app:usd_certification}) is stated for a noiseless channel and costs a test sample of about $2n_\text{t}$ transmitted qubits, with $n_\text{t}$ between $10^{8}$ and $10^{10}$ for a slack of $0.09$ to $0.02$ in the certified rate; together with the commitment of the previous paragraph, it dominates the quantum communication.  Finally, with weak coherent pulses the signal intensity must stay about ten times below its QKD value, which costs the same factor in throughput (Section~\ref{sec:decoy_state}).

In summary, we proposed a single-database quantum SPIR with statistical security and one-way quantum communication, for parties who cannot store quantum states, a stricter form of the restriction in the bounded- and noisy-storage models~\cite{DFSS_BQSM, WST_noisy_storage}.  It uses standard BB84/QKD-style \PM components plus a passive three-outcome receiver, rather than the multiple non-colluding servers or computational assumptions of previous proposals and experiments~\cite{QSPIR_theory1, QSPIR_multiple_database1, Wang_SPIR_experiment, KO97_PIR}.  Our protocol is compatible with current BB84-style QKD setups.  Three questions remain open: a certificate of the client's measurement that tolerates channel noise, a cheaper joint permutation than the storage-model commitment of Step~\ref{QSPIR:shuffle}, and index privacy against a server who departs from assumption (P).

\appendix

\section{Unambiguous State Discrimination via Generalized Measurement}
\label{app:USD_POVM}

This appendix constructs the optimal symmetric USD POVM whose inconclusive outcome is information-free, as claimed in Section~\ref{subsec:soft_info}.

Here, we demonstrate that if Bob performs a generalized measurement, he can unambiguously identify the state sent by Alice with some probability, and if he fails to do so, he cannot gain any partial information on the initial state.

Without loss of generality, suppose Bob guessed the set $S_0=\{\ket{0},\ket{+}\}$ and the guess is correct, as it is for every position retained in Step~\ref{QSPIR:set_disclosure}; the case of $S_1$ follows by symmetry.  Bob defines generalized POVM elements $P_p, P_q, P_r$ on the qubit plus one ancilla level $\ket{2}$, where $P_p+P_q+P_r=I_3$, with

\begin{equation}
\begin{split}
    \ket{p}&=c_1\ket{1}+c_2\ket{2}\\
    \ket{q}&=c_1\ket{-}+c_2\ket{2}\\
\end{split}
\end{equation}

and $P_p=\ket{p}\bra{p}$, $P_q=\ket{q}\bra{q}$, $P_r=I_3-P_p-P_q$.  Since $\braket{0}{p}=0$ and $\braket{+}{q}=0$, a click on $P_p$ excludes $\ket{0}$ and a click on $P_q$ excludes $\ket{+}$.  Normalization, $c_1^2+c_2^2=1$, and orthogonality, $\braket{p}{q}=-c_1^2/\sqrt{2}+c_2^2=0$, give $c_1^2=2-\sqrt{2}$ and $c_2^2=\sqrt{2}-1$.  Then

\begin{gather}
    P_p=\left(
\begin{array}{ccc}
 0 & 0 & 0 \\
 0 & c_1^2 & c_1c_2 \\
 0 & c_1c_2 & c_2^2 \\
\end{array}
\right),\nonumber\\
P_q=\left(
\begin{array}{ccc}
 \frac{c_1^2}{2} & -\frac{c_1^2}{2}
   & \frac{c_1c_2}{\sqrt{2}} \\
 -\frac{c_1^2}{2} & \frac{c_1^2}{2}
   & -\frac{c_1c_2}{\sqrt{2}} \\
 \frac{c_1c_2}{\sqrt{2}} & -\frac{c_1c_2}{\sqrt{2}} & c_2^2 \\
\end{array}
\right),
\end{gather}
\begin{equation}
P_r=\left(
\begin{array}{ccc}
 1-\frac{c_1^2}{2} &
   \frac{c_1^2}{2} & -\frac{c_1c_2}{\sqrt{2}} \\
 \frac{c_1^2}{2} & 1-\frac{3
   c_1^2}{2} & \frac{(\sqrt{2}-2)c_1c_2}{2} \\
 -\frac{c_1c_2}{\sqrt{2}} & \frac{(\sqrt{2}-2)c_1c_2}{2} & 1-2 c_2^2 \\
\end{array}
\right)
\end{equation}

Here $c_1c_2=\sqrt{3\sqrt{2}-4}\approx0.493$, so $P_p$, $P_q$ and $P_r=\ket{r}\bra{r}$ are rank-one orthogonal projectors with $P_p+P_q+P_r=I_3$.  Restricted to the qubit, $P_p$ and $P_q$ reduce to the elements $(2-\sqrt{2})\ket{1}\bra{1}$ and $(2-\sqrt{2})\ket{-}\bra{-}$ of Step~\ref{QSPIR:initial}.  Physically, $\ket{2}$ is an ancillary optical mode and $\{P_p,P_q,P_r\}$ is a projective measurement of one photon over three modes: passive linear optics routes $\ket{p}$, $\ket{q}$ and $\ket{r}$ to three output ports, each watched by one detector.  On the qubit, $P_r$ reduces to the inconclusive element $2(\sqrt2-1)\ket{\phi_0}\bra{\phi_0}$ of Step~\ref{QSPIR:initial}.

For a correct guess, the state Bob receives is the mixed state,
\begin{equation}
\begin{split}
    \rho_{S_0} &= \frac{1}{2}(\ket{0}\bra{0}+\ket{+}\bra{+})\\
    &=\left(
\begin{array}{ccc}
 \frac{3}{4} & \frac{1}{4} &
   0 \\
 \frac{1}{4} & \frac{1}{4} &
   0 \\
 0 & 0 & 0 \\
\end{array}
\right)
\end{split}
\end{equation}

If Bob's outcome is $P_p$, he concludes unambiguously that Alice sent $\ket{+}$; if it is $P_q$, that she sent $\ket{0}$.  Each happens with probability $1-1/\sqrt{2}\approx0.293$ for the corresponding state.  An inconclusive result occurs if the outcome is $P_r$.

It is easy to check that for $\rho_{0}=\ket{0}\bra{0}$ and $\rho_{+}=\ket{+}\bra{+}$, the probability of projecting onto $\ket{r}$ is the same; mathematically,

\begin{equation}
    \mathrm{Tr}[\rho_{0}P_r]=\mathrm{Tr}[\rho_{+}P_r]=\mathrm{Tr}[\rho_{S_0}P_r]=\frac{1}{\sqrt{2}}
\end{equation}

Hence, if the non-deterministic result is obtained, Bob cannot gain any partial information on the initial state.  On a wrong guess the outcome does carry information on the state, but such positions are discarded in Step~\ref{QSPIR:set_disclosure}.

\section{Depolarizing and Dephasing Noise}
\label{sec:Depolarizing_and_Dephasing_Noise}

A quantum channel can be described by a CPTP map, which can be written in many different representations. For example, in the Kraus representation, a channel on a state $\rho$, $\mathcal{E}(\rho)$, can be written as:
\begin{equation}
    \mathcal{E}(\rho)=\sum_l A_l\rho A_l^\dagger
\end{equation}
where the $A_l$ are called Kraus operators and satisfy the trace-preservation condition
\begin{equation}
    \sum_lA_l^\dagger A_l=I .
\end{equation}
Here, we consider the two most common types of errors, the depolarizing noise and the dephasing noise, and we explore their effect on the error rate of USD.

Recall that for a state $\rho$ and a projective operator $P$, the probability of projecting $\rho$ onto $P$ is given by
\begin{equation}
    p_{\text{projection}}(\rho, P)=\mathrm{Tr}[\rho P]
\end{equation}

\subsection{Depolarizing Noise}
The general form of depolarizing noise with depolarizing probability $p$ on a $d$-level quantum system in Kraus representation is:
\begin{equation}
    \mathcal{E}^{\text{depolar}}_{p}(\rho)=(1-p)\rho+p\frac{I}{d}
\end{equation}

Consider the retained positions, where Bob's set guess was correct, and the measurement of Step~\ref{QSPIR:initial}.  Each conclusive element is $(2-\sqrt{2})$ times a projector onto a state of the \emph{complementary} basis.  For instance, for $S_0$ the element $(2-\sqrt{2})\ket{1}\bra{1}$ announces $\ket{+}$.  A state prepared in one basis and measured in the complementary one gives either result with probability $1/2$, and depolarization does not change this:
\begin{equation}
    \mathrm{Tr}[\mathcal{E}^{\text{depolar}}_p(\ket{+}\bra{+})\ket{1}\bra{1}]=1/2 .
\end{equation}
Hence the correct-conclusive probability stays $(2-\sqrt{2})/2=\USDprob$ for every prepared state.  What does change is the false-positive rate.  For example, if Alice prepares $\ket{0}$ and Bob guessed $S_0$,
\begin{equation}
\mathrm{Tr}[\mathcal{E}^{\text{depolar}}_p(\ket{0}\bra{0})\,(2-\sqrt{2})\ket{1}\bra{1}]=(2-\sqrt{2})\,\frac{p}{2}=\USDprob\, p ,
\end{equation}
so Bob wrongly announces $\ket{+}$.  The same holds for each of the four states.  Hence, among retained positions,
\begin{equation}
    \text{False Positive}=\USDprob\,p,\quad
    \Pr[\text{conclusive}]=\USDprob(1+p),
\end{equation}
and the false-positive error rate is
\begin{equation}
    e_{\text{depolar}}=\frac{\USDprob\, p}{\USDprob+\USDprob\, p}=\frac{p}{1+p}
\end{equation}
This is the post-USD error rate $\tilde{\ebit}$ of Step~\ref{QSPIR:code_sizing}.  The channel itself flips a bit with probability $\ebit=\ephase=p/2$, so $\tilde{\ebit}=2\ebit/(1+2\ebit)$ is about twice the channel error rate.  The inconclusive element is symmetric between the two states of each set, and depolarization preserves this symmetry, so the inconclusive posterior remains exactly $1/2$ (Section~\ref{subsec:soft_info}).

\subsection{Dephasing Noise}
The general form of dephasing noise in the $\sigma_X$, $\sigma_Y$, and $\sigma_Z$ directions, with probabilities $p_x, p_y, p_z$, on a qubit is:
\begin{equation}
    \mathcal{E}^{\text{dephase}}_{p_x, p_y, p_z}(\rho)=(1-p_x-p_y-p_z)\rho+\sum_{D=X, Y, Z}p_D\sigma_D\rho\sigma_D^\dagger
\end{equation}

As before, a state measured in the complementary basis gives either result with probability $1/2$ under any Pauli channel, so the correct-conclusive probability stays $\USDprob$.  The false positives come from a Pauli error that flips the prepared state onto the state it is tested against:
\begin{equation}
\mathrm{Tr}[\mathcal{E}^{\text{dephase}}_{p_x, p_y, p_z}(\ket{0}\bra{0})\ket{1}\bra{1}]=p_x+p_y
\end{equation}
\begin{equation}
\mathrm{Tr}[\mathcal{E}^{\text{dephase}}_{p_x, p_y, p_z}(\ket{1}\bra{1})\ket{0}\bra{0}]=p_x+p_y
\end{equation}
\begin{equation}
\mathrm{Tr}[\mathcal{E}^{\text{dephase}}_{p_x, p_y, p_z}(\ket{+}\bra{+})\ket{-}\bra{-}]=p_y+p_z
\end{equation}
\begin{equation}
\mathrm{Tr}[\mathcal{E}^{\text{dephase}}_{p_x, p_y, p_z}(\ket{-}\bra{-})\ket{+}\bra{+}]=p_y+p_z
\end{equation}

Each is weighted by $2-\sqrt{2}=2\USDprob$, and the four prepared states are equally likely.  Hence, the probability of a state measurement being falsely considered a successful USD is:
\begin{equation}
\begin{split}
    \text{False Positive}&=\frac{1}{4}\times2\USDprob\times\left(2p_x+2p_y+2p_y+2p_z\right)\\
    &=\USDprob\,(p_x+2p_y+p_z)
    \end{split}
\end{equation}

Hence, the false-positive error rate is:
\begin{equation}
    \begin{split}
    e_{\text{dephase}}&=\frac{\USDprob\,(p_x+2p_y+p_z)}{\USDprob+\USDprob\,(p_x+2p_y+p_z)}\\
    &=\frac{p_x+2p_y+p_z}{1+p_x+2p_y+p_z}
    \end{split}
\end{equation}
This is $\tilde{\ebit}$ for this channel.  With the channel's bit-flip and phase-flip rates $\ebit=p_x+p_y$ and $\ephase=p_y+p_z$ it reads $\tilde{\ebit}=(\ebit+\ephase)/(1+\ebit+\ephase)$.  This is how Step~\ref{QSPIR:code_sizing} obtains $\tilde{\ebit}$ from the estimates of Step~\ref{QSPIR:rate_estimation}, taken at their upper confidence limits, for the margin $\gamma$ covers only the count of errors inside $\aset$.  Unlike depolarization, dephasing with $p_x\ne p_z$ breaks the symmetry between the $\sigma_Z$ and $\sigma_X$ bases, so the inconclusive posterior deviates slightly from $1/2$ (e.g.\ $0.504$ for $p_z=0.02$, $p_x=p_y=0$).  The error rate then also differs between claims, $2\ebit/(1+2\ebit)$ among claims of an $X$-basis state and $2\ephase/(1+2\ephase)$ among claims of a $Z$-basis state, so the check symbols must be sized with the larger of the two, and the erasure hypothesis of Lemma~\ref{lem:seeded_parity} holds only approximately (Section~\ref{subsec:soft_info}).

\section{Storage-Model String Commitment for the Joint Permutation}
\label{app:storage_commitment}

Step~\ref{QSPIR:shuffle} uses a commitment to a classical string, with Alice as committer and Bob as verifier.  We instantiate it through weak string erasure (WSE) and classical processing~\cite{Koenig_Wehner_Wullschleger_storage}.  All auxiliary quantum communication is from Alice to Bob.  The construction is statistical under a specified storage restriction; it does not assume unrestricted, unconditionally secure quantum bit commitment, which is excluded by Refs.~\cite{Mayers_no_BC, Lo_Chau_no_BC}.

\subsection{Storage assumption and auxiliary parameters}

First consider noiseless single-photon auxiliary transmissions.  Let $L_{\mathrm{com}}$ be their length per attempt.  After a prescribed delay $\Delta t_{\mathrm{com}}$, dishonest Bob can retain at most $q_{\mathrm{com}}$ qubits in total, including ancillary registers, while retaining unlimited classical information.  Honest parties measure immediately and store only classical records.  The strict no-memory model sets $q_{\mathrm{com}}=0$.

Choose a WSE min-entropy rate $\lambda_{\mathrm{com}}>0$ and error $\varepsilon_{\mathrm{WSE}}$ certified by the storage-model uncertainty bound in Ref.~\cite{Koenig_Wehner_Wullschleger_storage}, Theorem~III.2.  For a bounded memory, a positive asymptotic rate is available when $q_{\mathrm{com}}/L_{\mathrm{com}}<1/2-\zeta$ for some fixed $\zeta>0$; finite lengths must satisfy that theorem's error bound.  The entropy guarantee includes Bob's measurement record and the announced auxiliary bases.

Before the attempt, the parties fix a binary linear code $\mathcal D$ with parameters $[L_{\mathrm{com}},K_{\mathrm{com}},d_{\mathrm{com}}]$, a syndrome map $\operatorname{Syn}_{\mathcal D}$, and a two-universal hash family $\{f_u\}$ with $\ell_P$ output bits.  For an error parameter $\varepsilon_*>0$, choose
\begin{subequations}
\label{eq:commitment_parameters}
\begin{align}
 d_{\mathrm{com}}&\ge\left\lceil2\log_2(1/\varepsilon_*)\right\rceil,\\
 \ell_P&\le\lambda_{\mathrm{com}}L_{\mathrm{com}}
       -(L_{\mathrm{com}}-K_{\mathrm{com}})-2\log_2(1/\varepsilon_*).
\end{align}
\end{subequations}
Here $\ell_P=\lceil\log_2((\numd\num)!)\rceil$, with the right-hand side rounded down when choosing an integer length.  The code $\mathcal D$ is a separate classical code, unrelated to the check matrix $\dstr{H}$ of Step~\ref{QSPIR:final_step}.

\subsection{Commitment and opening}

The following is the WSE-based construction of Ref.~\cite{Koenig_Wehner_Wullschleger_storage}, Protocols~1, 2a and~2b, with the standard masking step for a chosen input string.  The invocation uses fresh auxiliary states and local randomness, after the retained data list has been fixed.

\begin{enumerate}
 \item Alice samples independent uniform $X,\Theta\in\{0,1\}^{L_{\mathrm{com}}}$ and sends the auxiliary states $H^{\Theta_j}\ket{X_j}$, with $H$ the Hadamard gate.  Bob measures in independent uniform bases $\widehat\Theta_j$ and records $\widehat X_j$.

 \item After the delay $\Delta t_{\mathrm{com}}$, Alice announces $\Theta$.  Bob privately forms $I=\{j:\Theta_j=\widehat\Theta_j\}$ and keeps $\widehat X_I$.  He reveals neither $I$ nor his measurement outcomes.  In the noiseless honest execution, $\widehat X_I=X_I$.

 \item To commit to $y=\operatorname{bin}_{\ell_P}(t_a)$, Alice samples a hash seed $u$ and sends
 \begin{equation}
     w=\operatorname{Syn}_{\mathcal D}(X),\qquad
     u,\qquad e=y\oplus f_u(X).
     \label{eq:commitment_mask}
 \end{equation}
 Bob checks the message formats and that $u$ belongs to the agreed hash family.  The commitment phase is then complete.  Only now does Bob generate his fresh rank $t_b$ (Step~\ref{QSPIR:shuffle}).

 \item Alice opens by revealing $X$.  Bob verifies the syndrome and $X_I=\widehat X_I$, recovers $y=e\oplus f_u(X)$, and rejects unless its integer value is in $\{0,\ldots,M_P-1\}$.  Any failure aborts the session.
\end{enumerate}

Theorem~IV.2 of Ref.~\cite{Koenig_Wehner_Wullschleger_storage} gives error
\begin{equation}
    \varepsilon_{\mathrm{com}}=2\varepsilon_{\mathrm{WSE}}+\varepsilon_*
    \label{eq:commitment_error}
\end{equation}
under Eq.~\eqref{eq:commitment_parameters}.  Its hiding condition bounds the distance of the random mask from uniform given dishonest Bob's view.  Its binding condition supplies a fixed classical committed string in an ideal state close to the real state, and bounds the probability of accepting a different opening.  These are the properties needed for the chosen-input masking step above.  They do not follow merely from excluding one particular entanglement attack.

\subsection{Uniform sampling, aborts, and sequential use}

With ideal commitment, binding fixes Alice's valid rank before honest Bob samples his uniform rank.  For every fixed $t_a$ and $r\in\{0,\ldots,M_P-1\}$,
\begin{equation}
 \Pr[(t_a+t_b)\bmod M_P=r]=\frac{1}{M_P}.
\end{equation}
If Bob is dishonest, hiding makes his choice independent of honest Alice's uniform rank.  Addition modulo $M_P$ again preserves uniformity.  Applying the fixed bijection $\operatorname{unrank}_N$ therefore produces a uniform candidate permutation.  This argument allows the dishonest contribution to depend on the earlier data record.

An opening can nevertheless be withheld after $t_b$ is known.  Accordingly, the security statement includes an abort output.  For any event $E\subseteq S(N)$ fixed from information preceding the honest fresh contribution, the ideal event bound becomes, for the real sampler,
\begin{equation}
 \Pr[\mathrm{accept}\ \land\ P\in E]
 \le \frac{|E|}{N!}+\epsilon_{\mathrm{rnd}}.
 \label{eq:permutation_abort_bound}
\end{equation}
For one invocation, $\epsilon_{\mathrm{rnd}}=2\varepsilon_{\mathrm{com}}$ conservatively covers the ideal-state replacement and binding error, as well as hiding and honest correctness.  This bound does not imply that $P$ is uniform conditioned on acceptance: a dishonest committer could open only for even permutations.  The bound extends from events to functions $f$ of $P$ and the preceding information with values in $[0,1]$: integrating it over the level sets of $f$ gives $\EX[f\,1_{\mathrm{accept}}]\le\EX_{\mathrm{unif}}f+\epsilon_{\mathrm{rnd}}$.  The proofs apply it to quantities capped at $1$, such as $\min(1,W)$.

The invocation uses its own auxiliary records, and no other quantum subprotocol is interleaved with it.  To use the error bound inside SPIR, it must hold with the preceding transcript and permitted retained state as auxiliary information.  This is the sequential security interface required here; standalone security alone must not be treated as an unrestricted composition theorem~\cite{Wehner_Wullschleger_composition}.  For the strict no-memory case, earlier participant records are classical at the storage checkpoints.  Extensions with nonzero retained quantum memory must also satisfy the chosen composition framework's memory conditions.

Under this interface, the protocol invokes the sampler once, for the blocking of Step~\ref{QSPIR:shuffle}, and the invocation contributes at most $\epsilon_{\mathrm{rnd}}$ to an unconditioned statistical security budget.  There are no further attempts.  A failed opening terminates the session.

\subsection{Resource cost and nonideal implementations}

In the noiseless construction, the invocation consumes $L_{\mathrm{com}}$ additional transmitted qubits and one delay of duration $\Delta t_{\mathrm{com}}$.  These costs are additional to the data acquisition and the QKD-key overhead $\nums$.  They are not small.  Bob learns about half of $X$ by measuring every auxiliary qubit in one basis, so $\lambda_{\mathrm{com}}\le\tfrac12$ and $L_{\mathrm{com}}>2\ell_P$.  At $\num=10^{4}$ and $\numd=980$ the rank has $\ell_P\approx2.1\times10^{8}$ bits, so the commitment needs more than $4.2\times10^{8}$ auxiliary qubits, about $22$ times the $2\times10^{7}$ data-acquisition positions of Step~\ref{QSPIR:initial} without the test sample.  The rank length $\ell_P$ is determined by the permutation space and cannot be replaced by a short pseudorandom seed while retaining uniformity on $S(N)$.

For noisy or lossy auxiliary transmissions, the equality test above must be replaced by a storage-model commitment with a corresponding finite-size correctness and security analysis, such as the treatment in Ref.~\cite{Ng_noisy_storage_commitment}.  Its source, detection, error, and memory parameters determine the admissible code and commitment length.  A weak-coherent-pulse implementation requires its own multiphoton analysis for this auxiliary transmission.  The data-mode decoy estimates of Section~\ref{sec:decoy_state} do not supply that analysis, and the numerical overheads in this paper do not include the commitment cost.

\section{Proof of the Database-Privacy Bound}
\label{app:simulator_proof}
This appendix proves Theorem~\ref{thm:success_protocol}(ii) and gives the parameters in closed form.  Why is $\SecretB$ defined through a single database bit of the generator's choice?  A cheating Bob has no fixed index, so ``the other $\num-1$ bits'' means nothing until someone names his bit.  The generator, called the simulator below, does.  Recall that $K_j$ is the set of locations of block $j$ known to Bob, $a_j=|\aset\setminus K_j|$, $x_j=\tfrac12\min(1,2^{\hmarg-a_j})$ and $o_{jj'}=|K_j\cap K_{j'}|$.

\begin{proposition}[Simulator]
\label{prop:simulator}
Under the erasure hypothesis of Lemma~\ref{lem:seeded_parity}, Bob's final state is within trace distance $\EX[W]$ of a state generated from one database bit, with $W=\sum_jx_j-\max_jx_j$.
\end{proposition}
\begin{proof}
The simulator plays honest Alice with a dummy database up to and including $\dstr{H}$, $\dstr{z}$ and all check symbols, none of which depends on $\dstr{d}$.  It runs Bob's strategy internally, so it knows his classical record and hence every $x_j$, sets $j^{*}=\arg\max_jx_j$ and asks for the database bit that the announced shift places on block $j^{*}$.  In the last message it outputs the true bit at that position and fresh uniform bits elsewhere.

Fix $\dstr{H}$, the sets $K_j$ and the set $\aset$, which are all determined before $\dstr{z}$ is drawn, and let $\mathcal{B}$ be the set of blocks with $\dstr{z}_{E_j}$ in the row space of $\dstr{H}_{E_j}$.  By Lemma~\ref{lem:seeded_parity}, the parity of a block in $\mathcal{B}$ is a function of Bob's view, while the parities outside $\mathcal{B}$ are uniform and independent of that view and of each other, the unknown values of distinct blocks being independent.  Bob knows $\keydb$, so the last message is a bijective image of these parities.  The real and the simulated state therefore differ only on $\mathcal{B}\setminus\{j^{*}\}$, where the real bits are determined and the simulated ones uniform.  Their distance is $1-2^{-|\mathcal{B}\setminus\{j^{*}\}|}\le\tfrac12|\mathcal{B}\setminus\{j^{*}\}|$.  Over the uniform seed, $\Pr[j\in\mathcal{B}]=2^{\operatorname{rank}\dstr{H}_{E_j}-a_j}\le2x_j$, so the expected distance is at most $\sum_{j\ne j^{*}}x_j=W$.
\end{proof}

\subsection{How the leak depends on the missed locations}
\label{app:leak_vs_missed}
Before proving Eq.~\eqref{eq:EW_common}, we look at the leakage score.  It depends on block $j$ only through the number $a_j$ of announced locations that Bob misses there, and it has two regimes:
\begin{equation*}
    x_j=\tfrac12\min\big(1,2^{\hmarg-a_j}\big)=
    \begin{cases}
    \tfrac12, & a_j\le\hmarg,\\[2pt]
    \tfrac12\,2^{-(a_j-\hmarg)}, & a_j>\hmarg.
    \end{cases}
\end{equation*}
Why this shape?  By Lemma~\ref{lem:seeded_parity}, Bob learns the parity of block $j$ only if the seed $\dstr{z}$, restricted to the $a_j$ missed locations, lies in the row space of the check matrix restricted to them.  This space holds at most $2^{\min(\hmarg,a_j)}$ of the $2^{a_j}$ restrictions.  If Bob misses at most $\hmarg$ locations, the $\hmarg$ check symbols can fill all of them, and the bit may leak completely.  The factor $\tfrac12$ is the trace distance between a known bit and a uniform one.  Beyond $\hmarg$, every further missed location doubles the number of restrictions but not the row space, and so halves the leak.  For example, an honest Bob misses about $(1-\USDprob)\cthr$ locations, many more than $\hmarg$, in every block but his own, and none in his own block, whose bit is the one he is entitled to.

Three properties carry the proof below.  (i)~$x_j$ is non-increasing in $a_j$, so a lower bound on the misses of a block is an upper bound on its leak.  (ii)~$x_j$ decays exponentially in $a_j$, so the bound $\max(a_j,a_{j'})\ge(\cthr-o_{jj'})/2$ of Lemma~\ref{lem:pair_bound} turns into the factor $2^{-(\cthr-o_{jj'})/2}$ of Eq.~\eqref{eq:pair_leak}.  Two further announced locations thus halve the bound, which is why the size rule of Eq.~\eqref{eq:size_rule} doubles every term but $\cthr$.  (iii)~The cap at $\tfrac12$ binds only for $a_j\le\hmarg$.  The second line of Eq.~\eqref{eq:pair_leak} drops it, which overestimates the leak only for pairs with $o_{jj'}>\cthr-2\hmarg$, the pairs that one set could serve; the exact evaluation of Fig.~\ref{fig:CS_window} and Table~\ref{tab:common_set} keeps it.

\subsection{One set cannot serve two blocks}
\label{app:common_set}
Next, we prove Eq.~\eqref{eq:EW_common}.  The point is that the bound holds for every set $\aset$ of $\cthr$ locations, however Bob chooses it, so no union over sets is needed.

\begin{lemma}[Pair bound]
\label{lem:pair_bound}
For all blocks $j\ne j'$ and every $\aset$ with $|\aset|=\cthr$,
$\max(a_j,a_{j'})\ge\tfrac12\big(\cthr-o_{jj'}\big)$.
\end{lemma}
\begin{proof}
A location of $\aset$ outside $K_j\cap K_{j'}$ is missed in at least one of the two blocks.  Hence $a_j+a_{j'}\ge\cthr-|\aset\cap K_j\cap K_{j'}|\ge\cthr-o_{jj'}$.
\end{proof}

\begin{lemma}[Overlap of two random blocks]
\label{lem:overlap_mgf}
Let each of $N$ positions carry a flag, set independently with probability $q_i$ for position $i$, and write $\bar q=\frac1N\sum_iq_i$.  Fill blocks of $\numd$ locations with positions drawn uniformly without replacement, independently of the flags.  Let $K_j$ be the set of flagged locations of block $j$ and $o_{jj'}=|K_j\cap K_{j'}|$.  Then, for $j\ne j'$ and every $\vartheta\ge0$,
\begin{equation*}
    \EX\,e^{\vartheta o_{jj'}}\;\le\;\big[1+(e^{\vartheta}-1)\,\bar q^{2}\big]^{\numd}.
\end{equation*}
\end{lemma}

What does the lemma say?  A flag marks a position whose bit Bob knows.  The right side is the moment generating function of $\mathrm{Bin}(\numd,\bar q^{2})$, which is the law of $o_{jj'}$ when every $q_i$ equals $\bar q$, as for the measurement of Step~\ref{QSPIR:initial}.  A Bob who measures differently may know some positions well and others not at all, and Appendix~\ref{app:usd_certification} bounds only his mean rate $\bar q$.  The lemma shows that such an uneven spread cannot raise the moment, as long as the blocks are drawn independently of what he knows.

\begin{proof}
Let $P_l$ and $P'_l$ be the positions at location $l$ of blocks $j$ and $j'$, and write $P,P'$ for the two lists.  These $2\numd$ positions are distinct.  We first fix them and then average over the draw.

\emph{Step 1: fixed draw.}  Location $l$ counts in $o_{jj'}$ exactly when $P_l$ and $P'_l$ are both flagged, which happens with probability $q_{P_l}q_{P'_l}$.  Different locations involve different positions, and the flags are independent of each other and of the draw.  So $o_{jj'}$ is a sum of $\numd$ independent Bernoulli variables, and
\begin{equation*}
    \EX\big[e^{\vartheta o_{jj'}}\,\big|\,P,P'\big]=\prod_{l=1}^{\numd}\big[1+(e^{\vartheta}-1)\,q_{P_l}q_{P'_l}\big].
\end{equation*}

\emph{Step 2: average over the draw.}  Were the positions drawn with replacement, the $\numd$ factors would be independent, and the average of their product would equal the product of their averages.  Drawing without replacement can only lower that average: a location that receives two positions of large $q$ leaves fewer of them for the other locations.  Formally, values drawn without replacement are negatively associated, and for negatively associated variables the average of a product of positive non-decreasing functions of disjoint sets of the variables is at most the product of the averages~\cite{JoagDev_Proschan_NA}.  The factor of location $l$ is positive, depends on $q_{P_l}$ and $q_{P'_l}$ only, and is non-decreasing in both since $\vartheta\ge0$.  Hence
\begin{equation*}
    \EX\,e^{\vartheta o_{jj'}}\le\prod_{l=1}^{\numd}\big[1+(e^{\vartheta}-1)\,\EX(q_{P_l}q_{P'_l})\big].
\end{equation*}

\emph{Step 3: mean of $q_{P_l}q_{P'_l}$.}  The pair $(P_l,P'_l)$ is a uniformly random ordered pair of distinct positions, so
\begin{equation*}
    \begin{split}
    \EX(q_{P_l}q_{P'_l})&=\frac{\big(\sum_iq_i\big)^{2}-\sum_iq_i^{2}}{N(N-1)}\\
    &=\bar q^{2}-\frac{v}{N-1}\;\le\;\bar q^{2},
    \end{split}
\end{equation*}
where $v=\frac1N\sum_i(q_i-\bar q)^{2}$ is the variance of the $q_i$.  An uneven spread thus lowers the mean.  As $e^{\vartheta}-1\ge0$, inserting this into Step~2 proves the lemma.  If every $q_i$ equals $\bar q$, the bound holds with equality.
\end{proof}

\begin{proof}[Proof of Eq.~\eqref{eq:EW_common}]
Order the $x_j$ decreasingly.  Then $W$ is the sum of all but the largest, and $\sum_{j<j'}\min(x_j,x_{j'})$ counts the $r$th largest $r-1$ times, so $W\le\sum_{j<j'}\min(x_j,x_{j'})$.  Since $x_j$ is non-increasing in $a_j$, Lemma~\ref{lem:pair_bound} gives
\begin{equation}
    \begin{split}
    \min(x_j,x_{j'})&\le\tfrac12\min\!\Big(1,2^{\hmarg-(\cthr-o_{jj'})/2}\Big)\\
    &\le\tfrac12\,2^{\,[o_{jj'}-(\cthr-2\hmarg)]/2},
    \end{split}
    \label{eq:pair_leak}
\end{equation}
which no longer depends on $\aset$.  Lemma~\ref{lem:overlap_mgf} with $\vartheta=\tfrac12\ln2$ bounds $\EX\,2^{o_{jj'}/2}$ by $[1+(\sqrt2-1)\pknown^{2}]^{\numd}$, and summing over the $\binom{\num}{2}\le\num^{2}/2$ pairs gives the claim.  Replacing the ideal permutation sampler by the real one costs $\epsilon_{\mathrm{rnd}}$ for its one invocation, since the blocks must be a uniform partition of the retained positions, by Eq.~\eqref{eq:permutation_abort_bound} applied to $\min(1,W)$, which gives the term $\epsilon_{\mathrm{rnd}}$ of Eq.~\eqref{eq:epsB_main}.  For the measurement of Step~\ref{QSPIR:initial} the flags are i.i.d., $o_{jj'}\sim\mathrm{Bin}(\numd,\pknown^{2})$ exactly, and the first line of Eq.~\eqref{eq:pair_leak} can be averaged over this law, with $(\cthr-o_{jj'})/2$ rounded up since the $a_j$ are integers; this exact form is used in Fig.~\ref{fig:CS_window} and in the right half of Table~\ref{tab:common_set}.
\end{proof}

\begin{proof}[Proof of Eq.~\eqref{eq:size_rule}]
Let $\SecretB$ stand for the budget left for $\EX[W]$ once the permutation terms of Eq.~\eqref{eq:epsB_main} are paid.  By Eq.~\eqref{eq:EW_common}, $\EX[W]\le\SecretB$ holds whenever
\begin{equation*}
    \frac{\num^{2}}{4}\;2^{-(\cthr-2\hmarg)/2}\,\Big[1+(\sqrt2-1)\,\pknown^{2}\Big]^{\numd}\;\le\;\SecretB .
\end{equation*}
Both sides are positive and $\log_2$ is increasing, so we may take $\log_2$ of both sides, and the logarithm of a product is the sum of the logarithms of its factors.  Write $h_\text{o}=\log_2[1+(\sqrt2-1)\pknown^{2}]$.  The condition therefore reads
\begin{equation*}
    2\log_2\num-2-\frac{\cthr-2\hmarg}{2}+\numd h_\text{o}\;\le\;-\log_2\frac{1}{\SecretB}.
\end{equation*}
Rearranging the terms gives,
\begin{equation*}
    \cthr\;\ge\;2\hmarg+2\numd h_\text{o}+4\log_2\num+2\log_2\frac{1}{\SecretB}-4 ,
\end{equation*}
which is Eq.~\eqref{eq:size_rule}.
\end{proof}

\subsection{Choosing the parameters}
Can the parameters be chosen without numerics?  Yes.  (a) For a target abort probability $\epsilon_\text{abort}$ take $\CSSprob=\sqrt{\ln(1/\epsilon_\text{abort})/(2\numd)}$ [Theorem~\ref{thm:success_protocol}(i)].  (b) For a target probability $\epsilon_\perp$ of the Chernoff term of Lemma~\ref{lem:decodability} put $\Lambda=\ln(1/\epsilon_\perp)/\cthr$ and
\begin{equation}
    \gamma=\Lambda+\sqrt{\Lambda^{2}+2\tilde{\ebit}\Lambda},
    \label{eq:gamma_closed}
\end{equation}
which suffices since $D(x\Vert y)\ge(x-y)^2/(2x)$ for $x\ge y$; for $\tilde{\ebit}=0$ no conclusive bit is in error, and $\gamma=0$, so that $t=0$ and Lemma~\ref{lem:decodability} would hold even without check symbols.  We keep $\hmarg=\lceil\log_2(1/\ECprob)\rceil=20$ there, so that an error outside the model is still detected, Bob then outputting $\perp$, except with probability $\ECprob$; with $\hmarg=0$ the smallest noiseless block of Table~\ref{tab:common_set} would be $730$ instead of $980$, and $644$ instead of $884$ in the exact evaluation.  (c) Put $h=H_2(\tilde{\ebit}+\gamma)$.  Eq.~\eqref{eq:syndrome_length} and $\lceil x\rceil<x+1$ give $\hmarg\le\cthr h+\log_2(1/\ECprob)+1$.  The smaller choice $\hmarg=\lceil\log_2\mathrm{Vol}(\cthr,t)+\log_2(1/\ECprob)\rceil$ obeys the same bound, since $t\le(\tilde{\ebit}+\gamma)\cthr$ and $\mathrm{Vol}(\cthr,t)\le2^{\cthr H_2(t/\cthr)}\le2^{\cthr h}$ for $\tilde{\ebit}+\gamma\le\tfrac12$~\cite{MacWilliams_Sloane}.  Inserting this bound on $\hmarg$ into Eq.~\eqref{eq:size_rule} and moving the term $2\cthr h$ to the left shows that Eq.~\eqref{eq:size_rule} holds whenever
\begin{equation}
    \begin{split}
    &\cthr(1-2h)-2\numd\log_2\!\big[1+(\sqrt2-1)\pknown^{2}\big]\\
    &\qquad\ge\;4\log_2\num+2\log_2\frac{1}{\SecretB\,\ECprob}-2 .
    \end{split}
    \label{eq:nd_closed}
\end{equation}
The left side grows linearly in $\numd$ once $\USDprob(1-2h)$ exceeds $2\log_2[1+(\sqrt2-1)\pknown^{2}]$, so the smallest block follows by increasing $\numd$.  Table~\ref{tab:common_set} lists the smallest multiple of ten for which Eq.~\eqref{eq:size_rule} itself holds with these choices, slightly below what Eq.~\eqref{eq:nd_closed} demands, together with the exact evaluation, in which $\cthr$ and $t$ sit at the exact binomial quantiles and $\hmarg=\lceil\log_2\mathrm{Vol}(\cthr,t)+\log_2(1/\ECprob)\rceil$.
At a fixed block length the set size has a window.  For $\numd=10^{4}$ the exact evaluation allows $\cthr$ between $1119$ and $2789$ over a noiseless channel and between $2083$ and $2848$ at $\tilde{\ebit}=2\%$, and the window closes between $3\%$ and $4\%$; inside these ranges, isolated values of $\cthr$ at which $t$ or $\hmarg$ jumps by one can fail.  Note that the known-rate enters only through $\pknown^{2}$.  Over a noiseless channel the rule remains feasible up to $\pknown\approx0.51$, which is what makes the certified rate of Appendix~\ref{app:usd_certification} and the decoy-certified rate of Section~\ref{sec:decoy_state} affordable.

\subsection{Classical--Quantum Form And Shannon Corollary}
\label{app:LHL_application}

This subsection restates the exact leakage bound of Lemma~\ref{lem:seeded_parity}, Eq.~\eqref{eq:seeded_leak}, in classical--quantum notation, together with its Shannon-information corollary.  As there, Bob knows the bits at the locations he determined, and the others are erasures ($p_a=\tfrac12$ in the notation of Section~\ref{subsec:soft_info}).

We first fix the classical--quantum state notation, following Refs.~\cite{Renner_thesis, Tomamichel_leftover_hashing}. The joint state of the bits $S$ of a block at the announced locations $\aset$ and Bob's side information $V$ is the cq-state
\begin{equation}
    \rho_{\text{SV}}=\sum_{s\in\{0,1\}^{\cthr}}P_S(s)\,\ket{s}\!\bra{s}_S\otimes\rho_V^{\,s},
    \label{eq:cq_state_block}
\end{equation}
where $P_S(s)$ is the probability that the block carries the physical bit-string $s$ at these locations and $\rho_V^{\,s}$ is the state of $V$ conditioned on $S=s$. Tracing out the block gives the side-information marginal
\begin{equation}
    \rho_V:=\mathrm{Tr}_S\,\rho_{\text{SV}}=\sum_{s\in\{0,1\}^{\cthr}}P_S(s)\,\rho_V^{\,s}.
    \label{eq:cq_marginal_V}
\end{equation}
Applying the seeded parity $f(S)=\dstr{z}\cdot S$ of Theorem~\ref{thm:PA_block} below maps $S$ to the one-bit register $L=f(S)$ and yields the classical--quantum state
\begin{equation}
    \rho_{\text{LV}}=\sum_{s\in\{0,1\}^{\cthr}}P_S(s)\,\ket{f(s)}\!\bra{f(s)}_L\otimes\rho_V^{\,s},
    \label{eq:cq_state_LV}
\end{equation}
whose ideal reference is $\omega_L\otimes\rho_V$ with $\omega_L:=\tfrac12\,\mathbf{1}$ the maximally mixed state of $L$.

\begin{theorem}[Leakage of the block parity]
\label{thm:PA_block}
Let $S\in\{0,1\}^{\cthr}$ be the physical bits of a block at the announced locations $\aset$, of which Bob misses the $a$ locations $E$, let $V$ be Bob's quantum-classical side information, including the check symbols $\sigma=\dstr{H}S$, the matrix $\dstr{H}$ and the seed $\dstr{z}$ of Step~\ref{QSPIR:final_step}, and let $L=\dstr{z}\cdot S$.  Assume the erasure hypothesis: given $V$ without $\sigma$, the bits of $S$ on $E$ are uniform and independent of everything else ($p_a=\tfrac12$).  Then, for every fixed $\dstr{H}$,
\begin{equation}
    \tfrac{1}{2}\big\Vert \rho_{\text{LV}}-\omega_{L}\otimes\rho_{V}\big\Vert_{1}\;=\;\tfrac12\,2^{\operatorname{rank}\dstr{H}_E-a}\;\le\;\tfrac12\min\!\big(1,2^{\hmarg-a}\big).
    \label{eq:Hmin_block}
\end{equation}
The same bound holds when the bits Bob knows are noisy.
\end{theorem}
\begin{proof}
This is Lemma~\ref{lem:seeded_parity} with $E_j=E$ and $a_j=a$: given $V$, the bits on $E$ are uniform on the affine space fixed by the check symbols, the parity is determined by $V$ if $\dstr{z}_E$ lies in the row space of $\dstr{H}_E$ and uniform otherwise, and $\rho_{LV}$ is block diagonal in the seed, so the trace distance is the average of the two cases.

If Bob's conclusive bits are noisy, he holds $S_{\bar E}\oplus\dstr{e}$ instead of the bits $S_{\bar E}$ at the locations $\bar E=\aset\setminus E$ that he knows, with an error pattern $\dstr{e}$ independent of $S$.  Adjoin $\dstr{e}$ to $V$: the enlarged view determines $S_{\bar E}$, so the bound applies to it, and tracing $\dstr{e}$ out cannot increase the trace distance.
\end{proof}

\begin{corollary}[Shannon-information corollary]
\label{cor:Shannon_from_trace}
Let $\hat L$ be the bit Bob computes from his view $V$, and let $\delta\le x_j$ be the trace distance of Eq.~\eqref{eq:seeded_leak}.  By the Alicki--Fannes--Winter continuity of mutual information~\cite{Alicki_Fannes, Winter_continuity} for the one-bit register $L$, and by data processing, $I(L;\hat L)\le I(L;V)$ and
\begin{equation}
    I(L;\hat L)\;\le\;2\delta+(1+\delta)\,H_2\!\left(\frac{\delta}{1+\delta}\right)\;=\;O\!\big(\delta\log_2\tfrac{1}{\delta}\big),
    \label{eq:Shannon_corollary}
\end{equation}
where $H_2(\cdot)$ is the binary entropy.  For an $m$-bit register the first term becomes $2m\delta$, which gives Eq.~\eqref{eq:epsI_epsB} with $m=\num-1$.
\end{corollary}

\section{Certifying Bob's Measurement}
\label{app:usd_certification}
Theorem~\ref{thm:success_protocol} rests on the erasure hypothesis, which holds for the measurement of Step~\ref{QSPIR:initial}.  But what if Bob measures differently?  In this appendix, we show that the test of Step~\ref{QSPIR:usd_test} replaces the hypothesis by a certified rate $\USDprob^\text{cert}$ slightly above $\USDprob$, at an additive cost $\epsilon_\text{cert}$ in $\SecretB$.  We assume the following.
\begin{itemize}
    \item[(A1)] After the waiting time of Step~\ref{QSPIR:initial}, Bob holds only a classical record $y$.
    \item[(A2)] Bob measures the qubits one at a time, possibly adaptively: qubit $i$ is measured by a POVM $\{\Pi^{(i)}_{y_i|y_{<i}}\}$ chosen from his earlier outcomes.  Note that the no-long-term-memory assumption does not imply (A2); joint measurements within the coherence time are not covered.
    \item[(A3)] Alice is honest and the permutation sampler is ideal; its error $\epsilon_{\mathrm{rnd}}$ is added as in Theorem~\ref{thm:success_protocol}.
\end{itemize}
Write $N_\text{acq}$ for the number of data-acquisition positions, $N_\text{c}=N_\text{acq}-n_\text{t}$ for the number of untested ones, $s_i$ for the state Alice prepared at position $i$, $S_{G_i}$ with $G_i\in\{0,1\}$ for the set Bob announced, $D_i$ for his declaration if $i\in\mathcal{T}$, and $\rho_s$, $\mathbf{n}_s$ for the four states and their Bloch vectors.  For a set $S_G$, write $u_0,u_1$ for its bit-$0$ and bit-$1$ states: $u_0=\ket0$, $u_1=\ket+$ for $S_0$, and $u_0=\ket-$, $u_1=\ket1$ for $S_1$.  The entries of Table~\ref{tab:usd_test} follow from $\Pr[\text{claim }u_b\mid\psi]=\tfrac12(2-\sqrt2)|\braket{u_{1-b}^{\perp}}{\psi}|^{2}$, where $u^{\perp}$ is the state orthogonal to $u$, together with normalization; the overlaps $|\braket{u^{\perp}}{\psi}|^{2}$ take only the values $0,\tfrac12,1$, and the frequencies are not conditioned on the set guess being correct.  Over a depolarizing channel of strength $p$ the same expressions apply with $|\braket{u^{\perp}}{\psi}|^{2}\mapsto(1-p)|\braket{u^{\perp}}{\psi}|^{2}+p/2$, which lifts the zero entries to $(2-\sqrt2)p/4$.

\begin{lemma}[Independence given the record]
\label{lem:cert_independence}
Under (A2) and (A3), $\Pr[s_1,s_2,\ldots\mid y]=\prod_iP_i(s_i)$ with $P_i(s)\propto\operatorname{Tr}[\Pi^{(i)}_{y_i|y_{<i}}\rho_s]$.  Moreover, $P_i(s)=\tfrac14(1+\mathbf{v}_i\cdot\mathbf{n}_s)$ for a vector $\mathbf{v}_i$ in the $XZ$ plane with $|\mathbf{v}_i|\le1$.
\end{lemma}
\begin{proof}
$\Pr[s_1,s_2,\ldots,y]=\prod_i\tfrac14\operatorname{Tr}[\Pi^{(i)}_{y_i|y_{<i}}\rho_{s_i}]$, for $y_i$ depends on the past only through the choice of POVM and on the prepared states only through $\rho_{s_i}$.  For fixed $y$, each factor involves a single $s_i$, so the conditional law factorizes, and its $i$th factor, normalized, is $P_i$.  The second claim is the Bloch form of a qubit effect $\Pi$, normalized by $\sum_s\operatorname{Tr}[\Pi\rho_s]=2\operatorname{Tr}\Pi$, since the two states of each basis add up to $I$.
\end{proof}

In words, $P_i(s)$ is Bob's posterior probability, given his whole record $y$, that Alice prepared $s$ at position $i$; it is all he knows about that position.  It is fixed by one point $\mathbf{v}_i$ of the unit disc: $P_i(\ket0)=\tfrac14(1+v_{i,Z})$, $P_i(\ket1)=\tfrac14(1-v_{i,Z})$, $P_i(\ket+)=\tfrac14(1+v_{i,X})$ and $P_i(\ket-)=\tfrac14(1-v_{i,X})$.  Note that adaptivity does not spoil the product form: it couples the outcomes $y_i$, not the states $s_i$ once $y$ is fixed.  A joint measurement of several qubits would put several $s_i$ into one factor, which is why we need (A2).

Note that conditioning on the announced sets, on the outcome of the test (an event in the states $s_i$ with $i\in\mathcal{T}$, in $y$ and in $\mathcal{T}$) or on Alice's retained list keeps the product form on the retained positions.  Indeed, the announced sets are functions of $y$; the test event involves only the states at the positions of $\mathcal{T}$, which are discarded; and retention is the event $s_i\in S_{G_i}$, position by position, which restricts $P_i$ to the two states of $S_{G_i}$ and renormalizes it.

\begin{lemma}[Erasure domination]
\label{lem:erasure_domination}
Let a bit $S$ have posterior $(g,1-g)$ with $g\ge\tfrac12$ given Bob's record.  Hand Bob a flag $K$: if $S$ is the likelier value, $K=1$ with probability $(2g-1)/g$; otherwise $K=0$.  Then $\Pr[K=1]=2g-1$, the event $K=1$ reveals $S$, and given $K=0$ the bit $S$ is exactly uniform.
\end{lemma}
\begin{proof}
$\Pr[K=1]=g\cdot\tfrac{2g-1}{g}=2g-1$, and $K=1$ occurs only when $S$ takes the likelier value.  Moreover, $\Pr[S\text{ likelier},K=0]=g\cdot\tfrac{1-g}{g}=1-g=\Pr[S\text{ other},K=0]$.
\end{proof}

So a bit that Bob guesses with bias $2g-1$ is worth no more to him than a bit that he knows with probability $2g-1$ and otherwise not at all.  The flags are a device of the proof, not a step of the protocol: think of a referee who knows $S$ and hands them to Bob.  The flags are drawn independently per position.  Extra information only helps Bob, so any bound against the view $(y,K)$ holds against $y$.  More importantly, with the flagged positions counted as conclusive and all others as ambiguous, the erasure hypothesis of Lemma~\ref{lem:seeded_parity} holds \emph{exactly}.  An honest noiseless Bob has $g\in\{\tfrac12,1\}$, so nothing is lost.

How large can the flag rate be?  For position $i$ put $r_i=\sum_{s\in S_{G_i}}P_i(s)$, the posterior probability that the set guess is correct, and $\alpha_i=|P_i(u_0)-P_i(u_1)|$ for the two states $u_0,u_1$ of $S_{G_i}$.  Once Alice confirms the guess, Bob's posterior on the bit is $(P_i(u_0),P_i(u_1))/r_i$; let $g_i$ be its larger entry.  The bit of a correctly guessed position then has $q_i:=2g_i-1=\alpha_i/r_i$, which by Lemma~\ref{lem:erasure_domination} is the probability that the position is flagged.  For example (Table~\ref{tab:posterior_examples}), an honest Bob who guessed $S_0$ has $r_i=\alpha_i=\tfrac14$ and $q_i=1$ after a conclusive outcome, and $r_i=\tfrac12(1+2^{-1/2})$ and $\alpha_i=q_i=0$ after the inconclusive one; averaged over his outcomes, $\EX\alpha_i/\EX r_i=\USDprob$.  The guessing Bob of Section~\ref{subsec:usd_test} has $r_i=\tfrac12$ and $q_i=2^{-1/2}$ at every position.

\begin{table*}[t]
\caption{\label{tab:posterior_examples}Bob's posterior $P_i(s)$ and the numbers $r_i$, $\alpha_i$, $q_i$ at a position where he announces $S_0$, so that $u_0=\ket0$ and $u_1=\ket+$: the three outcomes of the measurement of Step~\ref{QSPIR:initial}, and one outcome of the guessing Bob of Section~\ref{subsec:usd_test}, whose other outcome gives the same $r_i$, $\alpha_i$, $q_i$.  The posteriors follow from Bayes' rule, $P_i(s)\propto\operatorname{Tr}[\Pi\rho_s]$, and ``probability'' is $\tfrac12\operatorname{Tr}\Pi$, the probability of the outcome.}
\begin{ruledtabular}
\begin{tabular}{llcccccc|ccc}
Bob's outcome & element $\Pi$ & $(v_{i,X},v_{i,Z})$ & $P_i(\ket0)$ & $P_i(\ket+)$ & $P_i(\ket1)$ & $P_i(\ket-)$ & probability & $r_i$ & $\alpha_i$ & $q_i$\\
\colrule
claims $\ket+$ & $(2-\sqrt2)\ket1\bra1$ & $(0,-1)$ & $0$ & $\tfrac14$ & $\tfrac12$ & $\tfrac14$ & $\USDprob$ & $\tfrac14$ & $\tfrac14$ & $1$\\
claims $\ket0$ & $(2-\sqrt2)\ket-\bra-$ & $(-1,0)$ & $\tfrac14$ & $0$ & $\tfrac14$ & $\tfrac12$ & $\USDprob$ & $\tfrac14$ & $\tfrac14$ & $1$\\
inconclusive & $2(\sqrt2-1)\ket{\phi_0}\bra{\phi_0}$ & $(2^{-1/2},2^{-1/2})$ & $\tfrac{2+\sqrt2}{8}$ & $\tfrac{2+\sqrt2}{8}$ & $\tfrac{2-\sqrt2}{8}$ & $\tfrac{2-\sqrt2}{8}$ & $\sqrt2-1$ & $\tfrac{2+\sqrt2}{4}$ & $0$ & $0$\\
guessing, ``$+$'' & $\tfrac12\big[I+(\sigma_Z-\sigma_X)/\sqrt2\big]$ & $(-2^{-1/2},2^{-1/2})$ & $\tfrac{2+\sqrt2}{8}$ & $\tfrac{2-\sqrt2}{8}$ & $\tfrac{2-\sqrt2}{8}$ & $\tfrac{2+\sqrt2}{8}$ & $\tfrac12$ & $\tfrac12$ & $\tfrac{1}{2\sqrt2}$ & $\tfrac{1}{\sqrt2}$\\
\end{tabular}
\end{ruledtabular}
\end{table*}

Why do we work with $\alpha_i$ and $r_i$ rather than with $q_i$?  Whether a set guess is correct is random: given $y$, the indicator $I_i$ of the event $s_i\in S_{G_i}$ is a $\mathrm{Bern}(r_i)$ variable (Lemma~\ref{lem:cert_independence}).  Hence
\begin{equation}
    \EX\sum_iI_iq_i=\sum_ir_iq_i=\sum_i\alpha_i,\qquad\EX\sum_iI_i=\sum_ir_i ,
    \label{eq:alpha_r_identity}
\end{equation}
so $\sum_i\alpha_i/\sum_ir_i$ is, up to fluctuations, the mean of $q_i$ over the positions that will be retained.  Both sums are linear in the posteriors, and hence in the frequencies that the test measures on correct and incorrect set guesses alike; the ratio $q_i$ is not.

The certificate rests on a bound for a single qubit.  Let $M_d\ge0$ be positive operators, one for each of the six declarations $d$ of Table~\ref{tab:usd_test}, whose table is within $\delta$ of the honest one, $|\operatorname{Tr}(\rho_sM_d)-T_{s,d}|\le\delta$ in all $24$ cells.  Completeness, $\sum_dM_d=I$, is not required, so a Bob who reports some qubits as lost is covered.  A \emph{private refinement} splits each $M_d$ into parts that only Bob can tell apart, and so may carry more information than $d$ itself.  Write $p_{\mathrm{guess}}$ for the probability that Bob guesses the bit of a position whose set guess is correct, for a uniformly prepared state.

\begin{lemma}[Single-qubit guessing bound]
\label{lem:guessing_bound}
For $0\le\delta<\tfrac16$ and every private refinement, $2p_{\mathrm{guess}}-1\le\min\{1,B(\delta)\}$, that is, $p_{\mathrm{guess}}\le G(\delta)$, where
\begin{equation}
\begin{split}
 A(\delta)&=\sqrt{2\delta+(8+4\sqrt2)\delta^2},\\
 B(\delta)&=\frac{\USDprob+4\delta+A(\delta)}{1-6\delta},\\
 G(\delta)&=\min\{1,[1+B(\delta)]/2\}.
\end{split}
\label{eq:decoy_G_bound}
\end{equation}
For $\delta\ge\tfrac16$ we use $G=1$.  Thus $B(0)=\USDprob$, $G(0)=(1+\USDprob)/2$ and $G(10^{-4})\le0.653813$.
\end{lemma}
\begin{proof}
Follow one position.  Alice prepares each state with probability $\tfrac14$, Bob declares $d$, and a refinement splits $M_d=\sum_kM_{d,k}$ with $M_{d,k}\ge0$.  Write $S_d$ for the set named by $d$ and $u_0,u_1$ for its two states.  The position is retained if the state lies in $S_d$, which has weight
\begin{equation*}
 R_1=\tfrac14\sum_d\sum_{s\in S_d}\operatorname{Tr}[M_d\rho_s].
\end{equation*}
Given $(d,k)$, Bob guesses the likelier of $u_0$ and $u_1$, so the weights of a right and of a wrong guess add up to $R_1$ and differ by
\begin{equation}
 W_\text{right}-W_\text{wrong}=\tfrac14\sum_{d,k}\big|\operatorname{Tr}[M_{d,k}(\rho_{u_0}-\rho_{u_1})]\big| .
 \label{eq:guess_weight}
\end{equation}
Hence $2p_{\mathrm{guess}}-1=(W_\text{right}-W_\text{wrong})/R_1$.  We bound the denominator in Step~1 and the numerator in Steps~1--4.

\emph{Step 1: retention and conclusive weight.}  For each state $s$, the three cells of Table~\ref{tab:usd_test} whose declaration names a set containing $s$ add up to $\tfrac12$; for $s=\ket0$, for example, $\tfrac{2-\sqrt2}{4}+0+\tfrac{\sqrt2}{4}=\tfrac12$.  These are twelve cells in all, each off by at most $\delta$, and $R_1$ weights them by $\tfrac14$.  Likewise, the two conclusive cells among them add up to $\tfrac{2-\sqrt2}{4}=\tfrac12\USDprob$ for each $s$, eight cells in all.  Writing $C_1$ for the part of $R_1$ with a conclusive declaration, we get
\begin{equation}
 R_1\ge\tfrac12-3\delta,\qquad C_1\le\tfrac12\USDprob+2\delta.
\end{equation}
We credit Bob with knowing the bit after every conclusive declaration.  Each term of Eq.~\eqref{eq:guess_weight} is at most $\tfrac14\operatorname{Tr}[M_{d,k}(\rho_{u_0}+\rho_{u_1})]$, so the conclusive declarations contribute at most $C_1$.

\emph{Step 2: an inconclusive effect is nearly rank one.}  Take the inconclusive declaration of $S_0$; that of $S_1$ is the same with $\phi_1$ in place of $\phi_0$.  Its honest effect, including the factor $\tfrac12$ of the set guess, is $\lambda_\text{inc}\ket{\phi_0}\bra{\phi_0}$ with $\lambda_\text{inc}=\sqrt2-1$.  It vanishes on $\ket{\phi_0^\perp}$, whose Bloch vector is $-(\hat x+\hat z)/\sqrt2$.  Write Bob's actual effect as $M=(m_0I+m_X\sigma_X+m_Y\sigma_Y+m_Z\sigma_Z)/2$, so that $\operatorname{Tr}[M\rho_s]=\tfrac12(m_0+\mathbf{m}\cdot\mathbf{n}_s)$.  The four cells of this declaration then give
\begin{equation*}
\begin{split}
 m_0&=\operatorname{Tr}[M\rho_{\ket0}]+\operatorname{Tr}[M\rho_{\ket1}],\\
 m_Z&=\operatorname{Tr}[M\rho_{\ket0}]-\operatorname{Tr}[M\rho_{\ket1}],\\
 m_X&=\operatorname{Tr}[M\rho_{\ket+}]-\operatorname{Tr}[M\rho_{\ket-}],
\end{split}
\end{equation*}
so each of them lies within $2\delta$ of its honest value, $\lambda_\text{inc}$, $\lambda_\text{inc}/\sqrt2$ and $\lambda_\text{inc}/\sqrt2$ respectively.  The component $m_Y$ is not tested, and it does not enter below.  The weight of $M$ in the null direction,
\begin{equation*}
 \bra{\phi_0^\perp}M\ket{\phi_0^\perp}=\tfrac12\big[m_0-(m_X+m_Z)/\sqrt2\big],
\end{equation*}
vanishes for the honest effect, so it is at most $\tfrac12\big[2\delta+2\cdot2\delta/\sqrt2\big]=(1+\sqrt2)\delta$.  The weight in the direction $\ket{\phi_0}$ is at most $\operatorname{Tr}M=m_0\le\lambda_\text{inc}+2\delta$.  In the basis $\{\ket{\phi_0^\perp},\ket{\phi_0}\}$, $M$ is a positive $2\times2$ matrix, whose determinant is at most the product of its diagonal entries.  Hence
\begin{equation}
 \det M\le(1+\sqrt2)\,\delta\,(\lambda_\text{inc}+2\delta).
\end{equation}

\emph{Step 3: what a refinement of $M$ can reveal.}  Every refinement of $M$ has the form $M_k=\sqrt ME_k\sqrt M$ with $E_k\ge0$ and $\sum_kE_k\le I$: as $0\le M_k\le M$, one may take $E_k=M^{-1/2}M_kM^{-1/2}$ on the support of $M$, a finite-dimensional case of Douglas' factorization lemma~\cite{Douglas_factorization}.  In other words, Bob first applies the L\"{u}ders operation $\sqrt M(\cdot)\sqrt M$~\cite{Lueders_1950,Lueders_translation} and then measures $\{E_k\}$ on the filtered state.  Put $X=\sqrt M(\rho_{u_0}-\rho_{u_1})\sqrt M$ and split it into its positive and negative parts, $X=X_+-X_-$.  Then $|\operatorname{Tr}[M_k(\rho_{u_0}-\rho_{u_1})]|=|\operatorname{Tr}[E_kX]|\le\operatorname{Tr}[E_k(X_++X_-)]$, and the sum over $k$ bounds the contribution of $M$ to Eq.~\eqref{eq:guess_weight} by $\tfrac14\|X\|_1$; this is the Helstrom bound for the two states filtered by $M$.  Now $X$ is a Hermitian $2\times2$ matrix with
\begin{equation*}
\begin{split}
 \operatorname{Tr}X&=\pi_0-\pi_1,\\
 \det X&=\det M\,\det(\rho_{u_0}-\rho_{u_1})=-\tfrac12\det M,
\end{split}
\end{equation*}
where $\pi_b=\operatorname{Tr}[M\rho_{u_b}]$.  Here $\det(\rho_{u_0}-\rho_{u_1})=-\tfrac12$ because the Bloch vectors of $u_0$ and $u_1$ are perpendicular, so that $\rho_{u_0}-\rho_{u_1}$ has the eigenvalues $\pm1/\sqrt2$.  As $\det X\le0$, the eigenvalues $x_\pm$ of $X$ have opposite signs, and
\begin{equation}
\begin{split}
 \|X\|_1=x_+-x_-&=\sqrt{(\operatorname{Tr}X)^2-4\det X}\\
   &=\sqrt{(\pi_0-\pi_1)^2+2\det M}.
\end{split}
\end{equation}

\emph{Step 4: the function $A$.}  The honest inconclusive effect is equally likely for $u_0$ and $u_1$, so $|\pi_0-\pi_1|\le2\delta$.  With Step~2 and $(1+\sqrt2)(\sqrt2-1)=1$,
\begin{equation*}
\begin{split}
 \|X\|_1^2&\le4\delta^2+2(1+\sqrt2)\,\delta\,(\sqrt2-1+2\delta)\\
 &=2\delta+(8+4\sqrt2)\delta^2=A(\delta)^2 .
\end{split}
\end{equation*}
So each inconclusive declaration contributes at most $\tfrac14A(\delta)$.  This is where the square root comes from.

\emph{Step 5: the function $B$.}  Adding the two inconclusive declarations to Step~1,
\begin{equation}
\begin{split}
 2p_{\mathrm{guess}}-1&\le\frac{C_1+\tfrac12A(\delta)}{R_1}\\
 &\le\frac{\tfrac12\USDprob+2\delta+\tfrac12A(\delta)}{\tfrac12-3\delta}=B(\delta),
\end{split}
\end{equation}
where $\delta<\tfrac16$ keeps the denominator positive.  As $p_{\mathrm{guess}}\le1$, this gives $2p_{\mathrm{guess}}-1\le\min\{1,B(\delta)\}$ and $p_{\mathrm{guess}}\le G(\delta)$.  Note that $B(0)=\USDprob$: an exact table leaves Bob only his conclusive outcomes.  The terms $4\delta$ and $1-6\delta$ are the linear slack of $C_1$ and $R_1$, while $A(\delta)\approx\sqrt{2\delta}$ dominates for small $\delta$.
\end{proof}

We now apply the lemma to the test sample.  The hypothesis of the next lemma is the test of Step~\ref{QSPIR:usd_test} with Bob's posteriors in place of the observed frequencies; part~(a) of the proof of Proposition~\ref{prop:cert_rate} connects the two.

\begin{lemma}[Certificate on the test sample]
\label{lem:cert_sample}
Suppose that $\big|\tfrac{4}{n_\text{t}}\sum_{i\in\mathcal{T},\,D_i=d}P_i(s)-T_{s,d}\big|\le\delta_1<\tfrac16$ for all $24$ cells of Table~\ref{tab:usd_test}.  Then
\begin{equation}
    \sum_{i\in\mathcal{T}}\alpha_i\le\min\{1,B(\delta_1)\}\sum_{i\in\mathcal{T}}r_i,\qquad \frac{1}{n_\text{t}}\sum_{i\in\mathcal{T}}r_i\ge\tfrac12-3\delta_1,
\end{equation}
with $B$ given by Eq.~\eqref{eq:decoy_G_bound}.
\end{lemma}
\begin{proof}
Let $\tilde\Pi_i=\tfrac12(I+\mathbf{v}_i\cdot\boldsymbol{\sigma})\ge0$, which is Bob's element at position $i$ rescaled to unit trace and without its $\sigma_Y$ part, so that $\operatorname{Tr}[\tilde\Pi_i\rho_s]=2P_i(s)$, and let $M_d=\tfrac{2}{n_\text{t}}\sum_{i\in\mathcal{T},\,D_i=d}\tilde\Pi_i$ be the average element behind the declaration $d$.  Then $\operatorname{Tr}[M_d\rho_s]=\tfrac{4}{n_\text{t}}\sum_{i\in\mathcal{T},\,D_i=d}P_i(s)$.  So the $M_d$ are positive operators whose table is within $\delta_1$ of the honest one, and the positions are a private refinement of them: Bob knows not only $d$ but also which position contributed.  This is the hypothesis of Lemma~\ref{lem:guessing_bound}, whose proof uses the positivity of each $M_d$ and its four table entries only, never completeness.  With the prior $\tfrac14$ on each state, the part of $M_d$ contributed by position $i$ carries weight $P_i(s)/n_\text{t}$ on the state $s$.  Its weight on a correct set guess is therefore $r_i/n_\text{t}$, and on a correct set guess followed by a correct bit guess it is $\max\{P_i(u_0),P_i(u_1)\}/n_\text{t}=(r_i+\alpha_i)/(2n_\text{t})$.  Hence $2p_{\mathrm{guess}}-1=\sum_{i\in\mathcal{T}}\alpha_i/\sum_{i\in\mathcal{T}}r_i$, and the lemma gives the first claim; the minimum with $1$ is just $\alpha_i\le r_i$.  For the second claim, note that for every state $s$ the three cells whose declaration names a set containing $s$ add up to $\tfrac12$ in Table~\ref{tab:usd_test}.  Averaged over $s$, the same twelve cells of the $M_d$ give $\tfrac{1}{n_\text{t}}\sum_{i\in\mathcal{T}}r_i$, and their deviations, weighted by $\tfrac14$, add up to at most $3\delta_1$.
\end{proof}

As a check, the linear program of Section~\ref{subsec:usd_test} gives $+0.011$, $+0.037$ and $+0.133$, the values quoted below Eq.~\eqref{eq:usd_test_slack}, which lie below $B(\delta)-\USDprob=+0.015$, $+0.051$ and $+0.217$ at $\delta=10^{-4},10^{-3},10^{-2}$.  Without its completeness constraint, which Lemma~\ref{lem:cert_sample} does not use, they rise by at most $0.004$.

We now combine the lemmas.  The test sees only $\mathcal{T}$, while the blocks are drawn from the untested positions with a correct set guess; the next proposition bridges the two.  Set
\begin{equation}
    \begin{split}
    h_\text{t}&=\sqrt{\frac{\ln(104/\epsilon_\text{cert})}{2n_\text{t}}},\qquad
    h_\text{c}=\sqrt{\frac{\ln(4/\epsilon_\text{cert})}{2N_\text{c}}},\\
    h_\text{t}'&=\frac{h_\text{t}N_\text{acq}}{N_\text{c}},\qquad
    \delta_1=\delta_\text{test}+4h_\text{t},\\
    \USDprob^\text{cert}&=\frac{\min\{1,B(\delta_1)\}\big(\tfrac12-3\delta_1\big)+h_\text{t}'+h_\text{c}}{\big(\tfrac12-3\delta_1\big)-h_\text{t}'-h_\text{c}} .
    \end{split}
    \label{eq:p_cert}
\end{equation}
Here $h_\text{t}$ and $h_\text{c}$ are Hoeffding margins: the largest deviation, from its expectation, that we tolerate for an average over the $n_\text{t}$ test positions and over the $N_\text{c}$ untested ones, respectively.  The test-sample margin enters $\delta_1$ as $4h_\text{t}$ because $\hat f(s,d)$ is a count divided by $n_\text{t}/4$, and $h_\text{t}'$ carries it from the test sample to the untested positions.  The constants $104$ and $4$ split $\epsilon_\text{cert}$ into two halves.  The proof below uses $26$ two-sided bounds of Hoeffding type on the test sample, $24$ for the cells and $2$ for the means of $\alpha_i$ and $r_i$.  Each side fails with probability at most $e^{-2n_\text{t}h_\text{t}^{2}}$, so these $52$ one-sided events fail with probability at most $52e^{-2n_\text{t}h_\text{t}^{2}}$, which equals $\epsilon_\text{cert}/2$ when $e^{-2n_\text{t}h_\text{t}^{2}}=\epsilon_\text{cert}/104$.  On the untested positions it uses $2$ one-sided bounds, and $2e^{-2N_\text{c}h_\text{c}^{2}}=\epsilon_\text{cert}/2$ gives the $4$.

\begin{proposition}[Certified known-rate]
\label{prop:cert_rate}
Under (A1)--(A3), with $\delta_1<\tfrac16$ and a positive denominator in Eq.~\eqref{eq:p_cert}, let $\bar q$ be the mean of $q_i$ over the non-test positions with a correct set guess, the population from which the blocks are drawn.  Then $\Pr[\text{test passes and }\bar q>\USDprob^\text{cert}]\le\epsilon_\text{cert}$.
\end{proposition}
\begin{proof}
(a) Given $y$, $\mathcal{T}$ and the declarations, the $s_i$ with $i\in\mathcal{T}$ are independent with laws $P_i$ (Lemma~\ref{lem:cert_independence}).  Conditioned on Bob's record, Alice's preparations are thus fresh randomness that he does not control, and $\hat f(s,d)$ is an average of independent indicators with mean $\tfrac{4}{n_\text{t}}\sum_{i\in\mathcal{T},\,D_i=d}P_i(s)$.  By Hoeffding's inequality and a union bound over the $24$ two-sided cells, every $\hat f(s,d)$ is within $4h_\text{t}$ of $\tfrac{4}{n_\text{t}}$$\sum_{i\in\mathcal{T},\,D_i=d}P_i(s)$ except with probability $48e^{-2n_\text{t}h_\text{t}^2}$.  On passing, the hypothesis of Lemma~\ref{lem:cert_sample} then holds with the stated $\delta_1$.  (b) The numbers $\alpha_i,r_i\in[0,1]$ are fixed by $y$ and the announced sets before $\mathcal{T}$ is drawn; the declarations may depend on $\mathcal{T}$, but $\alpha_i$ and $r_i$ do not depend on them.  Hoeffding's inequality for sampling without replacement~\cite{Hoeffding_1963} puts the sample means of $\alpha_i$ and $r_i$ within $h_\text{t}$ of the population means except with probability $4e^{-2n_\text{t}h_\text{t}^2}$; hence the means over the complement of $\mathcal{T}$ differ from the sample means by at most $h_\text{t}'$, since $N_\text{acq}$ times the population mean is $n_\text{t}$ times the sample mean plus $N_\text{c}$ times the mean over the complement.  (c) On the complement the correct-guess indicators $I_i$ are independent $\mathrm{Bern}(r_i)$, so that by Eq.~\eqref{eq:alpha_r_identity} the means of $\sum I_iq_i$ and $\sum I_i$ are $\sum\alpha_i$ and $\sum r_i$.  One-sided Hoeffding bounds give $\sum I_iq_i\le\sum \alpha_i+h_\text{c}N_\text{c}$ and $\sum I_i\ge\sum r_i-h_\text{c}N_\text{c}$ except with probability $2e^{-2N_\text{c}h_\text{c}^2}$.  (d) Combining, with $\bar\alpha_{\mathcal{T}}$ and $\bar r_{\mathcal{T}}$ the means of $\alpha_i$ and $r_i$ over $\mathcal{T}$, $\bar q\le(\bar\alpha_{\mathcal{T}}+h_\text{t}'+h_\text{c})/(\bar r_{\mathcal{T}}-h_\text{t}'-h_\text{c})$, by inserting (b) into the bounds of (c) on $\bar q=\sum I_iq_i/\sum I_i$.  The right-hand side decreases in $\bar r_{\mathcal{T}}$ once $\bar\alpha_{\mathcal{T}}\le\min\{1,B\}\bar r_{\mathcal{T}}$ is inserted, its derivative being $-h(1+\min\{1,B\})/(\bar r_{\mathcal{T}}-h)^2$ with $h=h_\text{t}'+h_\text{c}$, so the worst case is $\bar r_{\mathcal{T}}=\tfrac12-3\delta_1$, the lower bound of Lemma~\ref{lem:cert_sample}.  The three failure probabilities, $48e^{-2n_\text{t}h_\text{t}^2}$, $4e^{-2n_\text{t}h_\text{t}^2}$ and $2e^{-2N_\text{c}h_\text{c}^2}$, add up to $\epsilon_\text{cert}$ by the choice of $h_\text{t}$ and $h_\text{c}$ explained below Eq.~\eqref{eq:p_cert}.
\end{proof}

\begin{theorem}[Database privacy against a cheating Bob]
\label{thm:cert_privacy}
Assume (A1)--(A3).  Then Theorem~\ref{thm:success_protocol}(ii) holds without the erasure hypothesis, with the known-rate $\pknown=\USDprob^\text{cert}$:
\begin{equation}
    \SecretB\le\epsilon_\text{cert}+\frac{\num^{2}}{4}\,2^{-(\cthr-2\hmarg)/2}\Big[1+(\sqrt2-1)\big(\USDprob^\text{cert}\big)^{2}\Big]^{\numd}+\epsilon_{\mathrm{rnd}}.
    \label{eq:epsB_cert}
\end{equation}
\end{theorem}
\begin{proof}
On the event of Proposition~\ref{prop:cert_rate}, hand Bob the flags of Lemma~\ref{lem:erasure_domination}, one per correctly guessed position, set with probability $q_i$, and let $K_j$ be the flagged locations of block $j$.  The erasure hypothesis of Lemma~\ref{lem:seeded_parity} then holds exactly, the flags are independent given his record (Lemma~\ref{lem:cert_independence}), and he chooses $\aset$ from that record and the flags.  The retained list is a uniform subsample of the correctly guessed non-test positions, and the blocks are a uniform partition of it; so the positions of any two blocks are drawn uniformly without replacement from the population on which $\bar q$ is certified, and Lemma~\ref{lem:overlap_mgf} applies with $\bar q\le\USDprob^\text{cert}$, and its bound increases in $\bar q$.  The rest is the proof of Eq.~\eqref{eq:EW_common} in Appendix~\ref{app:simulator_proof}, which holds for every set $\aset$.  Note that only the mean $\bar q$ enters: by Step~3 of the proof of Lemma~\ref{lem:overlap_mgf}, an uneven spread of the $q_i$ can only lower the overlap of two blocks.
\end{proof}

What does the certificate cost?  As $\delta_\text{test}\to0$ and $n_\text{t}\to\infty$, $\USDprob^\text{cert}\to B(0)=\USDprob$, so the cost is the slack at finite $n_\text{t}$.  With $\num=10^{4}$ and the block length $\numd$ found below, $N_\text{acq}=2(\numd\num+n_\text{t})$, $\epsilon_\text{cert}=10^{-6}$ and $\delta_\text{test}=4h_\text{t}$, the slack $\USDprob^\text{cert}-\USDprob$ is $0.088$, $0.024$, $0.0071$ and $0.0007$ at $n_\text{t}=10^{8},10^{10},10^{12}$ and $10^{16}$.  These values rely on the acquisition of Step~\ref{QSPIR:initial}, which makes the untested complement $N_\text{c}$ larger than $n_\text{t}$ and so keeps $h_\text{t}'\le2h_\text{t}$; the test thus costs about $2n_\text{t}$ transmitted qubits.  It falls only as $n_\text{t}^{-1/4}$, the square root of Eq.~\eqref{eq:usd_test_slack} acting on a sampling error of order $n_\text{t}^{-1/2}$.  Since the certified rate enters Eq.~\eqref{eq:epsB_cert} only through its square, this slack is affordable: at $\num=10^{4}$ and a budget $\EX[W]\le10^{-6}$ the closed-form choice of Appendix~\ref{app:simulator_proof} gives $\numd=1640$, $1090$ and $1010$ at the first three test sizes, against $980$ without slack.  The certificate is, however, stated for the noiseless table: channel noise must be absorbed in $\delta_\text{test}$, and a zero entry lifted to $(2-\sqrt2)p/4$ costs a slack of order $\sqrt p$.

\subsection{Bob's Cheating by Delayed Measurement}
\label{sec:Bob_delayed_measurement}

What if Bob has long-term quantum memory?  Recall that we assumed that Bob does not, which is a fair assumption in the current experimental frontier and in the near future.  However, if Bob does have long-term quantum memory, he can store all the qubits from Alice, pass the measurement test of Step~\ref{QSPIR:usd_test} by measuring the tested positions as prescribed once the test sample is drawn, and keep the other qubits until Alice discloses the correct sets in Step~\ref{QSPIR:set_disclosure}.

Then, once Alice has disclosed the sets, Bob could replace USD by the best measurement for the two states of each set, the Helstrom measurement~\cite{Barnett_Croke_review}.  For two pure states with overlap $\nu$ it guesses correctly with probability $\tfrac12(1+\sqrt{1-\nu^2})$, here $\tfrac12+\tfrac{1}{2\sqrt2}\approx0.854$ for $\nu=|\braket{0}{+}|=1/\sqrt2$.  This is much larger than the USD success probability $\USDprob\approx0.293$, although Bob does not know which guesses are right.  Knowing the set does not help him, however.  For both sets the states carrying the bits $0$ and $1$ differ by $\ket{0}\bra{0}-\ket{+}\bra{+}=\ket{-}\bra{-}-\ket{1}\bra{1}=(\sigma_Z-\sigma_X)/2=:{}\Gamma_-$, so the Helstrom measurement is always the one along $(\sigma_Z-\sigma_X)/\sqrt2$, the fixed measurement of Section~\ref{subsec:usd_test}, which needs no memory.  Measured one by one, the stored qubits therefore give him nothing beyond a memoryless Bob who guesses every bit with probability $0.854$; his guesses do not decode a block, since correcting their $14.6\%$ errors takes $H_2(0.146)\approx0.60$ check symbols per location, whereas the protocol publishes $\hmarg/\cthr<0.31$.  What a memory adds is a joint measurement.  If Bob measures, jointly or not, before the announcements of Step~\ref{QSPIR:final_step}, we can still bound what each block leaks to him.  Let $f(\dstr{w})$ be his bias on the parity $\dstr{w}\cdot\Psubstr[j]{A}|_{\aset}$ given his view without the check symbols $\sigma_j$.  Expanding the indicator of the coset fixed by $\sigma_j$ in characters, $\Pr[\sigma_j=\sigma]=2^{-\hmarg}D_\sigma$ and the bias of the seeded parity given $\sigma_j=\sigma$ is $N_\sigma/D_\sigma$, where $N_\sigma=\sum_{\dstr{u}\in\{0,1\}^{\hmarg}}(-1)^{\dstr{u}\cdot\sigma}f(\dstr{u}\dstr{H}\oplus\dstr{z})$ and $D_\sigma$ is the same sum with $\dstr{z}=0$.  The leak is therefore $\tfrac12\,2^{-\hmarg}\sum_\sigma|N_\sigma|$, and the Cauchy--Schwarz inequality, Parseval's identity over $\sigma$ and Jensen's inequality bound its mean by
\begin{equation}
    \tfrac12\Big(\sum_{\dstr{u}\in\{0,1\}^{\hmarg}}\EX f(\dstr{u}\dstr{H}\oplus\dstr{z})^{2}\Big)^{1/2}.
    \label{eq:parity_leak_CS}
\end{equation}
For a uniform seed, $\dstr{u}\dstr{H}\oplus\dstr{z}$ is uniform.  If Bob applies the Helstrom measurement to each stored qubit, his bias on a parity of $|\dstr{w}|$ locations is $\pm2^{-|\dstr{w}|/2}$, so each of the $2^{\hmarg}$ terms equals $(3/4)^{\cthr}$ and the leak is at most $\tfrac12\,2^{(\hmarg-\beta\cthr)/2}$ with $\beta=\log_2(4/3)\approx0.415$.  At the noiseless point of Table~\ref{tab:common_set}, $\cthr=218$ and $\hmarg=20$, this is $1.2\times10^{-7}$ summed over all $\num=10^{4}$ blocks, below $\SecretB=10^{-6}$.  For a joint measurement, note that given the sets, the states of $k$ positions whose bits have even and odd parity differ by $2^{-(k-1)}\bigotimes_{i}(\pm\Gamma_-)$, of trace norm $2\cdot2^{-k/2}$ since $\Vert\Gamma_-\Vert_1=\sqrt2$.  By the Helstrom bound, his bias $f$ on the parity of $k$ positions fixed independently of his record therefore satisfies $\EX|f|\le2^{-k/2}$ for every measurement; tensoring with the other positions leaves the trace norm unchanged, so what he learns later about other blocks is covered as well.  With $f^2\le|f|$, Eq.~\eqref{eq:parity_leak_CS} bounds the leak of the seeded parity by
\begin{equation}
    \tfrac12\,2^{(\hmarg-\beta_\text{coll}\cthr)/2},\qquad \beta_\text{coll}=\log_2\frac{2}{1+2^{-1/2}}\approx0.228,
    \label{eq:collective_leak}
\end{equation}
against $\beta\approx0.415$ for the Helstrom measurement qubit by qubit.  For this leak, summed over $\num=10^{4}$ blocks, to stay below $10^{-6}$ over a noiseless channel, the set needs $\cthr=370$ locations instead of the $218$ of Table~\ref{tab:common_set}, and the leak falls with $\cthr$ only while $H_2(\tilde{\ebit}+\gamma)<\beta_\text{coll}$, that is, up to $\tilde{\ebit}\approx3.7\%$.  Note that the bound is for a set fixed in advance: it does not cover Bob's choice of the retained positions through his set guesses, nor his choice of $\aset$ from his record, which the certificate of Appendix~\ref{app:usd_certification} handles under (A2).  If he measures after $\dstr{H}$, $\dstr{z}$ and the check symbols are announced, we have no bound.

Thus the no-long-term-quantum-memory assumption is needed against a Bob who measures after the announcements of Step~\ref{QSPIR:final_step}, and, at the present set sizes, against one who measures jointly before them.

\section{Index Privacy Against Alice}
\label{app:alice_audit}
This appendix supports Section~\ref{sec:Alice_cheating}.  Bob is honest.  Alice has no long-term quantum memory either, so after the waiting time each position carries a state whose classical description she knows.  For Bob's guess $G\in\{S_0,S_1\}$ write $\kappa_G(\rho)$ for the probability that the measurement of Step~\ref{QSPIR:initial} is conclusive; the honest value on a retained position is $\USDprob$.  The conclusive elements of the two guesses add up to $(2-\sqrt2)(\ket0\bra0+\ket1\bra1+\ket+\bra++\ket-\bra-)=4\USDprob\,I$, so
\begin{equation}
    \begin{split}
    &\kappa_{S_0}(\rho)+\kappa_{S_1}(\rho)=4\USDprob\quad\text{for every }\rho,\\
    &\kappa_G(\rho)\in[3-2\sqrt2,\,1].
    \end{split}
    \label{eq:kappa_invariant}
\end{equation}
How can Alice cheat?
Wrongly prepared single-qubit states by a dishonest Alice can bias Bob's USD record, and the popular no-signaling argument does not exclude this: no local preparation lets Alice sense Bob's measurement choices during the quantum phase~\cite{No_signaling1, No_signaling2}, but Bob's \emph{later classical messages} depend on his USD record.  Indeed, a state $\ket{\psi}$ triggers a conclusive outcome of the guess-$S_0$ POVM with probability $(2-\sqrt{2})\left(|\bra{1}\ket{\psi}|^2+|\bra{-}\ket{\psi}|^2\right)\in[3-2\sqrt{2},\,1]$, compared with the honest $\USDprob\approx0.29$.  Together with her keep-or-discard choice in Step~\ref{QSPIR:set_disclosure}, this lets Alice mark locations: she knows the conclusive probability of every retained position, and the announced set is drawn from the conclusive positions of Bob's block.  A block whose over-conclusive locations are over-represented in $\aset$ is more likely to be his, and the shift then localizes $\pidx$.  The parameter estimation of Step~\ref{QSPIR:rate_estimation} catches it only when it is large, as we now show.

Nor can a finer error-rate test catch her, for the tampering is intrinsically cheap.  For a slightly rotated preparation $\ket{\psi_\theta}=\cos\theta\ket{0}-\sin\theta\ket{1}$, the conclusive probability
\begin{equation}
    \begin{split}
    &(2-\sqrt{2})\left[\sin^2\theta+\tfrac{1}{2}\left(1+\sin 2\theta\right)\right]\\
    &\qquad=\USDprob+(2-\sqrt{2})\,\theta+O(\theta^2)
    \end{split}
    \label{eq:tamper_linear_gain}
\end{equation}
grows \emph{linearly} in $\theta$, because the honest value lies in the interior of the achievable range $[3-2\sqrt{2},\,1]$; the matched-basis error rate $\sin^2\theta$ grows only \emph{quadratically}, because the honest state is an extremum (zero error) of that test. Therefore, if Alice prepares $\ket{\psi_\theta}$ for some small $\theta$, she can bias Bob's conclusive pattern by $O(\theta)$ while inflating the error rate by only $O(\theta^2)$, which is statistically invisible at large $\num$. Hence a dishonest Alice can potentially use it to get partial information on Bob's index.  Over a noisy channel, any test of Alice's states that accepts the honest noise also accepts a small rotation with random signs: averaged over the signs, which only Alice knows, its effect on Bob's record is exactly that of stronger depolarizing noise.  A false claim in Step~\ref{QSPIR:set_disclosure} that a wrong set guess was correct needs no tampering at all: $\ket0$ kept under the guess $S_1$ is conclusive with probability $3\USDprob\approx0.88$.  We therefore do not bound a tampering Alice, and analyze index privacy under assumption (P) of Section~\ref{sec:Alice_cheating}.

Bob's detectors do not mix photon numbers, so we may give Alice the photon number of every position, which can only help her.  After the permutation, write $\kappa_{jl}$ for the conclusive probability of location $l$ of block $j$ and $K_j$ for the set of conclusive locations of block $j$.

\begin{lemma}[Reduction to single blocks]
\label{lem:alice_reduction}
Let $P_j(\aset)$ be the probability that Bob, having drawn block $j$ in Step~\ref{QSPIR:block_draw}, does not abort and announces the set $\aset$, and let $\mathrm{TV}_j$ be the total-variation distance between the law of $K_j$ and the honest law $\mathrm{Bern}(\USDprob)^{\otimes\numd}$.  Then Alice's final states $\rho_A^{\pidx}$ and $\rho_A^{\pidx'}$ for two indices obey
\begin{equation}
    \begin{split}
    \tfrac12\big\Vert\rho_A^{\pidx}-\rho_A^{\pidx'}\big\Vert_1&=\frac{1}{\num}\sum_j\tfrac12\sum_{\aset}\big|P_j(\aset)-P_{j+\pidx'-\pidx}(\aset)\big|\\
    &\le\frac{2}{\num}\sum_j\mathrm{TV}_j .
    \end{split}
    \label{eq:alice_reduction}
\end{equation}
\end{lemma}
\begin{proof}
Condition on Alice's view before Step~\ref{QSPIR:shift_idx}: her preparations, Bob's guesses and test declarations, and the permutation; $P_j$ and $\mathrm{TV}_j$ are conditional on it, and $\SecretA$ is its average.  Everything Alice sees before Bob's last message is independent of $\pidx$ and of his draw $i$.  If he aborts she learns only that, with a probability $\frac1\num\sum_j\Pr[|K_j|<\cthr]$ that does not involve $\pidx$.  Otherwise she holds $\keysh$ and sees $\aset$ and the shift $\Delta_\text{sh}=\pidx-i$, with probability $\frac1\num P_{\pidx-\Delta_\text{sh}}(\aset)$; substituting $j=\pidx-\Delta_\text{sh}$ gives the equality.  Let $P^{0}$ be $P_j$ under the honest law, which is the same for every block.  By the triangle inequality the middle expression is at most $\frac1\num\sum_j\sum_{\aset}|P_j(\aset)-P^{0}(\aset)|$.  The map from $K_j$ to Bob's message is a fixed channel that does not depend on Alice's states, so $\tfrac12\sum_{\aset}|P_j-P^{0}|\le\mathrm{TV}_j$ by data processing.
\end{proof}

Note that the private uniform draw of Step~\ref{QSPIR:block_draw} is what turns a maximum over blocks into an average.  Without it, the bound would be set by the block that happens to carry the most tampered positions.

\begin{proposition}[Honest preparation]
\label{prop:index_privacy_honest}
If every retained position is conclusive with one and the same probability, then $\SecretA=0$.  This holds when Alice sends one photon per position in one of the four legitimate states and reports the correct guesses truthfully, over a noiseless channel and over any channel whose noise does not distinguish the states she sent and acts alike at every position, such as a depolarizing channel of constant strength.
\end{proposition}
\begin{proof}
The law of $K_j$ is then the same for every block, so $P_j=P_{j'}$ for all $j,j'$ and the middle expression of Eq.~\eqref{eq:alice_reduction} vanishes.
\end{proof}
Note that Proposition~\ref{prop:index_privacy_honest} fails for an honest Alice who knows a systematic misalignment by a Bloch-sphere angle $\theta_\text{ch}$ of the channel, since then $\kappa=\USDprob(1\mp\theta_\text{ch})$ with a sign fixed by the state she sent; Pinsker's inequality and the $\chi^2$ bound on the relative entropy give $\mathrm{TV}_j^2\le\numd\USDprob\theta_\text{ch}^2/[2(1-\USDprob)]$, so that Lemma~\ref{lem:alice_reduction} gives $\SecretA\le0.91\sqrt{\numd}\,\theta_\text{ch}$, so the alignment must be calibrated to well below $1/\sqrt{\numd}$.  An Alice with quantum memory is not covered.

\subsection{Several photons per position}
\label{app:alice_multiphoton}
Assumption (P) fixes the state of every photon but not the number of photons per position.  Bob's detectors cannot tell one photon from several, and a pulse of $N$ photons in a legitimate state is conclusive with probability $1-2^{-N/2}$.  It does, however, make two detectors click with probability $1-\USDprob^N-2^{-N/2}$, and the count of Step~\ref{QSPIR:initial} limits how often Alice can afford this.

\begin{lemma}[One position under (P)]
\label{lem:photon_number}
Let the channel be noiseless (loss allowed), let a data-acquisition position carry $N\ge1$ photons in one BB84 state, and let Bob's set guess be uniform.  Let $M$ indicate that two or more detectors click, and on a correctly guessed kept position let $X$ indicate that it is conclusive.  There is a $Y$ with $\Pr[Y=1\mid M,\text{Alice's view}]=\USDprob$ such that
\begin{equation}
  \Pr[\text{guess correct},\,X\ne Y]\le r\,\Pr[M=1],\qquad r=0.390 .
  \label{eq:one_position}
\end{equation}
\end{lemma}
\begin{proof}
With a correct guess every photon reaches the one conclusive detector its state allows with probability $\USDprob$ and the inconclusive one with $\nu=1-\USDprob$.  So the position clicks only conclusively with probability $\USDprob^N$, only inconclusively with $\nu^N$, and on both with $m_c=1-\USDprob^N-\nu^N$, and it is kept in all three cases.  Given $M=1$ it is conclusive; draw $Y$ afresh, which disagrees with probability $\nu$.  Given $M=0$ it is conclusive with probability $q_0=\USDprob^N/(\USDprob^N+\nu^N)$; couple $Y$ to $X$ maximally, which disagrees with probability $|q_0-\USDprob|$, and $(1-m_c)|q_0-\USDprob|=\USDprob\nu(\nu^{N-1}-\USDprob^{N-1})$.  With a wrong guess a photon reaches the two conclusive detectors with probabilities $2\USDprob$ and $\USDprob$ and the inconclusive one with $1-3\USDprob$, so $M=1$ with probability $m_w=1-(2\USDprob)^N-\USDprob^N-(1-3\USDprob)^N$.  The ratio of the two sides of Eq.~\eqref{eq:one_position} is therefore $[\nu m_c+\USDprob\nu(\nu^{N-1}-\USDprob^{N-1})]/(m_c+m_w)$, which is $0.390$ at $N=2$ and decreases towards $\nu/2=0.354$.  For $N=1$, $M=0$ and $X$ itself has the law of $Y$.
\end{proof}

Multi-photon light thus pays for every change of Bob's conclusive record with multi-clicks.  The count bounds how many Alice can afford.

\begin{lemma}[Multi-click count]
\label{lem:multiclick_count}
Let $c_a$ be the probability, given everything before data-acquisition position $a$, that two or more detectors click there, and let $B_D:=\tau_D+\Lambda_D+\sqrt{2\tau_D\Lambda_D+\Lambda_D^{2}}$ with $\Lambda_D=\ln(1/\epsilon_D)$.  Then $\Pr[D\le\tau_D\ \text{and}\ \sum_ac_a\ge B_D]\le\epsilon_D$.
\end{lemma}
\begin{proof}
Write $D_a$ and $C_a$ for the partial sums of the multi-click indicators and of the $c_b$ up to position $a$.  For $\lambda>0$, $\exp[-\lambda D_a+(1-e^{-\lambda})C_a]$ is a supermartingale, since the conditional expectation of $e^{-\lambda}$ to the power of the next indicator is $1-c_a(1-e^{-\lambda})\le\exp[-c_a(1-e^{-\lambda})]$.  Stop it at the first $a$ with $C_a\ge B_D$.  On the event of the lemma the stopped value is at least $\exp[-\lambda\tau_D+(1-e^{-\lambda})B_D]$, so Markov's inequality with $e^{-\lambda}=\tau_D/B_D$ ($\lambda\to\infty$ if $\tau_D=0$) bounds its probability by $\exp[-B_D+\tau_D+\tau_D\ln(B_D/\tau_D)]\le\exp[-(B_D-\tau_D)^2/(2B_D)]$, which is at most $\epsilon_D$.
\end{proof}

\begin{theorem}[Index privacy]
\label{thm:alice_multiphoton}
Let Bob be honest, let (P) hold, and let the channel be noiseless (loss allowed).  With the click rules and the count of Step~\ref{QSPIR:initial},
\begin{equation}
    \SecretA\le\frac{0.78\,B_D}{\num}+\epsilon_D .
    \label{eq:epsA_multiphoton}
\end{equation}
\end{theorem}
\begin{proof}
Give Alice the photon number of every position, which she may choose adaptively.  Let $T$ be the first data-acquisition position at which $\sum_{b\le a}c_b\ge B_D$, and let $E$ be the event that Bob does not abort and $T$ exists.  By Lemma~\ref{lem:multiclick_count}, $\Pr[E]\le\epsilon_D$, and $E$ does not depend on $\pidx$.  Outside $E$, Bob either aborts, and Alice learns nothing about $\pidx$, or every retained position lies before $T$.  There, replace Bob's conclusive flag at every retained position by the $Y$ of Lemma~\ref{lem:photon_number}.  In this hybrid every retained position is conclusive with probability $\USDprob$, independently of everything Alice sees, so her final state does not depend on $\pidx$ (Proposition~\ref{prop:index_privacy_honest}).  Bob's messages differ from the hybrid only if the two flags differ at a retained position of his block, which is uniform and private.  By (P) Alice confirms only correct guesses, and Lemma~\ref{lem:photon_number}, applied given the past and linear in the photon-number distribution, bounds the expected number of differing retained positions before $T$ by $r\,\EX\sum_{a<T}c_a<rB_D$.  Hence $\SecretA\le\Pr[E]+2rB_D/\num$, and $2r=0.78$.
\end{proof}

For an ideal single-photon source, $\tau_D=0$ and $B_D=2\ln(1/\epsilon_D)$, so the bound is $0.002$ at $\num=10^{4}$ and $\epsilon_D=10^{-6}$.  With one photon per position, $\SecretA=0$ over any channel whose noise does not distinguish the four states and acts alike at every position (Proposition~\ref{prop:index_privacy_honest}).  The constant $r$ uses the multi-clicks that Alice's pulses cause under a wrong guess as well; counting only those under a correct guess would give $0.914$ in place of $0.390$.  The count costs an honest Bob nothing with a single-photon source.  With weak coherent pulses the honest source produces rare multi-photon detections, which Bob keeps, and the count must allow for them (Section~\ref{sec:decoy_state}).

\section{Proofs for the Decoy-State Analysis}
\label{app:decoy_proofs}

This appendix collects the proofs of Lemmas~\ref{lem:beff} and~\ref{lem:beff_plus}, Propositions~\ref{prop:four_subs} and~\ref{prop:db_privacy_decoy}, and Theorem~\ref{thm:decoy_success} of Section~\ref{sec:decoy_state}.

\begin{proof}[Proof of Lemma~\ref{lem:beff}]
The $N$-fold copies have overlap $\big(\bra{0}\ket{+}\big)^{N}=\nu^{N}$, so by the Ivanovi\'c--Dieks--Peres bound for two equiprobable pure states the optimal unambiguous success probability is $1-|\nu^{N}|$~\cite{Ivanovic_USD, Dieks_USD, Peres_USD}; the two alternatives are equiprobable because Alice prepares the four states uniformly. A coherent state of mean photon number $\mu_1$ transmitted through a channel of transmittance $\eta$ remains coherent with mean $\eta\mu_1$, so the number of photons reaching Bob is a $\mathrm{Poisson}$ distribution with mean $\eta\mu_1$, $p_N(\eta\mu_1)=e^{-\eta\mu_1}(\eta\mu_1)^N/N!$, with the Poisson weights $p_N(x):=e^{-x}x^N/N!$. Using $\sum_{N\ge0}p_N(\eta\mu_1)\nu^{N}=e^{-\eta\mu_1}\,e^{\eta\mu_1 \nu}=e^{-\eta\mu_1(1-\nu)}$ and $1-\nu=\USDprob$,
\begin{equation}
\begin{split}
    \pUSDeff&=\sum_{N\ge0}p_N(\eta\mu_1)\pUSDN{N}=1-\sum_{N\ge0}p_N(\eta\mu_1)\nu^{N}\\
    &=1-e^{-\eta\mu_1(1-\nu)}=1-e^{-\USDprob\,\eta\mu_1}.
\end{split}
\end{equation}
Conditioning on a click removes the vacuum weight $p_0(\eta\mu_1)=e^{-\eta\mu_1}$, giving $\pUSDsift=\pUSDeff/(1-p_0(\eta\mu_1))$; a first-order expansion in $\eta\mu_1$ yields the stated form. Finally, $\pUSDsift\ge \USDprob$ for all $\eta\mu_1>0$: writing $x=\eta\mu_1$, the claim is $g(x):=(1-\USDprob)-e^{-\USDprob x}+\USDprob\,e^{-x}\ge0$, and indeed $g(0)=0$ while $g'(x)=\USDprob\,(e^{-\USDprob x}-e^{-x})>0$ for $x>0$ since $\USDprob<1$.
\end{proof}

\begin{proof}[Proof of Lemma~\ref{lem:beff_plus}]
Among the detected signal events (per-pulse probability $\bar Q_{X,\mu_1}$, the \emph{true} gain), the vacuum and single-photon detected fractions are $e^{-\mu_1}Y_{X,0}/\bar Q_{X,\mu_1}$ and $\mu_1 e^{-\mu_1}Y_{X,1}/\bar Q_{X,\mu_1}$, so the multi-photon detected fraction is $\Delta=1-(e^{-\mu_1}Y_{X,0}+\mu_1 e^{-\mu_1}Y_{X,1})/\bar Q_{X,\mu_1}$. Three replacements turn this into the computable $\Delta^{+}$ of Eq.~\eqref{eq:Delta_plus}, each valid in the enlarging direction on a one-sided Hoeffding event that fails with probability at most $\epsilon_{\text{sec}}/\chi(3)$~\cite{Decoy_state_key_rate1}. First, the yields are replaced by their decoy lower bounds: in the three-intensity forms, $Y_{X,0}^{-}$ rests on two one-sided fluctuated gains ($Q_{X,\mu_3}$ from below, $Q_{X,\mu_2}$ from above) and $Y_{X,1}^{-}$ on three more ($Q_{X,\mu_2}$ from below, $Q_{X,\mu_3}$ and $Q_{X,\mu_1}$ from above; the opposite sides of $Q_{X,\mu_2},Q_{X,\mu_3}$ are distinct events). Second, the unobservable true gain in the denominator is replaced by the upward-fluctuated observed gain, valid on the event $Q_{X,\mu_1}+\Delta Q_{X,\mu_1}\ge\bar Q_{X,\mu_1}$ --- without this fluctuation an observed gain lying \emph{below} the true one would break the direction $\Delta^{+}\ge\Delta$. The denominator event coincides with the upward $Q_{X,\mu_1}$ event already used in $Y_{X,1}^{-}$, so the union bound runs over five distinct one-sided events, and $\Delta^{+}\ge\Delta$ holds except with probability $\le5\epsilon_{\text{sec}}/\chi(3)$; clipping to $[0,1]$ preserves the bound. For \eqref{eq:beff_plus}, let $f_N$ be the fraction of detected symbols carrying exactly $N$ photons and $f_{\ge2}:=\sum_{N\ge2}f_N$ the multi-photon fraction, so that the vacuum, single- and multi-photon fractions satisfy $f_0+f_1+f_{\ge2}=1$ and $f_{\ge2}\le\Delta^{+}$. A vacuum (dark-count) detection is conclusive with probability $0$, a single-photon detection with probability at most $\USDprob$ by the additional conditional hypothesis of Lemma~\ref{lem:beff_plus}, and a multi-photon detection with probability at most $1$. Hence
\begin{equation}
\begin{split}
    \pUSDsift&\le f_1 \USDprob+f_{\ge2}\\
    &\le\ (1-f_{\ge2})\,\USDprob+f_{\ge2}=\USDprob+f_{\ge2}(1-\USDprob)\\
    &\le\ \USDprob+\Delta^{+}(1-\USDprob),
\end{split}
\end{equation}
using $f_1\le1-f_{\ge2}$ and $\USDprob<1$.
\end{proof}

\begin{proof}[Proof of Proposition~\ref{prop:four_subs}]
(S1) is definitional: the downstream analysis of Section~\ref{subsec:reconciliation} enters the QKD layer only through the key length of the pads, which $R_{\text{decoy}}$ converts into a number of pulses; the margin $\hmarg$ is sized with the conclusive-bit error rate $\tilde{\ebit}$ of item (S4), not with a QKD error rate. Substituting the decoy-certified values for these scalars changes no inequality in the derivation.  For $r_{\text{decoy}}\le r$, read $r$ with $\ebit\mapsto e_{X,1}$ as in (S4).  $R_{\text{decoy}}$ credits only the single-photon and vacuum contributions~\cite{GLLP_security}; write $f_0,f_1,f_{\ge2}$ for the vacuum, single-photon and multi-photon fractions of the sifted detections, so that the overall error rate is $e_{\text{tol}}=f_0/2+f_1e_{X,1}+f_{\ge2}e_{\ge2}$, with $e_{\ge2}$ the error rate of the multi-photon detections.  Concavity of $H_2$ gives $H_2(e_{\text{tol}})\ge f_0+f_1H_2(e_{X,1})$, so for $\lambda\ge1$ the asymptotic rate obeys $r_{\text{decoy}}=f_0+f_1[1-H_2(e_p)]-\lambda H_2(e_{\text{tol}})\le f_1[1-H_2(e_{X,1})-H_2(e_p)]\le r$ whenever $r\ge0$. (S2) is Proposition~\ref{prop:cert_rate_wcp} and Proposition~\ref{prop:db_privacy_decoy}.
\end{proof}

\begin{proof}[Proof of Proposition~\ref{prop:db_privacy_decoy}]
On the event of Proposition~\ref{prop:cert_rate_wcp}, hand Bob the flags of Lemma~\ref{lem:erasure_domination}, a retained multi-photon position being charged as known.  The flags are independent given his record (Lemma~\ref{lem:decoy_sample}), their mean rate on the retained positions is at most $\pknown^{\mathrm{WCP}}$, and the blocks are a uniform partition of these positions. So Lemma~\ref{lem:overlap_mgf} applies, and the proof of Eq.~\eqref{eq:EW_common} in Appendix~\ref{app:simulator_proof} goes through verbatim; its right side increases with the known-rate. Theorem~\ref{thm:success_protocol}(ii) adds $\epsilon_{\mathrm{rnd}}$, and the two certificates add $\epsilon_{\mathrm{PE}}+\epsilon_\text{cert}$, as in Theorem~\ref{thm:cert_privacy}.
\end{proof}

\begin{proof}[Proof of Theorem~\ref{thm:decoy_success}]
By Proposition~\ref{prop:four_subs} the substitutions $r\mapsto r_{\text{decoy}}$, $\pknown\mapsto\pknown^{\mathrm{WCP}}$, $(e_b, e_p)\mapsto(e_{X,1}, e_p)$ enter Theorem~\ref{thm:success_protocol} only through the known-rate in the size rule of Eq.~\eqref{eq:size_rule}, which becomes Eq.~\eqref{eq:overhead_decoy}, through $\tilde{e}_b$ in Eq.~\eqref{eq:syndrome_length}, and through the key cost $\ell/R_{\text{decoy}}$ of item (S1), which replaces Eq.~\eqref{eq:nums_sifting}.  The conversion of the key length to a number of pulses has its only pole at $R_{\text{decoy}}=0$, and $R_{\text{decoy}}>0$ exactly when $e_p<e_{\text{thr}}^{\text{decoy}}$, so a finite cost exists throughout that range. Eve-secrecy holds because $r_{\text{decoy}}$ counts only decoy-certified single-photon and vacuum contributions, so the pads are secret against PNS.  Database privacy is Proposition~\ref{prop:db_privacy_decoy} together with Lemma~\ref{lem:seeded_parity}.  For an Alice who uses the prescribed source, index privacy is unaffected because all data pulses carry the single intensity $\mu_1$, so every kept position is conclusive with the same probability (Proposition~\ref{prop:index_privacy_honest}), and the estimates entering $\pknown^{\mathrm{WCP}},r_{\text{decoy}},e_p$ use only sacrificed test and QKD-mode slots, which never enter a block.
\end{proof}

\subsection{Intensity-Resolved Measurement Test}
\label{app:decoy_receiver_test}
Here we state in full the test summarized in Section~\ref{sec:decoy_receiver_test}.

For the WCP variant, use a fixed emitted batch. Bob fixes the received or lost list, QKD or data modes, and guessed sets before Alice samples the test slots or reveals preparation and intensity labels.

After the storage-enforcing wait, Alice samples uniformly from emitted slots. Bob opens the six declaration types of Table~\ref{tab:usd_test} on sampled data slots before learning their labels. All sampled slots are discarded, including QKD-mode slots.

For each intensity $\mu$ and preparation $s$, Alice records those six outcomes and an exclusion outcome $\perp$ for losses or non-data modes. Frequencies are normalized by the expected emitted test count $p_\mu n_\text{t}/4$ of that type.

Alice combines the intensity groups through
\begin{equation}
    Q_{\mu,s,d}=\sum_{N\ge0}e^{-\mu}\frac{\mu^N}{N!}Y_{N,s,d},
    \qquad d\in\mathcal{D}_\text{dec}\cup\{\perp\},
    \label{eq:decoy_declaration_gain}
\end{equation}
where $\mathcal{D}_\text{dec}$ contains the six declarations and $N$ is the emitted photon number. At fixed $N,s$, the source state is independent of its hidden intensity label~\cite{Ma_Qi_Zhao_Lo_decoy}.

Thus decoys bound the single-photon declaration statistics. Alice compares that inferred component with Table~\ref{tab:usd_test}; the different WCP intensity groups need not reproduce the same table directly.

Under assumptions (A1)--(A3) of Appendix~\ref{app:usd_certification}, read pulse by pulse, Appendix~\ref{app:decoy_receiver_certificate} gives intervals $L_{s,d}\le\tilde Y_{1,s,d}\le U_{s,d}$ for the single-photon yields \emph{of the test sample} (Lemma~\ref{lem:decoy_sample}), with total failure at most $\epsilon_{\mathrm{PE}}$.  A fixed instrument is not needed.

Let $T_{s,d}$ be the ideal entry of Table~\ref{tab:usd_test}, including the guessed-set factor $1/2$. For a reference yield $Y_\text{ref}>0$, set
\begin{equation}
\begin{split}
    \delta_1=\frac{1}{Y_\text{ref}}\max_{\substack{s\\d\in\mathcal{D}_\text{dec}}}
      \max\big\{&|L_{s,d}-Y_\text{ref}T_{s,d}|,\\
                 &|U_{s,d}-Y_\text{ref}T_{s,d}|\big\}.
\end{split}
\label{eq:decoy_single_deviation}
\end{equation}
The choice of $Y_\text{ref}$ may be optimized within this certificate. Unlike the raw tolerance $\delta_{\mathrm{test}}$, $\delta_1$ includes both decoy and sampling uncertainty.

Let $R$ be the event that a slot is reported in the data mode with a correct guessed set, before the final uniform subsampling, and let $X$ be the bit value of such a slot. Let $V$ contain Bob's full private single-signal record, guessed set, confirmation, and intensity label, before the announcements of Step~\ref{QSPIR:final_step}.

Lemma~\ref{lem:guessing_bound} of Appendix~\ref{app:usd_certification}, applied as in Appendix~\ref{app:decoy_receiver_certificate}, proves $p_{\mathrm{guess}}(X\mid V,N=1,R)\le G(\delta_1)$, including selective loss and private refinements of declared outcomes. Here $G$ is given in Eq.~\eqref{eq:decoy_G_bound}.

The photon fraction must refer to this same retained event. With $g(s)$ the prepared set and $g(d)$ the declared set, define
\begin{equation}
\begin{split}
    r_N&=\frac14\sum_s\sum_{\substack{d\in\mathcal{D}_\text{dec}\\g(d)=g(s)}}Y_{N,s,d},\\
    Q_R(\mu)&=\sum_{N\ge0}e^{-\mu}\frac{\mu^N}{N!}r_N,\\
    \Delta_R^+&=\operatorname{clip}_{[0,1]}
      \left(1-\frac{e^{-\mu_1}r_0^-+\mu_1e^{-\mu_1}r_1^-}
                        {Q_R^+(\mu_1)}\right).
\end{split}
\label{eq:decoy_retained_tag}
\end{equation}
The bounds follow from the same declaration intervals. Abort on inconsistent intervals, a zero denominator, or an insufficient retained-yield lower bound.

Vacuum bits remain unbiased in the ideal source model. Crediting every retained multiphoton bit as known gives the replacement information certificate
\begin{equation}
\begin{split}
    p_{\mathrm{guess}}(X\mid V,R)
       &\le g_R^{\mathrm{cert}},\\
    g_R^{\mathrm{cert}}
       &:=(1-\Delta_R^+)G(\delta_1)+\Delta_R^+ .
\end{split}
\label{eq:decoy_retained_guess}
\end{equation}
Accept only if $g_R^{\mathrm{cert}}$ meets a pre-agreed target, namely that the size rule of Eq.~\eqref{eq:size_rule} holds at the certified rate of Eq.~\eqref{eq:pk_wcp}. Alice then confirms correct guesses on untested signal-intensity data and uniformly retains $\numd\num$ positions, aborting if too few remain.

This yields $H_{\min}(X\mid V,R)\ge-\log_2g_R^{\mathrm{cert}}$ before the announcements of Step~\ref{QSPIR:final_step}~\cite{Koenig_operational_min_entropy}.

\subsection{Receiver-information certificate}
\label{app:decoy_receiver_certificate}
Read (A2) pulse by pulse: Bob measures pulse $i$ by a POVM $\{\Pi^{(i)}_{y_i|y_{<i}}\}$ on all its photons, chosen from his earlier outcomes.  The decoy equations survive such a receiver, because they hold for the test sample itself.

\begin{lemma}[Decoy structure on the sample]
\label{lem:decoy_sample}
Let slot $i$ carry the intensity $\mu(i)$, the state $s_i$ and $N_i$ photons, and put $\bar p_N=\sum_\mu p_\mu p_N(\mu)$ and $P_i(s,N)=\Pr[s_i=s,N_i=N\mid y]$.  Under (A2) and (A3), the slots are independent given Bob's record $y$, and $\Pr[\mu(i)=\mu\mid N_i=N,y]=p_\mu p_N(\mu)/\bar p_N$.  Hence, for every test set $\mathcal{T}$ and declaration $d$,
\begin{equation}
\begin{split}
    \tilde Q_{\mu,s,d}&:=\frac{4}{p_\mu n_\text{t}}\sum_{\substack{i\in\mathcal{T}\\D_i=d}}\Pr[\mu(i)=\mu,s_i=s\mid y]\\
    &\;=\sum_{N\ge0}p_N(\mu)\,\tilde Y_{N,s,d},\\
    \tilde Y_{N,s,d}&:=\frac{4}{n_\text{t}\bar p_N}\sum_{\substack{i\in\mathcal{T}\\D_i=d}}P_i(s,N)\ \ge0,
\end{split}
\label{eq:decoy_sample_gain}
\end{equation}
and $\tilde Y_{1,s,d}=\operatorname{Tr}[M_d\rho_s]$ for positive operators $M_d$ on the one-photon space.
\end{lemma}
\begin{proof}
A phase-randomized pulse depends on its intensity only through its photon number~\cite{Ma_Qi_Zhao_Lo_decoy}, so
\begin{equation*}
\begin{split}
    \Pr[(s_i,N_i,\mu(i))_i,y]&=\prod_i\tfrac14p_{\mu(i)}p_{N_i}(\mu(i))\\
    &\quad\times\operatorname{Tr}\big[\Pi^{(i)}_{y_i|y_{<i}}\rho_{s_i}^{\otimes N_i}\big],
\end{split}
\end{equation*}
a product of one-slot factors; summing over $N$ gives Eq.~\eqref{eq:decoy_sample_gain}.  For $N=1$, $P_i(s,1)=\tfrac14\bar p_1\operatorname{Tr}[\Pi_1^{(i)}\rho_s]/Z_i$, with $\Pi_1^{(i)}\ge0$ the one-photon part of the effect and $Z_i$ the probability of $y_i$ given $y_{<i}$, so $M_d=\tfrac{1}{n_\text{t}}\sum_{i\in\mathcal{T},D_i=d}\Pi_1^{(i)}/Z_i$.
\end{proof}

Note that Eq.~\eqref{eq:decoy_sample_gain} is Eq.~\eqref{eq:decoy_declaration_gain} for the sample.  Given $y$, the number $X$ of test slots of type $(\mu,s)$ with declaration $d$ is therefore a sum of independent indicators with mean $\tfrac14p_\mu n_\text{t}\tilde Q_{\mu,s,d}$.  These cells are rare events, for which Hoeffding's inequality is wasteful; the Chernoff bounds, which also hold for a sample drawn without replacement~\cite{Hoeffding_1963}, give
\begin{equation}
    X-\sqrt{2X\Lambda}\;\le\;\EX X\;\le\;X+\Lambda+\sqrt{2X\Lambda+\Lambda^{2}},
    \label{eq:decoy_declaration_intervals}
\end{equation}
each side except with probability $e^{-\Lambda}$.  With $\Lambda=\ln(2K/\epsilon_{\mathrm{PE}})$ for the $K=3\cdot4\cdot7=84$ cells, they give bounds $Q_{\mu,s,d}^{\pm}$ that hold together except with probability $\epsilon_{\mathrm{PE}}$.  For each $s,d$ (indices suppressed, $\mu_3=0$) put
\begin{equation}
\begin{split}
 J&=e^{\mu_2}Q_{\mu_2}^- -\frac{\mu_2^2}{\mu_1^2}e^{\mu_1}Q_{\mu_1}^+-\Big(1-\frac{\mu_2^2}{\mu_1^2}\Big)Q_0^+,\\
 L&=\max\{0,\mu_1J/(\mu_1\mu_2-\mu_2^2)\},\\
 U&=(e^{\mu_2}Q_{\mu_2}^+-Q_0^-)/\mu_2.
\end{split}
\label{eq:decoy_outcome_Y1_bounds}
\end{equation}
Then $L\le\tilde Y_1\le U$ for every such receiver: the lower expression cancels $\tilde Y_0,\tilde Y_2$ and has nonpositive coefficients for $\tilde Y_{N\ge3}$, and the upper one follows from $e^{\mu_2}\tilde Q_{\mu_2}-\tilde Q_0\ge\mu_2\tilde Y_1$.

On this event, $|\operatorname{Tr}(\rho_sM_d)-Y_\text{ref}T_{s,d}|\le Y_\text{ref}\delta_1$ in all $24$ cells by Eq.~\eqref{eq:decoy_single_deviation}, so Lemma~\ref{lem:guessing_bound} applies to the operators $M_d/Y_\text{ref}$.  Rescaling every $M_d$ by the same factor leaves $p_{\mathrm{guess}}$ unchanged, so every private refinement guesses the bit of a retained single-photon position with probability at most $G(\delta_1)$.  For retained photon fractions $f_0,f_1,f_{\ge2}$, vacuum is unbiased and $f_{\ge2}\le\Delta_R^+$; since $G\ge1/2$, $p_{\mathrm{guess}}\le f_0/2+f_1G(\delta_1)+f_{\ge2}\le(1-\Delta_R^+)G(\delta_1)+\Delta_R^+$.  This proves Eq.~\eqref{eq:decoy_retained_guess}, except on an event of bad intervals that has probability at most $\epsilon_{\mathrm{PE}}$.

We now turn this certificate into a known-rate, as Proposition~\ref{prop:cert_rate} does for single photons.  Let $N_\text{c}$ be the number of untested slots, let $\bar r^{-}=p_{\mu_1}Q_R^{-}(\mu_1)$ be the lower bound on the retained fraction of the test slots that the twelve correct-set signal cells give, let $\Lambda_\text{c}=\ln(6/\epsilon_\text{cert})$, and set
\begin{equation}
\begin{split}
    \pknown^{\mathrm{WCP}}&:=\frac{q_\text{T}+\varsigma}{1-\varsigma},\\
    q_\text{T}&:=(1-\Delta_R^{+})\min\{1,B(\delta_1)\}+\Delta_R^{+},\\
    \varsigma&:=\frac{\Lambda_\text{c}+\sqrt{2n_\text{t}\bar r^{-}\Lambda_\text{c}+\Lambda_\text{c}^{2}}}{n_\text{t}\bar r^{-}}+2\,\frac{\Lambda_\text{c}+\sqrt{2N_\text{c}\bar r^{-}\Lambda_\text{c}}}{N_\text{c}\bar r^{-}}.
\end{split}
\label{eq:pk_wcp_full}
\end{equation}

\begin{proposition}[Certified known-rate under WCP]
\label{prop:cert_rate_wcp}
Assume (A1)--(A3), read pulse by pulse.  Let $\bar q$ be the mean flag rate of Lemma~\ref{lem:erasure_domination} over the untested signal slots that Bob reports in the data mode with a correct set guess.  Then $\Pr[\text{test passes and }\bar q>\pknown^{\mathrm{WCP}}]\le\epsilon_{\mathrm{PE}}+\epsilon_\text{cert}$.
\end{proposition}
\begin{proof}
The proof mirrors that of Proposition~\ref{prop:cert_rate}.  For a slot reported in the data mode put $r_i=\Pr[\mu(i)=\mu_1,s_i\in S_{G_i}\mid y]$ and $\alpha_i=|\Pr[\mu(i)=\mu_1,s_i=u_0\mid y]-\Pr[\mu(i)=\mu_1,s_i=u_1\mid y]|$, with $u_0,u_1$ the two states of $S_{G_i}$, and $r_i=\alpha_i=0$ for every other slot, so that a retained position has $q_i=\alpha_i/r_i$.  Split both numbers by photon number.  The vacuum state does not depend on $s$, so $\alpha_i\le\alpha_i^{(1)}+\sum_{N\ge2}r_i^{(N)}$.  On the interval event, Lemma~\ref{lem:guessing_bound} applied to the operators $M_d/Y_\text{ref}$ of Lemma~\ref{lem:decoy_sample} gives $\sum_{i\in\mathcal{T}}\alpha_i^{(1)}\le\min\{1,B(\delta_1)\}\sum_{i\in\mathcal{T}}r_i^{(1)}$, as in Lemma~\ref{lem:cert_sample}, and Eq.~\eqref{eq:decoy_retained_tag}, read for the sample yields, caps the multi-photon part of $\sum_{i\in\mathcal{T}}r_i$ at a fraction $\Delta_R^{+}$.  Hence $\sum_{i\in\mathcal{T}}\alpha_i\le q_\text{T}\sum_{i\in\mathcal{T}}r_i$, with $\sum_{i\in\mathcal{T}}r_i\ge n_\text{t}\bar r^{-}$.

The rest is steps (b)--(d) of the proof of Proposition~\ref{prop:cert_rate}, with Eq.~\eqref{eq:decoy_declaration_intervals} in place of Hoeffding's inequality, since the retained fraction is only of order $\eta\mu_1$.  The numbers $\alpha_i,r_i$ are fixed before $\mathcal{T}$ is drawn, so the sums of $\alpha_i$ and of $r_i$ over $\mathcal{T}$ and over its complement, and the sums of $I_iq_i$ and of $I_i$ over the complement, where the retained indicators $I_i$ are independent $\mathrm{Bern}(r_i)$, obey Eq.~\eqref{eq:decoy_declaration_intervals}.  These six events fail with probability at most $\epsilon_\text{cert}$ in total.  Chaining them gives $\sum I_iq_i\le\hat r(q_\text{T}+\varsigma)$ and $\sum I_i\ge\hat r(1-\varsigma)$ with $\hat r=\tfrac{N_\text{c}}{n_\text{t}}\sum_{i\in\mathcal{T}}r_i$, the worst case being $\sum_{i\in\mathcal{T}}r_i=n_\text{t}\bar r^{-}$.
\end{proof}

What does it cost?  For the honest receiver at $\eta=0.1$, with $\mu_2=\mu_1/10$, intensity probabilities $(0.5,0.4,0.1)$, $\epsilon_{\mathrm{PE}}=\epsilon_\text{cert}=10^{-6}$ and the test counts at their expected values, a test of $n_\text{t}=10^{11}$ pulses certifies $\pknown^{\mathrm{WCP}}=0.42$ and $0.45$ at $\mu_1=0.05$ and $0.1$, against $\pUSDplus=0.33$ and $0.36$; Fig.~\ref{fig:DS4_known_rate} shows the resulting blocks.  The slack has a floor, because the weak decoy resolves $\tilde Y_1$ only to a relative error of order $\mu_2$, and $B$ turns this error into its square root: $q_\text{T}\ge0.37$, $0.42$ and $0.51$ at $\mu_1=0.05$, $0.1$ and $0.2$, the last at the limit of Eq.~\eqref{eq:mu1_constraint}.

\section{Numerical Validation of the Decoy Layer}
\label{app:decoy_numerics_validation}

This appendix describes the Monte-Carlo simulation behind Section~\ref{sec:decoy_numerics}.  It checks the rate formulas of Lemmas~\ref{lem:beff} and~\ref{lem:beff_plus} and the binomial law of the per-block conclusive count, which fixes the honest abort probability.

Each detected data pulse delivers a zero-truncated $\mathrm{Poisson}(\eta\mu_1)$ number of photons to Bob, who applies the optimal $N$-copy unambiguous measurement of Eq.~\eqref{eq:bN}. Alongside this \emph{honest} receiver, which keeps every click as in Step~\ref{QSPIR:initial}, we simulate the strongest adversary charged by Lemma~\ref{lem:beff_plus} --- photon-number resolution, lossless collection, and $\pUSDN{N}\to1$ on every multi-photon emission. The decoy-certified quantities ($Y_{X,0}^{-}$, $Y_{X,1}^{-}$, $e_{X,1}$, $e_p$, $\Delta^{+}$, $\pUSDplus$) are computed from the simulated gains through the standard three-intensity decoy bounds~\cite{Ma_Qi_Zhao_Lo_decoy, Decoy_state_key_rate1}, with Hoeffding fluctuations at the stated sample sizes and the upward-fluctuated gain in the denominator of Eq.~\eqref{eq:Delta_plus}; the misalignment error is $e_{\text{det}}=0.01$ (defined in Section~\ref{sec:decoy_numerics}) and the $\mu_1$ sweep uses $\mu_2=\min(0.1,\mu_1/2)$ and $\mu_3=0$.  These figures check Lemmas~\ref{lem:beff} and~\ref{lem:beff_plus}.  The rate that enters the size rule is the certified rate $\pknown^{\mathrm{WCP}}$ of Appendix~\ref{app:decoy_receiver_certificate}, which also carries $B(\delta_1)$, depends on $\mu_2$, and lies higher: at least $0.37$ at $\mu_1=0.05$ for $\mu_2=\mu_1/10$.

Figure~\ref{fig:DS1_rates} confirms the rate formulas of Lemmas~\ref{lem:beff} and~\ref{lem:beff_plus}. The honest conclusive fraction follows the closed form of Lemma~\ref{lem:beff} to Monte-Carlo precision ($\sim10^{-3}$), staying within a few percent of the single-photon ceiling $\USDprob$ for all $\mu_1\le0.6$ at $\eta=0.1$. The adversary's conclusive fraction reaches $\pUSDplus=\USDprob+\Delta^{+}(1-\USDprob)$: within the simulation model, the ceiling of Lemma~\ref{lem:beff_plus} is attained by the worst-case receiver it is designed for, so no analysis of this form can certify a lower known-rate. Panel~(b) verifies Eq.~\eqref{eq:Delta_plus_scaling}: the multi-photon detected fraction \emph{grows} as the channel gets lossier and saturates at $1-e^{-\mu_1}$, confirming that, for exact statistics, the gap $\pUSDplus-\pUSDsift$ is $O(\mu_1)$ for every $\eta$. The fully certified $\Delta^{+}$ (computed through the yield bounds rather than from the true photon-number distribution) lies up to two percentage points above the idealized value even with infinite decoy statistics --- the residual looseness of the three-intensity estimator --- and finite statistics at $s_X=10^{8}$ add another ${\sim}0.7$ points to $\pUSDplus$.

Figure~\ref{fig:DS2_dist} shows that the per-block conclusive count is binomial for both receivers --- $\mathrm{Bin}(\numd,\pUSDsift)$ for the honest Bob and $\mathrm{Bin}(\numd,\pUSDplus)$ for the certified adversary --- with measured variance ratios $1.000\pm0.004$ over $2\times10^{5}$ blocks per case.  This fixes the honest abort probability; the privacy bound itself rests on the overlap law of Lemma~\ref{lem:overlap_mgf}, tested for single photons in Section~\ref{sec:sim_pair_attack}.

\begin{figure*}
    \centering
    \includegraphics[width=\textwidth]{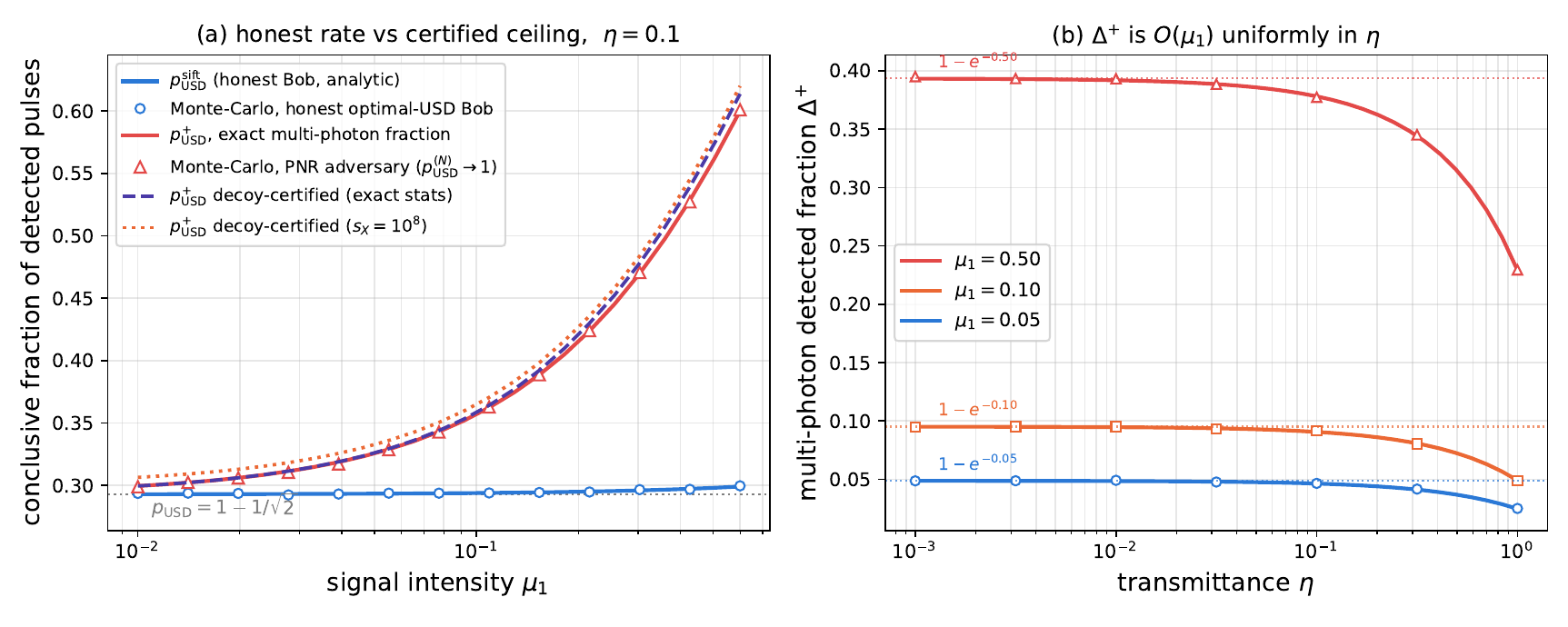}
    \caption{Photon-level Monte-Carlo validation of the effective USD rates. (a)~Conclusive fraction of detected data pulses vs signal intensity $\mu_1$ at $\eta=0.1$: honest optimal-USD Bob (circles) against the closed form of Lemma~\ref{lem:beff} (solid blue), and the photon-number-resolving adversary of Lemma~\ref{lem:beff_plus} (triangles), which saturates $\pUSDplus$ computed from the exact multi-photon fraction (solid red). The decoy-certified $\pUSDplus$ with exact statistics (dashed) and with $s_X=10^{8}$ $X$-sifted detections (dotted) lie slightly higher, showing the residual looseness of the three-intensity estimator. (b)~Multi-photon detected fraction $\Delta^{+}$ vs transmittance $\eta$ [Eq.~\eqref{eq:Delta_plus_scaling}]: markers are Monte-Carlo, curves the pure-loss prediction; $\Delta^{+}$ saturates at $1-e^{-\mu_1}$ (dotted asymptotes) as $\eta\to0$ --- the gap $\pUSDplus-\pUSDsift$ does not vanish with loss.}
    \label{fig:DS1_rates}
\end{figure*}

\begin{figure*}
    \centering
    \includegraphics[width=\textwidth]{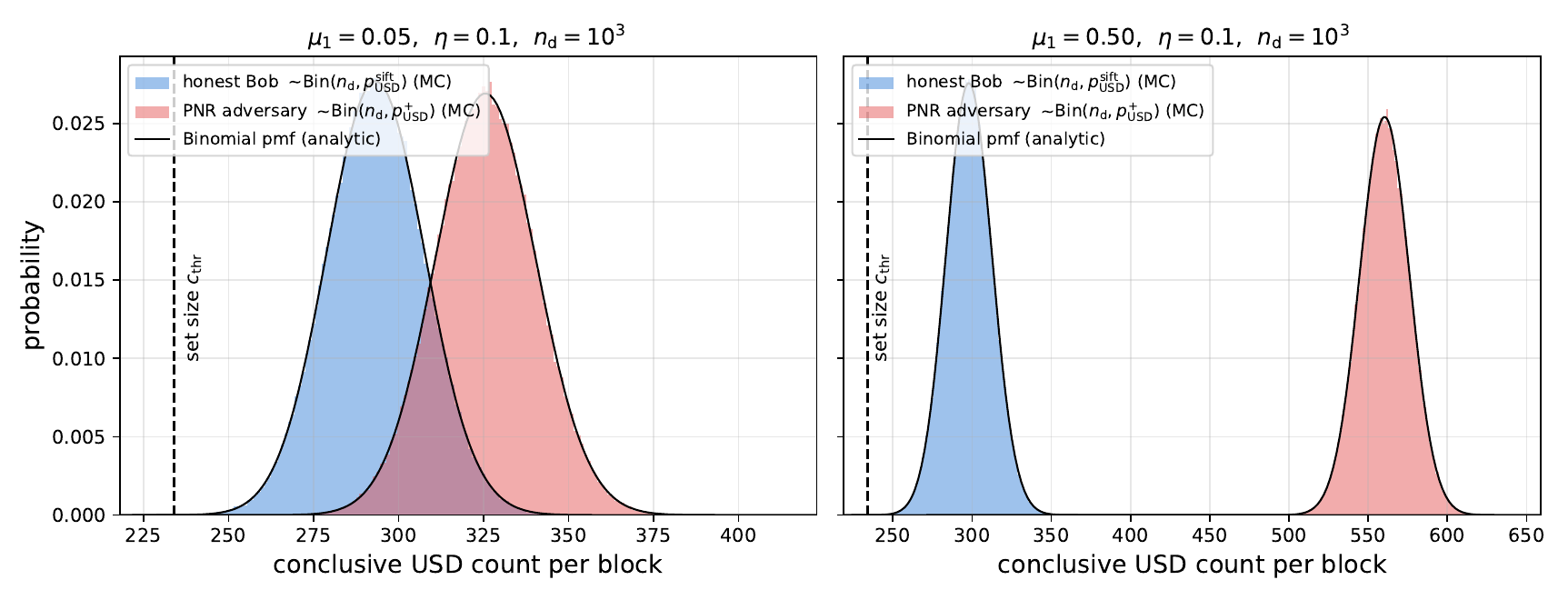}
    \caption{Per-block conclusive USD counts ($2\times10^{5}$ Monte-Carlo blocks of $\numd=10^{3}$ detected symbols, $\eta=0.1$) for the honest receiver (blue) and the certified adversary (red), with the binomial laws $\mathrm{Bin}(\numd,\pUSDsift)$ and $\mathrm{Bin}(\numd,\pUSDplus)$ (black curves). The dashed line is the set size $\cthr$ of Step~\ref{QSPIR:code_sizing} for $\epsilon_\text{abort}=10^{-3}$. The honest counts lie above it at both intensities, so correctness is unaffected. What changes is the tagging ceiling $\pUSDplus$ of Lemma~\ref{lem:beff_plus}, which moves from $\approx0.33$ at $\mu_1=0.05$ (left) to $0.56$ at $\mu_1=0.5$ (right), beyond the limit of Eq.~\eqref{eq:mu1_constraint}; the certified rate of Appendix~\ref{app:decoy_receiver_certificate} is higher still.}
    \label{fig:DS2_dist}
\end{figure*}

\clearpage
\bibliography{SPIR_Single_Database}%

\end{document}